\documentclass[letterpaper,12pt]{article}
\usepackage{amsfonts}
\usepackage{amsmath, amsthm, amssymb, astron, bbm, epsfig, lscape}
\usepackage{xcolor}

\numberwithin{equation}{section}

\newtheorem{assumption}{Assumption}

\newtheorem{theorem}{Theorem}[section]

\newtheorem{lemma}[theorem]{Lemma}

\newcommand{\argmin}{\operatorname*{argmin}}
\newcommand{\plim}{\operatorname*{plim}}

\def\ft#1#2{{\textstyle {\frac{#1}{#2}} }}

\definecolor{dgreen}{rgb}{0,0.6,0}

\newcommand{\smallW}{W}
\newcommand{\bigW}{{\cal W}}
\newcommand{\bigWest}{\widehat {\cal W}}

\begin{document}

\title{\bf Estimation of Random Coefficients Logit Demand Models with Interactive Fixed Effects\thanks{This is the last working paper version of the paper published in \textit{Journal of Econometrics} \textbf{206}(2), 613--644, 2018; doi:10.1016/j.jeconom.2018.06.016.   We thank participants in presentations at Georgetown, Johns Hopkins, Ohio State, Penn State, Rice, Texas A\&M, UC Davis, UC Irvine, UCLA, Chicago Booth, Michigan, UPenn, Wisconsin, Southampton, the 2009 California Econometrics Conference and the 2010 Econometric Society World Congress for helpful comments.
    Chris Hansen, Han Hong, Sung Jae Jun, Jinyong Hahn, and Rosa Matzkin provided very helpful discussions.
    Moon acknowledges the NSF for financial support via SES 0920903.
    Weidner  acknowledges support from the Economic and Social Research Council through the ESRC Centre for Microdata Methods and Practice grant RES-589-28-0001, and from the European Research Council grant ERC-2014-CoG-646917-ROMIA.}
    }

\author{\setcounter{footnote}{2}
Hyungsik Roger Moon\footnote{
Department of Economics, University of Southern California,
KAP 300, Los Angeles, CA 90089-0253.
Email: {\tt moonr@usc.edu}}
\footnote{
Department of Economics, Yonsei University,
Seoul, Korea. }
\and Matthew Shum\footnote{
Division of Humanities and Social Sciences,
California Institute of Technology,
MC 228-77, Pasadena, CA 91125.
Email: {\tt mshum@caltech.edu}.
}
\and Martin Weidner\footnote{
 Department of Economics,
 University College London,
 Gower Street,
 London WC1E 6BT, U.K.,
 and CeMMAP.
 Email: {\tt m.weidner@ucl.ac.uk}.
}}

\date{\vspace{-0.2cm} February 2017}

\maketitle

\vspace{-0.9cm}
\abstract{
\begin{center}
\vspace{3mm}
\begin{minipage}{0.75\textwidth}
\footnotesize
We extend the Berry, Levinsohn and Pakes (BLP, 1995)
random coefficients discrete-choice demand model, which
  underlies much recent empirical work in IO. We add
 interactive fixed effects in the form of a factor structure on the
 unobserved product characteristics.   The interactive fixed effects can be
 arbitrarily correlated with the observed product characteristics
 (including price), which accommodates endogeneity and, at the same time,
captures strong persistence in market shares across products and
markets.  We propose a two-step least squares-minimum distance (LS-MD)
procedure to calculate the estimator.
Our estimator is easy to compute, and Monte Carlo
simulations show that it performs well. We consider an empirical illustration
to US automobile demand.
\end{minipage}
\vspace{3mm}
\end{center}
}
\vspace{1mm}


\section{Introduction}

The Berry, Levinsohn and Pakes \cite*{BerryLevinsohnPakes1995}
(hereafter BLP)  demand model, based on the random coefficients
logit multinomial choice model, has become the workhorse of demand
modeling in empirical industrial organization and antitrust analysis.
An important virtue of this model is that it parsimoniously and
flexibly captures substitution possibilities between the products in a
market.   At the same time, the nested simulated GMM procedure
proposed by BLP accommodates possible endogeneity of the observed
product-specific regressors, notably price.   This model and
estimation approach has proven very popular (e.g. Nevo \cite*{nevo1}, Petrin \cite*{petrin};
surveyed in
Ackerberg et. al. \cite*{ack_pakes}).

Taking a cue from recent
developments in panel data econometrics (e.g. Bai and Ng
\cite*{BaiNg2006}, Bai \cite*{Bai2009}, and Moon and Weidner
\cite*{MoonWeidner2015,MoonWeidner2015b}), we extend the standard BLP demand model by adding
interactive fixed effects to the unobserved product characteristic, which
is the main ``structural error'' in the BLP model. This interactive fixed effect
specification combines market (or time) specific fixed effects with
product specific fixed effects in a multiplicative form, which is often
referred to as a factor structure.

Our factor-based approach extends the baseline BLP model in
  two ways.   First, we offer an alternative to the usual moment-based GMM
  approach.   The interactive fixed effects ``soak up'' some important
  channels of endogeneity, which may obviate the need for instrumental
  variables of endogenous regressors such as price.
   This is important as such instruments
   may not be easy to identify in practice.
Moreover, our analysis of the BLP model with interactive fixed effects illustrates that the problem of finding instruments for price (which arises in any typical demand model) is distinct from the problem of underidentification of some model parameters (such as the variance parameters for the random components), which arises from the specific nonlinearities in the BLP random coefficients demand model.   In our setting, the fixed effects may obviate the need for instruments to control for price endogeneity but, as we will point out, we still need to impose additional moment conditions in order to identify these nonlinear parameters.
Second, even
  if endogeneity persists in the presence of the interactive fixed effects,
the instruments only need to be exogenous with respect to the residual part
of the unobserved product characteristics, which is not explained by the
interactive fixed effect.
This may expand the set of variables which may be used as instruments.

 To our knowledge, the current paper presents the
 first application of some recent developments in the econometrics of
long panels (with product and market fixed effects) to the workhorse demand model in empirical IO.
 Relative to the existing panel factor
literature (for instance, Bai \cite*{Bai2009}, and Moon and Weidner
\cite*{MoonWeidner2015,MoonWeidner2015b})  that assume a linear regression with exogenous regressors, 
the nonlinear model that we consider here poses both identification and
estimation challenges. Namely, the
usual principal components approach for linear factor models with exogenous regressors is inadequate
due to the nonlinearity of the model and the potentially endogenous
regressors. At the same time, the conventional GMM approach of BLP cannot be used for identification and estimation due to the presence of
the interactive fixed effects.

We propose an alternative  identification and estimation scheme which we call the {\em Least
Squares-Minimum Distance (LS-MD)} method.
It consists of two steps.   The first step is a least squares regression of the mean utility on the included product-market specific regressors, factors, and the instrumental variables. The second step minimizes the norm of the least squares coefficient of the instrumental variables in the first step.
This estimation approach is similar to the two stage estimation method for a class of instrumental
	quantile regressions in Chernozhukov and Hansen \cite*{ChernozhukovHansen2006}.
We show that under regularity conditions that are comparable to the standard GMM problem, the parameter of interest is point identified and its estimator is consistent. We also derive the limit distribution under asymptotic
sequences where both the number of products and the number of markets
converge to infinity. In practice, the estimator is simple and straightforward to compute. Monte Carlo simulations demonstrate its good
small-sample properties.

Our work complements some recent
papers in which alternative estimation approaches and extensions of
the standard random coefficients logit model
have been proposed, including Villas-Boas and Winer
\cite*{Villasboas1999}, Knittel and Metaxoglou
\cite*{KnittelMetaxoglou2014}, Dube, Fox and Su
\cite*{DubeFoxSu2012}, Harding and Hausman \cite*{HardingHausman2007},
Bajari, Fox, Kim and Ryan \cite*{bajari_fox_kim_ryan},
and Gandhi, Kim and Petrin~\cite*{GandhiKimPetrin2010}.

We illustrate our estimator on a dataset of market shares for automobiles,
inspired by the exercise in BLP.   This application illustrates that our
estimator is easy to compute in practice.
 Significantly, we find that, once factors are included
  in the specification, the estimation results under the assumption of
  exogenous and endogenous price are quite similar, suggesting that the factors are
  indeed capturing much of the unobservable product and time effects
  leading to price endogeneity.

The paper is organized as follows. Section~2 introduces the
  model. In Section~3 we discuss how to identify the model when valid instruments are available.
  In Section~4  we introduce the LS-MD estimation method. Consistency and asymptotic normality are discussed in
  Section~5. Section 6 contains Monte Carlo simulation results, and Section~7
  discusses the empirical example. 
   Section~8 concludes. In the appendix we list the assumptions for the
  asymptotic analysis and provide technical derivations and proofs of
  results in the main text.

\subsection*{Notation}

We write $A'$ for the transpose of a matrix or vector $A$.
For column vectors $v$ the Euclidean norm is
defined by $\| v \| = \sqrt{v^{\prime}v}$ . For the $n$-th largest
eigenvalues (counting multiple eigenvalues multiple times) of a symmetric
matrix $B$ we write $\mu_n(B)$. For an $m\times n$ matrix $A$
the Frobenius norm is $\| A \|_{F} = \sqrt{{\rm Tr}(AA^{\prime})}$,
and the spectral norm is $\| A \| = \max_{0 \neq v \in
\mathbb{R}^n} \, \frac{ \| A v \|} {\| v\|}$, or equivalently $\| A \| =
\sqrt{ \mu_1(A^{\prime}A) }$. Furthermore, we use $P_A = A
(A^{\prime}A)^\dagger A'$ and $M_A = \mathbbm{1}_m - A (A^{\prime}A)^\dagger A'$,
where $\mathbbm{1}_m$ is the $m\times m$ identity matrix,
and $(A^{\prime}A)^\dagger$ denotes a generalized inverse, since
$A$ may not have full column rank. The vectorization of an $m\times n$ matrix $A$
is denoted ${\rm vec}(A)$, which is the $mn\times 1$ vector obtained by
stacking the columns of $A$. For
square matrices $B$, $C$, we use $B>C$ (or $B\geq C$) to indicate that $B-C$ is positive (semi) definite.
We use $\nabla$ for the gradient of a function, \textit{i.e.} $\nabla f(x)$ is the vector of
partial derivatives of $f$ with respect to each component of $x$. We use ``wpa1'' for ``with probability approaching one''.

\section{Model}
\label{sec:model}

The random coefficients logit demand model is an aggregate market-level
model, formulated at the individual consumer-level.   Consumer $i$'s
utility of product $j$ in market\footnote{
The $t$ subscript can also denote different time periods.
} $t$ is given by
\begin{align}
   u_{ijt} &= \delta^0_{jt} + \epsilon_{ijt} + X'_{jt} \, v_i \; ,
\end{align}
where $\epsilon_{ijt}$ is an idiosyncratic product-specific preference shock, and
$v_i = (v_{i1}, \ldots, v_{iK})'$ is an idiosyncratic characteristic
preference. The mean utility is defined as
\begin{align}
   \delta^0_{jt} &= X_{jt}' \beta^0  +  \xi^0_{jt}\; ,
   \label{DefDelta}
\end{align}
where  $X_{jt} = \left(X_{1,jt},\ldots,X_{K,jt} \right)'$
is a vector of $K$ observed product characteristics (including price), and
$\beta^0=\left(\beta^0_1,\ldots,\beta^0_K\right)'$ is the corresponding
vector of coefficients.   Following BLP, $\xi^0_{jt}$ denotes unobserved product
characteristics of product $j$, which can vary across markets $t$.   This
is a ``structural error'', in that it is observed by all consumers when
they make their decisions, but is unobserved by the econometrician.  In this
paper, we
focus on the case where these unobserved product characteristics vary
across products and markets according to a factor structure:
\begin{align}
    \label{factorstructure}
\xi^0_{jt}=\lambda^{0 \prime}_j \, f^0_t + e_{jt} \; ,
\end{align}
where $\lambda^0_j = \left( \lambda^0_{1j},\ldots, \lambda^0_{Rj} \right)'$
is a vector of factor loadings corresponding to the $R$ factors\footnote{
Depending on the specific application one
has in mind one may have different interpretations for $\lambda_j$
and $f_t$. For example, in the case of national brands
sold in different markets it seems more natural to interpret $\lambda_j$ as the underlying
factor (a vector product qualities) and $f_t$ as the corresponding loadings (market specific
tastes for these qualities).  For convenience, we refer to
$f_t$ as factors and  $\lambda_j$ as factor loadings
throughout the whole paper, which is the typical naming
convention in applications where $t$ refers to time.}
$f^0_t =
\left( f^0_{1t},\ldots, f^0_{Rt} \right)'$, and $e_{jt}$
is a product and market specific error term.   Here $\lambda^{0 \prime}_j \, f^0_t$ represent
  interactive fixed effects, in that both the factors $f^0_t$ and factor
  loadings $\lambda^0_j$ are unobserved to the econometrician, and can
  be correlated arbitrarily with the observed product characteristics
  $X_{jt}$. We
assume that the number of factors $R$ is known.\footnote{
Known $R$ is also assumed in Bai (2009) and
Moon and Weidner~\cite*{MoonWeidner2015} for the
linear regression model with interactive fixed effects.
Allowing for $R$ to be unknown presents a substantial
technical challenge even for the linear model, and therefore goes beyond
the scope of the present paper.
}
The superscript
zero indicates the true parameters, and objects evaluated at the true
parameters. Let $\lambda^0 = (\lambda^0_{jr})$ and $f^0 = (\lambda^0_{tr})$ be
$J \times R$ and $T \times R$ matrices, respectively.

The factor structure in
equation (\ref{factorstructure}) approximates reasonably
some unobserved product and market characteristics of interest in an interactive form.   For example,
television advertising is well-known to be composed of a
product-specific component as well as an annual cyclical component
(peaking during the winter and summer months).\footnote{
cf. {\em TV Dimensions} \cite*{tv_dimensions}.}
The factors and factor loadings can also explain strong
correlation of the observed market shares over both products and markets,
which is a stylized fact
in many industries that has
motivated some recent dynamic oligopoly models of industry
evolution (e.g. Besanko and Doraszelski \cite*{BesankoDoraszelski2004}).
The standard BLP estimation approach, based on moment conditions,
allows for weak correlation across
markets and products, but does not admit strong correlation
due to shocks
that affect all products and markets simultaneously, which we model
via the factor structure.

To begin with, we assume
that the regressors $X_{jt}$ are exogenous with respect to the errors $e_{jt}$,
 that is, $X_{jt}$ and $e_{jt}$ are uncorrelated for given ($j$, $t$).
This assumption, however, is only made for ease of exposition, and in
both Section~\ref{sec:endogenousX} below and in the empirical
illustration, we consider the more general case where regressors (such
as price) may be endogenous.
Notwithstanding, regressors which are strictly exogenous
with respect to $e_{jt}$ can still be endogenous with respect to
the $\xi^0_{jt}$, due to correlation of the regressors with the factors and factor loadings.
Thus, including
the interactive fixed effects may ``eliminate'' endogeneity problems, so
that instruments for endogeneity may no longer be needed.
This possibility of estimating a demand model without searching for
instruments may be of great practical use in antitrust analysis.

Moreover, when
endogeneity persists even given the interactive fixed effects, then our
approach may allow for a larger set of IV's.   For instance, one criticism
of the so-called ``Hausman'' instruments (cf. Hausman
\cite*{hausmancereal}) -- that is, using the price of product $j$ in market
$t'$ as an instrument for the price of product $j$ in market $t$ -- is that they may not be independent of ``nationwide'' demand
shocks -- that is, product-specific shocks which are correlated across
markets.   Our interactive fixed effect $\lambda_j'f_t$ can be interpreted
as one type of nationwide demand shock, where the $\lambda_j$ factor
loadings capture common (nationwide) components in the shocks across different markets
$t$ and $t'$.   Since the instruments in our model can be
arbitrarily correlated with $\lambda_j$ and $f_t$, the use of Hausman
instruments in our model may be (at least partially) immune to the
aforementioned criticism.

Next, we introduce the key equations for market shares in the random-coefficient logit demand model.  Following Berry, Levinsohn, and Pakes \cite*{BerryLevinsohnPakes1995}, the
probability that agent $i$ chooses product $j$ in market $t$ takes the
multinomial logit form:
\begin{align}\label{choiceprobs}
   \pi_{jt}(\delta_{t},X_t,v_i) &= \frac {\exp\left( \delta_{jt} + X_{jt}' v_i \right) }
                          { 1 + \sum_{l=1}^J \, \exp\left( \delta_{lt} + X_{lt}' v_i \right) } \; .
\end{align}
We do not observe individual specific choices, but
market shares of the $J$ products in the $T$ markets.
The market share of product $j$ in market $t$ is given by
\begin{align}
   s_{jt}(\alpha^0,\delta_t,X_t) &= \int \, \pi_{jt}(\delta_{t},X_t,v) \,
                   d G_{\alpha^0}(v) \; ,
   \label{DefShares}
\end{align}
where $G_{\alpha^0}(v)$ is the known distribution of consumer taste $v_i$ over the product
characteristic, and $\alpha^0$ is an $L\times 1$ vector of parameters of this distribution.\footnote{The dependence of $\pi_{jt}(\delta_{t},X_t,v_i)$ and $s_{jt}(\alpha^0,\delta_t,X_t)$ on $t$ stems from the arguments $\delta_{t}$ and $X_t$.}

Underlying these derivations are assumptions that (i) the distributions of $\epsilon=(\epsilon_{ijt})$ and $v=(v_i)$ are mutually
independent, and are also independent of $X=(X_{jt})$ and $\xi^0=(\xi^0_{jt})$;
(ii) $\epsilon_{ijt}$ follows a marginal  type I extreme
value distribution
iid across $i$ and $j$ (but not necessarily independent across $t$).\footnote{When the index $t$ refers to time (or otherwise possesses some
natural ordering), then sequential exogeneity is allowed throughout the whole
paper, that is, $X_{jt}$ can be correlated with past values of the errors,
$e_{js}, s < t$.  The errors $e_{jt}$ are assumed to be independent across $j$ and $t$,
but heteroscedasticity is allowed.
}
Moreover, the most often used specification for the distribution of random coefficients in the literature is to assume that
they have a multivariate normal distribution,
that is,
$v \sim {\cal N}(0,\Sigma^0)$, where $\Sigma^0$ is a $K\times K$
matrix of parameters, which can be subject to constraints
(e.g. only one or a few regressors
may have random coefficients, in which case the components of $\Sigma^0$ are only non-zero
for these regressors), and $\alpha^0$ consists of the
independent parameters in $\Sigma^0$.\footnote{
We focus in this paper on the case where the functional form of the
distribution function $G_{\alpha}$ is known by the researcher.
Recent papers have addressed estimation when this is not known;
e.g. Bajari, Fox, Kim and Ryan \cite*{bajari_fox_kim_ryan}, \cite*{bajari_fox_kim_ryan2}.
}

The observables in this model are the market shares $s_{jt}$ and the
regressors $X_{jt}$.\footnote{
In the present paper we assume
that the true market shares $s_{jt}=s_{jt}(\delta^0_t)$ are observed.
Berry, Linton and Pakes \cite*{BerryLintonPakes2004} explicitly consider
sampling error in the observed market shares in their asymptotic
theory. Here, we abstract away from this additional complication and focus
on the econometric issues
introduced by the factor structure in $\xi^0$.
}
In addition, we need $M$ instruments
$Z_{jt}=( Z_{1,jt}, \ldots, Z_{M,jt} )'$ to construct extra (unconditional) moment conditions,
in addition to the unconditional moment conditions constructed by $X_{jt}$,
in order to estimate the parameters $\alpha$, with $M \geq L$.
These additional
instruments are also needed in the usual BLP estimation procedure,
even in the absence of the factor structure.  Suppose that
$X_{jt}$ is exogenous with respect to $\xi^0_{j,t}$. From this, we construct unconditional moment conditions $\mathbb{E}(X_{jt} \xi^0_{j,t}) = 0$. Then, extra moment conditions are still required to identify the covariance parameters in the
random coefficients distribution. Notice that
those $Z$'s may be non-linear functions
of the exogeneous $X$'s, so we do not necessarily need to observe additional exogenous variables.\footnote{
   If one is willing to impose the conditional moment condition  $\mathbbm{E}(e_{jt} | X_{jt})=0$,
   then valid $Z_{jt}$ can be constructed as non-linear transformations of $X_{jt}$.
}

Let $s = (s_{jt})$, $X_k = (X_{k,jt})$, $Z_{m} = (Z_{m,jt})$ and $e=(e_{jt})$ be $J \times T$ matrices,
and also define the tensors $X=(X_{k,jt})$ and $Z= (Z_{m,jt})$, which contain all observed product characteristics
and instruments.
   In the presence of the unobserved factor structure, it is difficult to identify
   regression parameters of regressors $X_k$ that have a factor structure themselves,
   which includes product invariant and time invariant regressors.
 Our assumptions below rule out all those $X_k$ and $Z_m$ that have a low rank
   when considered as a $J \times T$ matrix.\footnote{%
   This is exactly analogous to the usual short panel case, in which the presence of fixed effects for each cross-sectional unit precludes identification of the coefficients on time-invariant regressors.
   If the number of factors $R$ is known accurately, then the coefficients of these low-rank
   regressors can be identified, but the necessary regularity conditions are
   relatively cumbersome.  For ease of exposition we will therefore rule out both
   low-rank regressors and low-rank instruments by our assumptions below,
   and we refer to Bai \cite*{Bai2009} and
   Moon and Weidner \cite*{MoonWeidner2015} for a further discussion of this topic.
 }
The unknown parameters are $\alpha^0$, $\beta^0$,
$\lambda^0$, and $f^0$.

The existing literature on demand estimation usually considers asymptotic sequences with either
$J$ growing large and $T$ fixed, or $T$ growing large and $J$ fixed.
Under these standard asymptotic sequences, the
estimation of the nuisance parameters $\lambda^0$ and $f^0$
 creates a Neyman and Scott \cite*{NeymanScott1948}
incidental parameter problem: because the number
of nuisance parameters grows with the sample size,
the estimators for the parameters of interest become inconsistent.
Following some recent panel
data literature, for example,
Hahn and Kuersteiner \cite*{Hahn:2002p717,Hahn:2004p878} and
Hahn and Newey \cite*{Hahn:2004p882}, we handle this problem
by considering asymptotic sequences where both $J$ and $T$ become large.
Under this alternative asymptotic, the incidental parameter problem
is transformed into the issue of asymptotic bias in the limiting distribution
of the estimators of the parameters of interest. This asymptotic bias
can be characterized and corrected for.
Our Monte Carlo simulations suggest that the alternative asymptotic provides
a good approximation of the properties of our estimator at finite sample
sizes, as long as $J$ and $T$ are moderately large.

\section{Identification}

Given the non-linearity of the model, %
  questions regarding the %
identification of the model parameters of interest are naturally raised.
In the following we provide conditions under which the parameters $\alpha$
and $\beta$ as well as the product $\lambda f'$ are identified.
We do not consider
how to identify $\lambda$ and $f$ separately, because they only
enter into the model jointly as $\lambda f'$.\footnote{The transformation
$\lambda \rightarrow \lambda S$ and $f \rightarrow f S^{-1}$ gives observationally
equivalent parameters for any non-degenerate $R \times R$ matrix $S$.
Once the product $\lambda f'$ is identified, one can impose further
normalization restrictions to identify  $\lambda$ and $f$ separately, if desired.}

 Following standard identification arguments (e.g. Matzkin~\cite*{Matzkin2013}), our proof demonstrates identification by showing the
existence of an injective mapping from the model parameters
$(\alpha,\beta,\lambda f')$ and the distribution of the random elements of the model $(e,X,Z)$ to the distribution of the observed data $(s,X,Z)$,
where the random elements of the model are comprised of unobserved error terms, product characteristics, and instruments and the observed data are the market shares, product characteristics, and instruments.\footnote{
Injectivity implies that the mapping is one-to-one -- and
hence invertible -- along the relevant range.   The range of this
mapping excludes some distributions of $(s,X,Z)$; for instance,
distributions in which some of the market shares take zero values
with non-zero probability
cannot be generated by our model, due to the multinomial logit
structure.  See Gandhi, Lu, and Shi~\cite*{GandhiLuShi2013} for
additional discussion of estimating discrete-choice demand models when
some of the products are observed to have zero market shares.}

Our identification result utilizes a population distribution  of a full $J \times T$ panel of observables ($s$, $X$, $Z$), conditional
on parameters $\alpha$, $\beta$ and $\lambda f^{\prime}$.
The fact that we have nuisance parameters $\lambda_j$ and $f_t$ in both panel dimensions
makes the distribution of the full $J \times T$ panel of observables a natural starting point for the identification
discussion (where $J$ and $T$ are are finite constants in this section).
Normally in the large $N,T$ panel data literature (e.g. 
in Hahn and Newey \cite*{Hahn:2004p882},
Bai \cite*{Bai2009}, etc.) there is no explicit identification discussion, but
consistency as $N,T \rightarrow \infty$ (or in our case $J,T \rightarrow \infty$) 
is shown directly. The reason is that there is no fixed population distribution that 
corresponds to the sample as both panel dimensions become large.
Thus, when going from identification to estimation
there will not be a simple analog principle that allows to treat the sample as
multiple draws from the population.
 This is a general
conceptual issue, independent of our paper.
The inference results below therefore do not follow immediately from the identification result presented in this section;
in particular, the incidental parameter problem (Neyman and Scott \cite*{NeymanScott1948}) related to
inference of $\lambda_j$ and $f_t$ needs to be properly addressed.

For our identification result we assume, as in BLP,  that there exists a one-to-one relationship between
market shares and mean utilities, as summarized by the following assumption.
Let ${\cal B}_\alpha \subset \mathbbm{R}^L$ be a given parameter set for $\alpha$.

\newtheorem{INVassumption}{Assumption}
\renewcommand{\theINVassumption}{INV}
\begin{INVassumption}[\bf Invertibility Assumptions] $\phantom{a}$
   \label{ass:INV}
    We assume that
equation \eqref{DefShares} is invertible, that is, for each
market $t$ the mean utilities $\delta_{t}=\left( \delta_{1t},\ldots,\delta_{Jt} \right)$
are unique functions of $\alpha \in {\cal B}_\alpha$, the market shares
$s_{t}=\left(s_{1t},\ldots,s_{Jt}\right)$, and the regressors
$X_{t}=\left(X_{1t},\ldots,X_{Jt}\right)$. We denote these functions by $\delta_{jt}(\alpha,\, s_t,\, X_t)$.\footnote{
Note that the dependence of $\delta_{jt}(\alpha,\, s_t,\, X_t)$ on $t$ stems from the arguments $s_t$ and $X_t$.}

\end{INVassumption}

Berry, Gandhi, and Haile~\cite*{BerryGandhiHaile2013} provide general conditions under which this invertibility assumption is satisfied,
and Berry and Haile \cite*{berry_haile_market} and Chiappori and
Komunjer \cite*{chiappori_komunjer} utilize this
inverse mapping in their nonparametric identification results.

Using Assumption~\ref{ass:INV} and the specifications \eqref{DefDelta} and \eqref{factorstructure} we have
\begin{align}
    \delta^0_{jt} &= \delta_{jt}(\alpha^0,\, s_t,\, X_t)
   =  \sum_{k=1}^K \, \beta^0_k \, X_{k,jt} +  \sum_{r=1}^R \, \lambda^0_{jr} f^0_{tr} + e_{jt} \; .
   \label{model0}
\end{align}
In $JT$-vector notation this equation can be written as
$\delta^{\rm vec}(\alpha^0) = x \beta^0 + \sum_{r=1}^R f^0_{\cdot r} \otimes \lambda^0_{\cdot r} + e^{\rm vec}$,
where $\delta^{\rm vec}(\alpha) = {\rm vec}[  \delta(\alpha,\, s,\, X) ] $
and
  $e^{\rm vec} = {\rm vec}(e)$ are $JT$-vectors, and $x$ is a $JT \times K$ matrix with
  columns $x_{.,k} = {\rm vec}\left( X_{k} \right)$. For simplicity we suppress the dependence of $\delta^{\rm vec}(\alpha)$
  on $s$ and $X$.
It is furthermore convenient to define the $JT \times M$ matrix $z$
with columns $z_{.,m} = {\rm vec}\left( Z_{m} \right)$, the mean utility difference
$d(\alpha) = \delta^{\rm vec}(\alpha) - \delta^{\rm vec}(\alpha^0)$,
and the unobserved utility difference  $\Delta \xi_{\alpha,\beta} = d(\alpha)-x(\beta-\beta^0)$.
Both $d(\alpha)$ and $\Delta \xi_{\alpha,\beta} $ are $JT$ vectors.
Note that $\Delta \xi_{\alpha,\beta}$
is simply the vectorized difference of the
residual unobserved product characteristic
at $(\alpha,\beta)$ and $(\alpha^0,\beta^0)$.
In the following the indices $j$ and $t$ run from $1$ to $J$
and $1$ to $T$, respectively.

\newtheorem{IDassumption}{Assumption}
\renewcommand{\theIDassumption}{ID}
\begin{IDassumption}[\bf Assumptions for Identification] $\phantom{a}$
   \label{ass:ID}
   \begin{itemize}
      \item[(i)] The second moments of $\delta_{jt}(\alpha)$,
                 $X_{jt}$ and $Z_{jt}$ exist for all $\alpha$,
                 and all $j$, $t$.

      \item[(ii)]  $\mathbbm{E}(e_{jt})=0$.

      \item[(iii)]
     $\mathbbm{E}(X_{jt} e_{jt})=0$,
    $\mathbbm{E}(Z_{jt} e_{jt})=0$, for all $j$, $t$.\footnote{
      The exogeneity assumption $ \mathbbm{E}(X_{jt} e_{jt})=0$ can be relaxed.
      Regression coefficients $\beta^{\rm end}$
       on endogenous regressors need to be included in the parameter vector $\alpha$,~i.e.
      $\alpha$ is replaced by $(\alpha, \beta^{\rm end})$, 
      see Section~\ref{sec:endogenousX} below.
    }

      \item[(iv)]
       $\mathbbm{E}[(x,z)' (\mathbbm{1}_T \otimes M_{(\lambda,\lambda^0)}) (x,z)]
     \geq b \, \mathbbm{1}_{K+M}$, for some $b>0$ and all $\lambda \in \mathbbm{R}^{J \times R}$.\footnote{%
     Here,
     $P_{(\lambda,\lambda^0)} = (\lambda,\lambda^0) [(\lambda,\lambda^0)^{\prime}(\lambda,\lambda^0)]^\dagger (\lambda,\lambda^0)'$, where $\dagger$ refers to a generalized inverse,
     and
     $M_{(\lambda,\lambda^0)} = \mathbbm{1}_J - P_{(\lambda,\lambda^0)}$ are the $J \times J$ matrices
     that project onto and orthogonal to the span of $(\lambda,\lambda^0)$.}
      \item[(v)]
   For all $(\alpha,\beta) \neq (\alpha^{0},\beta^0)$,
   and all $\lambda \in \mathbbm{R}^{J \times R}$ we assume that\footnotemark[\value{footnote}] \\
     $\mathbbm{E}\big[ \Delta \xi_{\alpha,\beta}' \, (x,z) \big]
     \mathbbm{E}\big[ (x,z)' (x,z) \big]^{-1}
     \mathbbm{E}\big[ (x,z)' \, \Delta \xi_{\alpha,\beta}  \big]
     > \mathbbm{E}\left[ \Delta \xi_{\alpha,\beta}'
     \left( \mathbbm{1}_T \otimes P_{(\lambda,\lambda^0)}\right)
     \Delta \xi_{\alpha,\beta} \right]$.

   \end{itemize}
\end{IDassumption}

In this assumption, and also for the remainder of the paper, we treat the fixed effects $\lambda^0$ and $f^0$ as non-random parameters, that is, all expectations in Assumption ID are implicitly conditional on $\lambda^0$ and $f^0$.
 The assumptions are discussed in Section~\ref{sec:discussIDcon} below.

To formulate our identification result we need to introduce some additional notation.
We denote the set of joint distributions of $e$, $X$, $Z$ by  ${\cal F}_{e,X,Z}$,
and the set of joint distributions of $s$, $X$, $Z$ (the observables) by ${\cal F}_{s,X,Z}$.
The model described in Section~\ref{sec:model} gives unique market shares $s$ for
any given $e$, $X$, $Z$ and parameters $\alpha$, $\beta$, $\lambda f^{\prime}$. The model therefore also
uniquely describes the distribution of observables for a given distribution
$F_{e,X,Z} \in {\cal F}_{e,X,Z}$ and parameters $\alpha$, $\beta$, $\lambda f^{\prime}$, and
we denote this distribution of observables given by the model
as $\Gamma(\alpha, \beta, \lambda f^{\prime}, F_{e,X,Z}) \in {\cal F}_{s,X,Z}$.
We say that two distributions $F_1,F_2 \in {\cal F}_{s,X,Z}$
are equal if the corresponding joint cdf's are the same, and we write $F_1 = F_2$ in that case.
Analogously, we define equality on ${\cal F}_{e,X,Z}$.

\begin{theorem}[\bf Identification] $\phantom{a}$
   \label{th:id-new}
  Let Assumption~\ref{ass:INV} be satisfied. Let
    $F_{e,X,Z}^0 \in {\cal F}_{e,X,Z}$ be such that it satisfies Assumption~\ref{ass:ID}.
    Let $F_{e,X,Z} \in {\cal F}_{e,X,Z}$ and consider two sets of parameters
    $(\alpha, \beta, \lambda f^{\prime})$
    and $(\alpha^0, \beta^0, \lambda^0 f^{0 \prime})$.
    Then,    $\Gamma(\alpha, \beta, \lambda f^{\prime}, F_{e,X,Z})
         = \Gamma(\alpha^0, \beta^0, \lambda^0 f^{0 \prime}, F^0_{e,X,Z})$ implies that
     $\alpha = \alpha^0$, $\beta = \beta^0$,  $\lambda f^{\prime}= \lambda^0 f^{0 \prime}$
     and $F_{e,X,Z} = F_{e,X,Z}^0$.
\end{theorem}

The theorem states that if the distribution of observables
  $F^0_{s,X,Z} =\Gamma(\alpha^0, \beta^0, \lambda^0 f^{0 \prime}, F^0_{e,X,Z})$
is generated from the parameters $(\alpha^0, \beta^0, \lambda^0 f^{0 \prime})$
and $F_{e,X,Z}^0$, satisfying Assumption~\ref{ass:ID},
then any other $(\alpha, \beta, \lambda f^{\prime})$ and $F_{e,X,Z}$ that generate the same distribution
of observables $F^0_{s,X,Z} = \Gamma(\alpha, \beta, \lambda f^{\prime},F_{e,X,Z})$ must be equal to the original
 $(\alpha^0, \beta^0, \lambda^0 f^{0 \prime})$
and $F_{e,X,Z}^0$. In other words, we can uniquely recover the model parameters
from the distribution of observables. Two observationally equivalent model structures
$(\alpha^0, \beta^0, \lambda^0 f^{0 \prime}, F^0_{e,X,Z})$
and $(\alpha, \beta, \lambda f^{\prime},F_{e,X,Z})$ need to be identical.

The key tool for the proof of Theorem~\ref{th:id-new} is the
the expected least squares objective function
$$Q\left( \alpha ,\beta ,\gamma ,\lambda ,f ; F^0_{s,X,Z} \right) =
  \mathbbm{E}_0  \left\{ \sum_{j=1}^J
               \sum_{t=1}^T
               \left[ \delta_{jt}(\alpha)
              - X_{jt}' \beta
              -  Z_{jt}' \gamma  -  \lambda_j' f_t \right]^2 \right\},$$
where $\gamma \in \mathbbm{R}^L$ is an auxiliary parameter,
and  $\mathbbm{E}_0$ refers to the expectation under the distribution of observables
$F^0_{s,X,Z}$,\footnote{%
Normally, we refer to $\mathbbm{E}_0$ simply as $\mathbbm{E}$. We only use different notation here to
stress at which point the argument $F^0_{s,X,Z}$ enters into $Q\left( \alpha ,\beta ,\gamma ,\lambda ,f ; F^0_{s,X,Z} \right)$.
}
 which is assumed to be generated from the model, i.e.
$F^0_{s,X,Z}= \Gamma(\alpha^0, \beta^0, \lambda^0 f^{0 \prime}, F^0_{e,X,Z})$,
with $F^0_{e,X,Z}$ satisfying Assumption~\ref{ass:ID}.

The true value of the auxiliary parameter $\gamma$ is zero, because of the exclusion restriction on $Z_{jt}$.
 In the proof of Theorem~\ref{th:id-new} we show
 that under our assumptions the minimizer of
 $\displaystyle Q\left( \alpha ,\beta ,\gamma ,\lambda ,f ; F^0_{s,X,Z} \right)  $
  over $(\beta, \lambda, f, \gamma)$, for fixed $\alpha$, only satisfies $\gamma=0$ if and only if
  $\alpha = \alpha^0$.
 Thus, by
  using  the  expected least squares objective function as a tool we can uniquely
   identify $\alpha^0$ from the distribution of obervables $F^0_{s,X,Z}$.
  Having identified $\alpha^0$ we can identify $\beta^0$ and $\lambda^0 f^{0 \prime}$
  simply as the unique minimizers of
   $\displaystyle Q\left(\alpha^0 ,\beta ,\gamma ,\lambda ,f ; F^0_{s,X,Z}
  \right)$.  These findings immediately preclude
observational equivalence,
  {\em viz} two sets of distinct parameters $(\alpha^0, \beta^0, \lambda^0 f^{0 \prime})
  \neq (\alpha^1, \beta^1, \lambda^1 f^{1 \prime})$ which are both consistent
  with the observed distribution $F^0_{s,X,Z}$.
For complete details we refer to the proof in the appendix.
   Furthermore, our identification argument is constructive, as it leads naturally to
the LS-MD estimator which we introduce in
subsequent sections.

\subsection{Discussion of the Identification Conditions}
\label{sec:discussIDcon}

In this section we discuss the conditions of the identification theorem.
First, we note that when no factors are present ($R=0$), then our
identification Assumptions~\ref{ass:ID} below essentially require that the unconditional moment conditions
 $\mathbbm{E}(X_{jt} e_{jt})=0$ and $\mathbbm{E}(Z_{jt} e_{jt})=0$
 uniquely identify the model parameters $\alpha$ and $\beta$,
 thus following the original identification strategy in BLP (1995).\footnote{
As such, our identification results do not add to the literature on non-parametric identification
of the BLP model (as in Berry and Haile~\cite*{berry_haile_market}, Chiappori and
Komunjer~\cite*{chiappori_komunjer}, Bajari, Fox, Kim and Ryan~\cite*{bajari_fox_kim_ryan});
our concern is, rather, to show that the logit demand model with parametrically-distributed random coefficients can still be identified after the
introduction of the interactive fixed effects.}

Assumption $(i)$ demands existence of second moments, assumption $(ii)$ requires the error process to have zero mean, and assumption $(iii)$ imposes exogeneity of the product characteristics $X_{jt}$ and the instruments $Z_{jt}$ with respect to the error $e_{jt}$ (endogenous regressors
are discussed in Section~\ref{sec:endogenousX}).
Apart from the term $M_{(\lambda,\lambda^0)}$,
    Assumption~\ref{ass:ID}$(iv)$ is a standard non-collinearity condition on the product characteristics and the instruments -- which jointly appear as regressors in the first step of \eqref{estimator}. The generalized condition
      $\mathbbm{E}[(x,z)' (\mathbbm{1}_T \otimes M_{(\lambda,\lambda^0)}) (x,z)] \geq b > 0$ requires non-collinearity of the regressors even after projecting out all directions proportional to the true factor loading $\lambda^0$ and to any other possible factor loadings $\lambda$.
  A sufficient condition for this assumption is the rank condition ${\rm rank}[\mathbbm{E}(\Xi \, \Xi')] > 2R$ for any non-zero linear combination $\Xi = \beta \cdot X + \gamma \cdot Z$. This rank condition, for example,   rules out product-invariant regressors and
  instruments, as already mentioned above.

Those parts of the conditions $(i)$ to $(iv)$ that do not contain $Z_{jt}$ are used
to identify $\beta^0$ and $\lambda^0 f^{0 \prime}$ when $\alpha^0$ is already identified.
These conditions are typical regularity conditions for identification of a linear regression model with a modification only required in condition~$(iv)$ to accommodate the interactive fixed effects. (See also Moon and Weidner~\cite*{MoonWeidner2015b}.)

  The key additional assumption that we need for identification of
  $\alpha^0$ is Assumption~\ref{ass:ID}$(v)$. Note that $\Delta \xi_{\alpha^0,\beta^0}=0$, that is, both the left and right side of the inequality in~assumption $(v)$ are zero for $(\alpha,\beta)=(\alpha^0,\beta^0)$, which is why this case is explicitly ruled out in the assumption. The left hand side of the inequality in assumption~$(v)$ is the sum of squares of that part of $\Delta \xi_{\alpha,\beta}$ that is explained by the regressors $x$ and the instruments $z$. The right hand side is the sum of squares of that part of $\Delta \xi_{\alpha,\beta}$ that is explained by the true factor loading $\lambda^0$ and an arbitrary other factor loading $\lambda$. Thus, the condition is a relevance condition on the instruments, which requires that the explanatory power of the regressors and the instruments needs to be larger than the explanatory power of $\lambda$ and $\lambda^0$ for $\Delta \xi_{\alpha,\beta}$.

  A more concrete intuition for Assumption~\ref{ass:ID}$(v)$
  can be obtained in the case without factors. Without factors, the identification condition simplifies to $\forall (\alpha,\beta) \neq (\alpha^0,\beta^0):$
\begin{align}
     \mathbbm{E}\left[ \Delta \xi_{\alpha,\beta}' \, (x,z) \right]
     \mathbbm{E}\left[ (x,z)' (x,z) \right]^{-1}
     \mathbbm{E}\left[ (x,z)' \, \Delta \xi_{\alpha,\beta}  \right]
     &> 0 .
\end{align}
This can be shown to be equivalent to the statement $\forall \alpha \neq \alpha^0:$
\begin{align}
\mathbbm{E}\left[ d(\alpha)' (x,z) \right]
   \mathbbm{E}\left[ (x,z)' (x,z) \right]^{-1}
   \mathbbm{E}\left[ (x,z)' d(\alpha)  \right]
    >\mathbbm{E}\left[ d\left( \alpha \right) ^{\prime }x\right]
    \mathbbm{E}\left(
x^{\prime }x\right) ^{-1}
\mathbbm{E}\left[ x^{\prime }d\left( \alpha \right) \right] .
\end{align}
We see that this condition is nothing more than the usual instrument
relevance  condition (for $z$ in this case) underlying
the typical GMM approach in estimating BLP models.
It can also be shown to be equivalent to the condition that
for all $\alpha \neq \alpha^0$ the matrix
  $\mathbbm{E}[ (d(\alpha),x)' (x,z)]$
  has full rank (equal to $K+1$).

The matrix valued function $\delta(\alpha)=\delta(\alpha,s,X)$ was introduced as the inverse of equation \eqref{DefShares} for the market shares $s_{jt}(\delta_t)$.
Thus, once a functional form for
$s_{jt}(\delta_t)$ is chosen and some distributional assumptions on the data generating
process are made, it is in principle possible to analyze
Assumption~\ref{ass:ID}$(v)$ further and to discuss validity and optimality of the
instruments. Unfortunately, too little is known about the
properties of $\delta(\alpha)$ to enable a general analysis.\footnote{
This is a problem not only with our approach, but also with
the estimators in BLP, and for
Berry, Linton and Pakes \cite*{BerryLintonPakes2004}.}
For this reason, in our Monte Carlo simulations in section~\ref{sec:MC} below,
we provide both analytical and
and numerical verifications for Assumption~\ref{ass:ID}$(v)$ for the
specific setup there.

The final remark is that Assumption~\ref{ass:ID}$(v)$ also restricts the family of the distribution of the random coefficient.  As a very simple example,
suppose that we would specify the distribution $G_{\alpha}$ for the random vector $v$
as $v \sim {\cal N}(\alpha_1,\alpha_2)$, where $\alpha =\left( \alpha _{1},\alpha _{2}\right)$,
and we would also include a constant in the vector of regressors $X_{jt}$. Then, the regression
coefficient on the constant and $\alpha_1$ cannot be jointly identified (because they both shift mean utility by a constant,
but have no other effect), and Assumption~\ref{ass:ID}$(v)$
will indeed be violated in this case.

\section{LS-MD Estimator}

If $\delta^0_{jt}$ is known,
then the above model reduces to the linear panel regression model with
interactive fixed effects.
Estimation of this model was discussed under fixed $T$ asymptotics in,
 for example, Holtz-Eakin, Newey and Rosen \cite*{HoltzEakin-Newey-Rosen1988}, and
Ahn, Lee, Schmidt \cite*{AhnLeeSchmidt2001},
and for $J,T \rightarrow \infty$ asymptotics in
Bai \cite*{Bai2009}, and Moon and Weidner \cite*{MoonWeidner2015,MoonWeidner2015b}.

The computational challenge in estimating the model \eqref{model0} lies in
accommodating both the model parameters $(\alpha,\beta)$, which in the
existing literature has mainly been done in a GMM framework, as well
as the nuisance elements $\lambda_j, f_t$, which in the existing
literature have been treated using a principal components
decomposition in a least-squares context (e.g., Bai \cite*{Bai2009}, and
Moon and Weidner \cite*{MoonWeidner2015,MoonWeidner2015b}).
Our estimation procedure -- which mimics the identification proof
discussed previously -- combines both
the GMM approach to demand estimation and the least squares
  approach to the interactive fixed effect model.

{\bfseries Definition:} the {\em least squares-minimum distance (LS-MD)} estimators for $\alpha$ and $\beta$ are defined by
\begin{align}
   &\text{Step 1 (least squares): for given $\alpha$ let}
   \nonumber \\
   & \qquad  \delta(\alpha) \, = \, \delta(\alpha,\, s,\, X) \; ,
   \nonumber \\
   & \qquad \left( \tilde \beta_{\alpha} \, , \;
          \tilde \gamma_{\alpha} \, , \; \tilde \lambda_{\alpha} \, , \; \tilde f_{\alpha} \right)
     \, = \, \argmin_{ \{ \beta, \, \gamma, \, \lambda, \, f \} }  \,
               \sum_{j=1}^J
               \sum_{t=1}^T
               \left[ \delta_{jt}(\alpha)
              - X_{jt}' \beta
              -  Z_{jt}' \gamma  -  \lambda_j' f_t \right]^2 ,
   \nonumber \\
   &\text{Step 2 (minimum distance):}
   \nonumber \\
    &\qquad \widehat \alpha \, = \,  \argmin_{\alpha \in {\cal B}_\alpha} \,
                    \tilde \gamma'_{\alpha} \, \smallW_{JT} \, \tilde \gamma_{\alpha} \; ,
   \nonumber \\
   &\text{Step 3 (least squares):}
   \nonumber \\
   & \qquad  \delta(\widehat \alpha) \, = \, \delta(\widehat \alpha,\, s,\, X) \; ,
   \nonumber \\
   & \qquad \left( \widehat \beta \, , \;
          \widehat \lambda \, , \; \widehat f \right)
     \, = \, \argmin_{ \{ \beta ,\, \lambda, \, f \} }   \,
                 \sum_{j=1}^J
               \sum_{t=1}^T
               \left[ \delta_{jt}(\widehat{\alpha})
              - X_{jt}' \beta -  \lambda_j' f_t \right]^2  .
   \label{estimator}
\end{align}
Here, $\beta \in \mathbbm{R}^K$,
$\delta(\alpha,\, s,\, X)$, $X_k$ and $Z_m$ are $J \times T$ matrices,
$\lambda$ is $J\times R$, $f$ is $T\times R$,
$\smallW_{JT}$ is a positive definite $M\times M$ weight matrix,
${\cal B}_\alpha \subset \mathbbm{R}^L$ is an appropriate parameter set
for $\alpha$.

Steps 1 and 2 are nested, because $\tilde \gamma_{\alpha}$ defined by step 1
needs to be calculated
multiple times while performing the numerical optimization in step 2, 
but step 3 only needs to be performed once after the calculation of $\widehat \alpha$ in step 2 is finished.  Step 1 resembles the linear least-squares estimators with interactive fixed effects considered in Bai \cite*{Bai2009} and Moon and Weidner \cite*{MoonWeidner2015}, but because our model also includes the nonlinear parameter~$\alpha$, this step is nested
within step 2, which involves iteration over different candidate values for $\alpha$.

In step 1, we include the IV's $Z_m$ as auxiliary regressors, with
coefficients $\gamma \in \mathbbm{R}^M$.   Step 2 is based on imposing the exclusion
restriction on the IV's, which requires that $\gamma=0$, at the true value
of $\alpha$.
Thus, we first estimate
$\beta$, $\lambda$, $f$, and the instrument coefficients $\gamma$
by least squares for fixed $\alpha$, and subsequently we estimate $\alpha$ by minimizing the norm of $\tilde \gamma_\alpha$ with respect to $\alpha$.

  Step 3 in \eqref{estimator}, which defines $\widehat \beta$,
  is just a repetition of step 1, but with $\alpha=\widehat \alpha$ and $\gamma=0$.
  One could also use the step 1 estimator $\tilde \beta_{\widehat \alpha}$
  to estimate $\beta$. Under the assumptions for consistency of $(\widehat \alpha,\widehat \beta)$
  presented below, this alternative estimator is also consistent for $\beta^0$.
  However, in general $\tilde \beta_{\widehat \alpha}$ has a larger variance than $\widehat \beta$, since irrelevant regressors are included in the estimation of
  $\tilde \beta_{\widehat \alpha}$.

For given $\alpha$, $\beta$ and $\gamma$ the optimal factors and factor
loadings in the least squares problems in step 1 (and step 3) of
\eqref{estimator} turn out to be the principal components estimators for
$\lambda$ and $f$. These incidental parameters can therefore be
      concentrated out easily,
and the remaining objective function for $\beta$ and $\gamma$ turns out to be given by an
eigenvalue problem (see e.g. Moon and Weidner \cite*{MoonWeidner2015,MoonWeidner2015b} for details), namely
\begin{align}
    \left( \tilde \beta_{\alpha} \, , \;
          \tilde \gamma_{\alpha} \right)
     \, &= \, \argmin_{ \{ \beta, \, \gamma \} }  \,
            \sum_{r=R+1}^{T} \mu_r\left[ \left( \delta(\alpha)
              - \beta \cdot X
              - \gamma \cdot Z  \right)'
      \left( \delta(\alpha)
              - \beta \cdot X
              - \gamma \cdot Z  \right)
                \right] \; ,
\end{align}
where $\beta \cdot X = \sum_{k=1}^K \, \beta_k \, X_{k}$,
$\gamma \cdot Z = \sum_{m=1}^M \, \gamma_m \, Z_m$,
and $\mu_r(.)$ refers to the $r$'th largest eigenvalue of the argument
matrix.
This formulation greatly simplifies the numerical calculation of the estimator, since
eigenvalues are easy and fast to compute, and we only need to perform numerical optimization over $\beta$ and $\gamma$, not over $\lambda$ and $f$.

The step 1 optimization problem in \eqref{estimator} has the same structure
as the interactive fixed effect regression model. Thus, for $\alpha =
\alpha^0$ it is known
from Bai \cite*{Bai2009}, and Moon and Weidner \cite*{MoonWeidner2015,MoonWeidner2015b}
that (under their assumptions) $\tilde \beta_{\alpha^0}$ is $\sqrt{JT}$-consistent
for $\beta^0$ and asymptotically normal as $J,T \rightarrow \infty$
with $J/T \rightarrow \kappa^2$, $0<\kappa<\infty$.

Step 1 also involves solving for the vector of $\delta$'s which solves the market share equations (2.5), at a given value for $\alpha$.  This computational problem is well-studied in the BLP literature.\footnote{We solve it using nonlinear equation solvers, which is a relatively standard procedure from the existing BLP literature. Its validity is ensured by results (in Berry, Levinsohn, Pakes (1995)) showing that, for fixed $\alpha$, these equations constitute a contraction mapping, and the nonlinear equation solver recovers the (unique) fixed point.} 

The LS-MD estimator we propose above is distinctive,
because of the inclusion of the instruments $Z$ as regressors in the
first-step.   This can be understood as a generalization of an estimation
approach for a linear regression model with endogenous regressors.
Consider a simple structural equation $y_{1}=Y_{2}\alpha +e,$ where the
endogenous regressors $Y_{2}$ have the reduced form specification
$Y_{2}=Z\delta +V$,
and $e$ and $V$ are correlated. The two stage least squares estimator of $%
\alpha $ is $\widehat{\alpha}_{\rm 2SLS}=\left( Y_{2}^{\prime }P_{Z}Y_{2}\right)
^{-1}Y_{2}^{\prime }P_{Z}y_{1},$ where $P_{Z}=Z\left( Z^{\prime }Z\right)
^{-1}Z^{\prime }$. In this set up, it is possible to show that $\widehat{\alpha}%
_{\rm 2SLS}$ is also an LS-MD estimator with a suitable choice of the weight
matrix. Namely, in the first step the OLS regression of
$\left(y_{1}-Y_{2}\alpha \right) $ on $Z$ yields the OLS
estimator $\tilde{\gamma}_{\alpha }=\left( Z^{\prime }Z\right)
^{-1}Z^{\prime }\left( y_{1}-Y_{2}\alpha \right)$. Then, in the second step
minimizing the distance $\tilde{\gamma}_{\alpha }^{\prime }W\tilde{\gamma}%
_{\alpha }$ with respect to $\alpha $ gives $\widehat{\alpha}( W) =%
[ Y_{2}^{\prime }Z( Z^{\prime }Z) ^{-1}W( Z^{\prime
}Z) ^{-1}Z^{\prime }Y_{2} ]^{-1} \linebreak[1] [ Y_{2}^{\prime }Z(
Z^{\prime }Z) ^{-1}W( Z^{\prime }Z) ^{-1}Z^{\prime }y_{1}]$.
Choosing $W=Z^{\prime }Z$ thus results in
$\widehat{\alpha}=\widehat{\alpha}\left( Z^{\prime }Z\right) =\widehat{\alpha}_{\rm 2SLS}$.
Obviously, for our nonlinear
model, strict 2SLS is not applicable; however, our estimation approach
can be considered a generalization of this alternative iterative
estimator, in which the exogenous instruments $Z$ are included as
``extra'' regressors in the initial least-squares step.\footnote{%
Moreover, the presence of the factors makes it inappropriate to use the moment
condition-based GMM approach proposed by BLP, see Appendix~\ref{app:GMM}.
Moment based approaches to factor model estimation like
 Holtz-Eakin, Newey and Rosen \cite*{HoltzEakin-Newey-Rosen1988}
and  Ahn, Lee, Schmidt \cite*{AhnLeeSchmidt2001,AhnLeeSchmidt2013}
would also have to be non-trivially extended to 
handle the random coefficient parameter $\alpha$ in the presence of two dimensional incidental parameters in a nonlinear framework, but we have not explored this possibility.}

The two-step procedure in the LS-MD estimation is similar to the two stage estimation method in  Chernozhukov and Hansen \cite*{ChernozhukovHansen2006} that investigated endogenous quantile regressions.

\subsection{Extension: regressor endogeneity with respect to $e_{jt}$}
\label{sec:endogenousX}

So far, we have assumed that the regressors $X$
could be endogenous only through the factors $\lambda_{j}' f_t$, and
they are exogenous wrt $e$. However, this could be restrictive in
some applications, for example, when price $p_{jt}$ is determined by
$\xi_{jt}$ contemporaneously. Hence, we consider here the possibility that the
regressors $X$ could also be correlated with
$e$.   This is readily accommodated within our framework.   Let
$X^{\rm end}\subset X$ denote the endogenous regressors, with
$\text{dim}(X^{\rm end})=K_2$.   (Hence, the number of exogenous regressors equals
$K-K_2$.)  Similarly, let $\beta^{\rm end}$ denote the coefficients on these
regressors, while $\beta$ continues to denote the coefficients on the
exogenous regressors.
Correspondingly, we assume that $M$, the number of instruments, exceeds $L+K_2$.

{\bfseries Definition:} the {\em least-squares minimum  distance (LS-MD)}
estimators for $\alpha$ and $\beta$ with endogenous regressors $X^{end}$
is defined by:
\begin{align}
   &\text{step 1: for given $\alpha^{\rm end}=(\alpha, \beta^{\rm end})$ let}
   \nonumber \\
   & \qquad  \delta(\alpha) \, = \, \delta(\alpha,\, s,\, X) \; ,
   \nonumber \\
   & \qquad \left( \tilde \beta_{\alpha^{\rm end}} \, , \;
          \tilde \gamma_{\alpha^{\rm end}} \, ,
            \; \tilde \lambda_{\alpha^{\rm end}} \, , \; \tilde f_{\alpha^{\rm end}} \right)
     \, = \, \argmin_{ \{ \beta , \, \gamma, \, \lambda, \, f \} }  \,
               \sum_{j=1}^J
               \sum_{t=1}^T
               \left[ \delta_{jt}(\alpha)
              - X^{\rm end \prime}_{jt} \beta^{\rm end}
              - X_{jt}' \beta
              -  Z_{jt}' \gamma  -  \lambda_j' f_t \right]^2  ,
   \nonumber \\
   &\text{step 2:}
   \nonumber \\
    &\qquad \widehat \alpha^{\rm end}
     = (\widehat \alpha, \widehat \beta^{\rm end})
       \, = \,  \argmin_{\alpha^{\rm end} \in {\cal B}_\alpha \times {\cal B}^{\rm end}_\beta} \,
                    \tilde \gamma'_{\alpha^{\rm end}} \, \smallW_{JT} \, \tilde \gamma_{\alpha^{\rm end}} \; ,
   \nonumber \\
   &\text{step 3:}
   \nonumber \\
   & \qquad  \delta(\widehat \alpha) \, = \, \delta(\widehat \alpha,\, s,\, X) \; ,
   \nonumber \\
   & \qquad \left( \widehat \beta \, , \;
      \widehat \lambda \, , \; \widehat f \right)
     \, = \, \argmin_{ \{ \beta \in \mathbbm{R}^K, \, \lambda, \, f \} }  \,
                 \sum_{j=1}^J
               \sum_{t=1}^T
               \left[ \delta_{jt}(\widehat{\alpha})
              - X^{\rm end \prime}_{jt} \beta^{\rm end}
              - X_{jt}' \beta -  \lambda_j' f_t \right]^2 ,
   \label{estimatorEND}
\end{align}
where ${\cal B}_{\alpha}$ and  ${\cal B}^{\rm end}_{\beta}$
are parameter sets for $\alpha$ and $\beta^{\rm end}$.

The difference between this estimator, and the previous one for
which all the regressors were assumed exogenous, is that the
estimation of $\beta^{\rm end}$, the coefficients on the endogenous
regressors $\tilde X$, has been moved to the second step.
 The  estimation procedure in \eqref{estimatorEND} can me mapped
into our original LS-MD procedure in  \eqref{estimator}, if we make the following formal replacements:
\begin{align}
    \alpha^{\rm end}=(\alpha, \beta^{\rm end})  &\mapsto \alpha  , &
     \delta(\alpha) - \beta^{\rm end} \cdot X^{\rm end}   &\mapsto \delta(\alpha) .
     \label{ReplacementEND}
\end{align}
Thus, by changing the meaning of $\alpha$ and $\delta(\alpha)$ accordingly,
the identification result above is still valid,
and all results below on the consistency, asymptotic distribution and bias correction
of the LS-MD estimator \eqref{estimator} with only (sequentially) exogenous regressors
directly  generalize to the estimator \eqref{estimatorEND} with
more general endogenous regressors.
Given this discussion, we see that the original BLP (1995) model can
be considered a special case of our model in which factors are absent (i.e. $R=0$).

\section{Consistency and Asymptotic Distribution}
\label{sec:asymptotics}

In this section we present our results on the properties of
the LS-MD estimator $\widehat \alpha$ and $\widehat \beta$ defined in \eqref{estimator}
under the asymptotics $J,T \rightarrow \infty$.

\begin{assumption}[\bf Assumptions for Consistency] $\phantom{a}$
    \label{ass:con}
   \begin{itemize}
     \item[(i)] $\displaystyle \sup_{\alpha \in {\cal B}_{\alpha} \setminus \alpha^0}
                   \frac{ \| \delta(\alpha)-\delta(\alpha^0) \|_F }
                        { \| \alpha - \alpha^0 \| } = {\cal O}_p(\sqrt{JT})$, \\[5pt]
     $\|X_k\|_F = {\cal O}_p(\sqrt{JT})$, $\|Z_m\|_F = {\cal O}_p(\sqrt{JT})$, for  $k=1,\ldots,K$ and $m=1,\ldots,M$.
      \item[(ii)]   $\|e\| = {\cal O}_p(\sqrt{\max(J,T)})$.
      \item[(iii)]  $\frac 1 {JT} \, {\rm Tr}\left( X_k e' \right) = o_p(1)$, \;
                      $\frac 1 {JT} \, {\rm Tr}\left( Z_m e'\right) = o_p(1)$, \;
                      for $k=1,\ldots,K$ and $m=1,\ldots,M$.\footnote{
       We can relax the exogeneity assumption $\frac 1 {JT} \, {\rm Tr}\left( X_k e' \right) = o_p(1)$.
       For all endogenous regressor the corresponding regression coefficients $\beta^{\rm end}$
       need to be included in the parameter vector $\alpha$, see the replacement \eqref{ReplacementEND}
       above.
    }
      \item[(iv)] $\displaystyle \min_{\lambda \in \mathbbm{R}^{J \times R}}
          \left\{ \mu_{K+M} \left[  \ft 1 {JT} (x,z)' (\mathbbm{1}_T \otimes M_{(\lambda,\lambda^0)}) (x,z) \right]
            \right\} \geq b$, wpa1, for some $b>0$.
      \item[(v)] There exists $b>0$ such that wpa1
   for all $\alpha \in {\cal B}_\alpha$ and $\beta \in \mathbbm{R}^K$
   \begin{align*}
       & \left[ \ft 1 {JT} \Delta \xi_{\alpha,\beta}' \, (x,z) \right]
        \left[ \ft 1 {JT} (x,z)' (x,z) \right]^{-1}
        \left[ \ft 1 {JT} (x,z)' \, \Delta \xi_{\alpha,\beta}  \right]
        \nonumber \\ & \qquad \qquad \qquad
     -  \max_{\lambda \in \mathbbm{R}^{J \times R}} \left[
         \ft 1 {JT} \Delta \xi_{\alpha,\beta}'
     \left( \mathbbm{1}_T \otimes P_{(\lambda,\lambda^0)}\right)
       \Delta \xi_{\alpha,\beta} \right]
           \geq b \| \alpha-\alpha^0\|^2 + b \|\beta-\beta^0\|^2.
   \end{align*}

    \item[(vi)] $\smallW_{JT} \; \operatorname*{\rightarrow}_p \; \smallW >0$.

   \end{itemize}

\end{assumption}

\begin{theorem}[\bf Consistency]
   \label{th:consistency}
   Let Assumption~\ref{ass:con} hold,
   and let $\alpha^0 \in {\cal B}_{\alpha}$.
   In the limit $J,T \rightarrow \infty$
   we then have $\widehat \alpha = \alpha^0 + o_p(1)$, and
           $\widehat \beta = \beta^0 + o_p(1)$.
\end{theorem}

The proof of Theorem~\ref{th:consistency} is given in the appendix.
The similarity between Assumption~\ref{ass:con} and
Assumption~\ref{ass:ID} is obvious, so that for the most part we can refer to
Section~\ref{sec:discussIDcon} for the interpretation of these assumptions, and
in the following we focus on discussing the differences between the consistency
and identification assumptions. The one additional assumption is the last one,
which requires existence of a positive definite probability limit of the weight matrix
$\smallW_{JT}$.

Apart from a rescaling with appropriate powers of
$JT$, the Assumptions~\ref{ass:con}$(i)$, $(iii)$, $(iv)$, and $(v)$
are almost exact sample analogs of their identification counterparts in
Assumption~\ref{ass:ID}.
The two main differences are that  assumption $(i)$ also imposes a Lipschitz-like
continuity condition on $\delta(\alpha)$ around $\alpha^0$, and that
the right hand-side of the inequality in assumption $(v)$ is not just zero, but a quadratic form
in $(\alpha-\alpha^0)$ and $(\beta-\beta^0)$ --- the latter is needed, because
expressions which are exactly zero in the identification proof
are now only converging to zero asymptotically.

Assumption~\ref{ass:con}$(ii)$ imposes
a bound on the the spectral norm of $e$, which
is satisfied as long as $e_{jt}$ has mean zero, has a uniformly bounded fourth moment
(across $j,t,J,T$) and is weakly correlated across $j$ and $t$.\footnote{
Such a statement on the spectral norm of a random matrix is a typical result in
random matrix theory. The difficulty -- and the reason why we prefer
such a high-level assumption on the spectral norm of $e$ --
 is to specify the meaning of
``weakly correlated across $j$ and $t$''. The extreme case is obviously
independence across $j$ and $t$, but weaker assumptions are possible.
We refer to the discussion in Moon and Weidner \cite*{MoonWeidner2015} for
other examples.} The assumption is therefore the analog of Assumption~\ref{ass:ID}$(ii)$.

At finite $J$, $T$, a sufficient condition for existence of $b>0$ such that the inequality
in Assumption~\ref{ass:con}$(iv)$ is satisfied,
is ${\rm rank}(\Xi) > 2R$ for any non-zero linear combination $\Xi$ of $X_k$ and $Z_m$.
This rank condition rules out product-invariant and market-invariant product characteristics $X_k$ and instruments $Z_m$, since those have rank 1 and can be absorbed into the factor structure.\footnote{
  Inclusion of product-invariant and market-invariant characteristics
  (``low-rank regressors'') does not hamper the identification
  and estimation of the regression coefficients on the other
  (``high-rank'') regressors. This is because including low-rank
  regressors is equivalent to increasing the number of factors $R$, and then
  imposing restrictions on the factors and factors loadings of these new
  factors. Conditions under which the coefficients of low-rank regressors
  can be estimated consistently are discussed in
  Moon and Weidner \cite*{MoonWeidner2015}.
}
There are many reformulations of this rank condition, but in one formulation or another this rank condition can be found in any of the above cited papers on linear factor regressions, and we refer to Bai~\cite*{Bai2009}, and Moon and Weidner \cite*{MoonWeidner2015} for a further discussion.

Next, we present results on the
limiting distribution of $\widehat \alpha$ and $\widehat \beta$. 
Some further regularity condition are necessary to derive the
limiting distribution of our LS-MD estimator, and those are summarized in
Assumption~\ref{ass:A2} to \ref{ass:A4} in the appendix. These assumptions
are straightforward generalization of the assumptions imposed by
Moon and Weidner~\cite*{MoonWeidner2015,MoonWeidner2015b} for the linear model, except for
part $(i)$ of Assumption~\ref{ass:A4}, which demands that
$\delta(\alpha)$ can be linearly approximated around $\alpha^0$ such that
the Frobenius norm of the remainder term of the expansion is of order
$o_p(\sqrt{JT}\|\alpha-\alpha^0\|)$
in any $\sqrt{J}$ shrinking neighborhood of $\alpha^0$.
Notice also that Assumption~\ref{ass:A4}(iv) implies $\mathbbm{E}(e_{jt} | X_{jt},Z_{jt} )=0$, while so
far we only required $e_{jt}$ to be uncorrelated with $X_{jt}$ and $Z_{jt}$.

\begin{theorem}
   \label{th:asympt_dist}
   Let Assumptions~\ref{ass:con}, \ref{ass:A2}, \ref{ass:A3}
   and \ref{ass:A4} be satisfied, and
   let $\alpha^0$ be an interior point of ${\cal B}_\alpha$.
   In the limit $J,T \rightarrow \infty$ with $J/T \rightarrow \kappa^2$, $0<\kappa<\infty$,
   we then have
   \begin{align*}
      \sqrt{JT} \left( \begin{array}{c} \widehat \alpha - \alpha^0 \\[2mm]
                                        \widehat \beta  - \beta^0 \end{array} \right)
          \; \; \operatorname*{\rightarrow}_d \; \;
          {\cal N}\left(  \kappa B_0 + \kappa^{-1} B_1 + \kappa B_2,
         \; \; \left(G  \bigW  G' \right)^{-1}  G  \bigW  \Omega
          \bigW G'  \left(G  \bigW  G' \right)^{-1} \right) \; ,
   \end{align*}
 with the formulas for 
   $G$, ${\cal W}$, $\Omega$,
   $B_0$, $B_1$ and $B_2$ given in the appendix
   \ref{app:DetailsTh}.
\end{theorem}
The proof of Theorem~\ref{th:asympt_dist} is provided in the appendix.
Analogous to the least squares estimator in the linear model with interactive fixed effects, there are three
bias terms in the limiting distribution of the LS-MD estimator.
The bias term $\kappa B_0$ is only present if regressors or instruments are pre-determined, that is, if $X_{jt}$ or $Z_{jt}$ are correlated with $e_{j\tau}$
for $t>\tau$ (but not for $t=\tau$, since this would violate weak exogeneity).
A reasonable interpretation of this bias terms thus requires that the index $t$ refers
to time, or has some other well-defined ordering. The other two bias terms
$\kappa^{-1} B_1$ and $\kappa B_2$ are due to heteroscedasticity of the idiosyncratic
error $e_{jt}$ across firms $j$ and markets $t$, respectively. The first and last bias terms are proportional to $\kappa$, and thus are large when
$T$ is small compared to $J$, while the second bias terms is proportional
to $\kappa^{-1}$, and thus is large when $T$ is large compared
to $J$. Note that no asymptotic bias is present if regressors and instruments
are strictly exogenous and errors $e_{jt}$ are homoscedastic.
There is also no asymptotic bias when $R=0$, since then there are no
incidental parameters.
For a more detailed discussion of the asymptotic bias, we again refer to
Bai~\cite*{Bai2009} and Moon and Weidner~\cite*{MoonWeidner2015}.

While the structure of the asymptotic bias terms is analogous to the bias encountered
in linear models with interactive fixed effects, we find that
the structure of the asymptotic
variance matrix for $\widehat \alpha$ and $\widehat \beta$ is analogous to the GMM variance
matrix. The LS-MD estimator
can be shown to be equivalent to the GMM estimator if no factors are present.
In that case the weight matrix
$\bigW$ that appears in Theorem~\ref{th:asympt_dist} can be shown
to be the probability limit of the GMM weight matrix that is implicit in
our LS-MD approach and, thus, our asymptotic variance matrix exactly coincides with the one for GMM (see also Appendix~\ref{app:GMM}).
If factors are present, there is no GMM analog of our estimator, but the only change
in the structure of the asymptotic variance matrix is the appearance of the projectors
$M_{f^0}$ and $M_{\lambda^0}$ in the formulas for $G$, $\Omega$ and
$\bigW$. The presence of these projectors implies that those components
of $X_k$ and $Z_m$ which are proportional to $f^0$ and $\lambda^0$ do not contribute
to the asymptotic variance, that is, do not help in the estimation of $\widehat \alpha$
and $\widehat \beta$. This is again analogous the standard fixed effect setup in
panel data, where time-invariant components do not contribute
to the identification of the regression coefficients.

Using the explicit expressions for the
asymptotic bias and variance of the LS-MD estimator, one
can provide estimators for this asymptotic bias and variance.
By replacing
the true parameter values ($\alpha^0$, $\beta^0$, $\lambda^0$, $f^0$) by
the estimated parameters ($\widehat \alpha$, $\widehat \beta$, $\widehat \lambda$, $\widehat f$),
the error term ($e$) by the residuals ($\widehat e$), and population values
by sample values it is easy
to define estimators $\widehat B_0$, $\widehat B_1$, $\widehat B_2$,
$\widehat G$, $\widehat \Omega$ and $\bigWest$
for $B_0$, $B_1$, $B_2$, $G$, $\Omega$ and $\bigW$.
This is done explicitly in appendix \ref{app:BiasVarEst}.

\begin{theorem}
   \label{th:biascorrection}
   Let the assumptions of Theorem \ref{th:asympt_dist} and Assumption \ref{ass:A5}
   be satisfied. In the limit $J,T \rightarrow \infty$ with $J/T \rightarrow \kappa^2$, $0<\kappa<\infty$
   we then have $\widehat B_1 = B_1 + o_p(1)$, $\widehat B_2 = B_2 + o_p(1)$,
   $\widehat G = G + o_p(1)$,  $\widehat \Omega = \Omega + o_p(1)$
   and $\bigWest = \bigW + o_p(1)$.
   If in addition
   the bandwidth parameter $h$, which enters in the definition of $\widehat B_0$,
   satisfies $h \rightarrow \infty$ and $h^5/T \rightarrow 0$, then
   we also have $\widehat B_0 = B_0 + o_p(1)$.
\end{theorem}

The proof is again given in the appendix.
Theorem~\ref{th:biascorrection} motivates the introduction of the bias corrected estimator
\begin{align}
   \left( \begin{array}{c} \widehat \alpha^* \\[2mm]
                           \widehat \beta^* \end{array} \right)
            &=   \left( \begin{array}{c} \widehat \alpha \\[2mm]
                           \widehat \beta \end{array}  \right)
                  - \frac 1 T  \, \widehat B_0
                    - \frac 1 J  \, \widehat B_1
                      - \frac 1 T  \, \widehat B_2   \; .
           \label{BCest}
\end{align}
Under the assumptions of Theorem \ref{th:biascorrection} the bias corrected estimator
is asymptotically unbiased, normally distributed, and has asymptotic variance
$\left(G  \bigW  G' \right)^{-1}  G  \bigW  \Omega
          \bigW G'  \left(G  \bigW  G' \right)^{-1}$,
which is consistently estimated by
$\left(\widehat G  \bigWest  \widehat G' \right)^{-1}  \widehat G  \bigWest  \widehat \Omega
          \bigWest \widehat G'  \left(\widehat G  \bigWest  \widehat G' \right)^{-1}$. These results allow inference on $\alpha^0$ and $\beta^0$.

From the standard GMM analysis it is know that the $(K+M) \times (K+M)$ weight
matrix $\bigW$ which minimizes the asymptotic variance is
given by $\bigW=c \, \Omega^{-1}$, where $c$ is an arbitrary scalar.
If the errors $e_{jt}$ are homoscedastic with variance $\sigma_e^2$ we
have
$\Omega=\sigma_e^2 \plim_{J,T \rightarrow \infty}  \frac 1 {JT}
          \left(x^{\lambda f}, z^{\lambda f}\right)'
                \left(x^{\lambda f}, z^{\lambda f}\right)$, 
 with  $x^{\lambda f}$ and $z^{\lambda f}$ defined in Appendix~\ref{app:DetailsTh}.
In this case
it is straightforward to show that the optimal
$\bigW =\sigma_e^2 \, \Omega^{-1}$ is
attained by choosing
\begin{align}
   \smallW_{JT}= \frac 1 {JT} z^{\prime} M_{x^{\lambda f}} z \; .
   \label{OptimalW}
\end{align}
Under homoscedasticity this choice of weight matrix
is optimal in the sense that it minimizes the asymptotic variance of our LS-MD estimator, but nothing is known about the efficiency bound in the presence
of interactive fixed effects, that is, a different alternative estimator could
theoretically have even lower asymptotic variance.

The unobserved factor loading $\lambda^0$ and factor $f^0$ enter into
the definition of $x^{\lambda f}$ and thus also into the optimal $\smallW_{JT}$
in \eqref{OptimalW}. A consistent
estimator for the optimal $\smallW_{JT}$ can be obtained by estimating
$\lambda^0$ and $f^0$ in a first stage LS-MD estimation, using an arbitrary
positive definite weight matrix.

Under heteroscedasticity of $e_{jt}$ there are in general not enough degrees
of freedom in  $\smallW_{JT}$ to attain the optimal $\bigW$.
The reason for this is that we have chosen the
first stage of our estimation procedure to be
an ordinary least squares step, which is optimal under
homoscedasticity but not under heteroscedasticity.
By generalizing the first stage optimization to weighted least squares one
would obtain the additional degrees of freedom to attain the optimal
$\bigW$ also under heteroscedasticity, but in the present paper
we will not consider
this possibility further.

\section{Monte Carlo Simulations}
\label{sec:MC}

We consider a model with only one regressors
$X_{jt}=p_{jt}$, which we refer to as price.
The data generating process for mean utility and price is given by
\begin{align}
   \delta_{jt} &=
         \beta^0 \, p_{jt} \, + \, \lambda^0_j \, f^0_t \, + \, e_{jt} ,
   \nonumber \\
   p_{jt} &= \max\left( 0.2, \;
   1 + \tilde p_{jt} + \lambda^0_j \, f^0_t \right)  ,
   \label{MC_DGP}
\end{align}
where $\lambda^0_j$, $f^0_t$, $e_{jt}$ and $\tilde p_{jt}$ are mutually
independent and are all independent and
identically distributed across $j$ and $t$ as ${\cal N}(0,1)$.
In the data generating process the number of factors is $R=1$.
For the number of factors used in the estimation
procedure, $R_{\rm EST}$, we consider the correctly
specified case $R_{\rm EST}=R=1$, the misspecified
case  $R_{\rm EST}=0$, and the case where the number
of factors is overestimated  $R_{\rm EST}=2$.
We have truncated the data generating process
 for price so that $p_{jt}$ takes no values smaller than
0.2.

The market shares are computed from the mean utilities according
to equation \eqref{choiceprobs} and \eqref{DefShares}, where we
assume a normally distributed random coefficient on price $p_{jt}$,
i.e. $v \sim {\cal N}(0,\alpha^2)$.
We chose the parameters of the model to be
 $\beta^0=-3$ and $\alpha^0=1$. These parameters corresponds
to a distribution of consumer tastes where more than $99\%$ of consumers prefer low prices.

Although the regressors are strictly exogenous with respect to $e_{jt}$, we still
need an instrument to identify $\alpha$. We choose $Z_{jt} = p_{jt}^2$,
the squared price. Thus, the number of instruments is $M=1$.
We justify the choice of squared price as an instrument in
subsection~\ref{sec:mc:verify} by verifying the instrument
relevance Assumption~\ref{ass:con}$(v)$ is satisfied for our simulation design.

Simulation results for three different samples sizes
$J=T=20$, $50$ and $80$, and three different
choices for the number of factors in estimation
$R_{\rm EST}=0$, $1$, and $2$
are presented in Table~\ref{tab:MCnew1}.
We find that the estimators for $\widehat \alpha$
and $\widehat \beta$ to be significantly biased when $R_{\rm EST}=0$
factors are chosen in the estimation.
This is because the factor and factor loading enter into the distribution
of  the regressor $p_{jt}$ and the instrument $Z_{jt}$, which makes them
endogenous with respect to the total unobserved error
$\xi^0_{jt} = \lambda^0_j \, f^0_t \, + \, e_{jt}$, and results in
the estimated model
with $R_{\rm EST}=0$ to be misspecified.
The standard errors of the estimators
are also much larger for $R_{\rm EST}=0$
than for $R_{\rm EST}>0$, since the variation of the total
unobserved error $\xi^0_{jt}$ is larger than the variation of $e_{jt}$,
which is the residual error  after accounting for the factor structure.

\begin{table}[tb]
   \center
   \begin{tabular}{ll|rr|rr|rr}
      & & \multicolumn{2}{|c|}{$R_{\rm EST} = 0$}
          & \multicolumn{2}{|c|}{$R_{\rm EST} = 1$}
          & \multicolumn{2}{|c}{$R_{\rm EST} = 2$}
      \\
      J,T & statistics &  \multicolumn{1}{|c}{$\widehat \alpha$} &
                                   \multicolumn{1}{c|}{$\widehat \beta$} &
                                \multicolumn{1}{|c}{$\widehat \alpha$} &
                                   \multicolumn{1}{c|}{$\widehat \beta$} &
                                \multicolumn{1}{|c}{$\widehat \alpha$} &
                                   \multicolumn{1}{c}{$\widehat \beta$}
      \\ \hline \hline
       20,20  & bias
          &  0.4255 & -0.3314
           &   0.0067& -0.0099
           &  0.0024& -0.0050  \\
              & std
         & 0.1644 &  0.1977
           &  0.0756&  0.0979
           & 0.0815&  0.1086   \\
              & rmse
           & 0.4562 &  0.3858
           &    0.0759&  0.0983
           & 0.0815&  0.1086
      \\ \hline
       50,50  & bias
           &  0.4305& -0.3178
           &  0.0005& -0.0012
           &   0.0022& -0.0024  \\
              & std
           &   0.0899&  0.0984
           &  0.0282&  0.0361
           &  0.0293&  0.0369  \\
              & rmse
           &    0.4398&  0.3326
           &   0.0282&  0.0361
           &   0.0293&  0.0369
      \\ \hline
       80,80  & bias
           &  0.4334& -0.3170
           & -0.0009&  0.0010
           &  0.0003& -0.0003   \\
              & std
           &   0.0686&  0.0731
           &    0.0175&  0.0222
           & 0.0176&  0.0223   \\
              & rmse
           &   0.4388&  0.3253
           &    0.0175&  0.0222
           & 0.0176&  0.0223
    \end{tabular}
    \caption{\label{tab:MCnew1}
           \footnotesize Simulation results for the data generating process
           \eqref{MC_DGP}, using 1000 repetitions.
            We report the bias, standard errors (std), and square roots of the
            mean square errors (rmse) of the LS-MD estimator
            $(\widehat \alpha,\widehat \beta)$. The true number of factors
            in the process is $R=1$, but we use
            $R_{\rm EST}=0,1$, and $2$ in the estimation.
              }
\end{table}

For the correctly specified case $R_{\rm EST}=R=1$
we find the biases of the estimators $\widehat \alpha$
and $\widehat \beta$ to be negligible relative to the standard errors.
For $J=T=20$ the absolute value of the biases is about one tenth
the standard errors,
and the ratio is even smaller for the larger sample sizes.
As the sample size increases from  $J=T=20$
to  $J=T=50$ and  $J=T=80$ one finds the standard error of the estimators
to decrease at the rate $1/\sqrt{JT}$, consistent with our asymptotic theory.

The result for the case $R_{\rm EST}=2$
are very similar to those for $R_{\rm EST}=1$,
that is, overestimating the number of factors does not affect
the estimation quality much in our simulation, at least as long as
$R_{\rm EST}$ is small relative to the sample size $J$, $T$.\footnote{%
In pure factor models consistent inference
procedures on the number of factors are known, e.g.
Bai and Ng \cite*{BaiNg2002}, Harding \cite*{Harding2007}, Onatski \cite*{Onatski2010},
and Ahn and Horenstein~\cite*{AhnHorenstein2013}. 
In our model the number of factor can be estimated by applying those pure factor model techniques
to  the residuals $\widehat \xi = \delta(\widehat \alpha)- \widehat \beta \cdot X$,
where $\widehat \alpha$ and $\widehat \beta$ are LS-MD estimator obtained with $R_{\rm EST} \geq R$.
Showing consistency of this procedure, however, goes beyond the scope of the current paper.
}
The biases for the estimators found for $R_{\rm EST}=2$ are still negligible
and the standard errors are about $10 \%$ larger for $R_{\rm EST}=2$
than for $R_{\rm EST}=1$ at $J=T=20$, and even less than $10 \%$ larger
for the larger sample sizes. The result
that choosing $R_{\rm EST}>R$
has only a small effect on the estimator
 is not covered by the asymptotic theory in this paper, where we assume
$R_{\rm EST}=R$, but is consistent with the analytical results
found in Moon and Weidner~\cite*{MoonWeidner2015b} for
the linear model with interactive fixed effects.

We have chosen a data generating process for our simulation where
regressors and instruments are strictly exogenous (as opposed
to pre-determined) with respect to $e_{jt}$, and where
the error distribution $e_{jt}$ is homoscedastic. According
to our asymptotic theory there is therefore no asymptotic bias in the estimators
$\widehat \alpha$ and $\widehat \beta$, which is consistent with the results in
Table~\ref{tab:MCnew1}. The simulation results
for the bias corrected estimators $\widehat \alpha^*$ and $\widehat \beta^*$
are reported in  Table~\ref{tab:MCnew2} in the appendix, but there is virtually no
effect from bias correction here, that is, the results in Table~\ref{tab:MCnew1}
and Table~\ref{tab:MCnew2} are almost identical.

Table~\ref{tab:MCnew2} also reports the average estimated standard error based on our asymptotic variance
estimator, as well as the empirical size of a nominal $5 \%$ t-test for the hypothesis
that the respective parameter equals its true value. Those are not particularly interesting for $R_{\rm EST}<R$,
where the model is badly misspecified. For $R_{\rm EST} \geq R$
we find that for small sample sizes ($J=T=20$)  our standard
errors underestimate the dispersion of the estimator distributions by around $20 \%$, and the t-test is 
oversized accordingly. For larger sample sizes ($J=T=80$) our standard errors are still a bit too small, but only by
around $5 \%$ or less, thus resulting in empirical sizes quite close to the nominal size.

\subsection{Remarks: Instrument Choice}
\label{sec:mc:verify}

 For the special case where there is only one normally distributed random  coefficient
attached to the regressor $p_{jt}$,
one can write equation \eqref{DefShares} as
\begin{align}
   s_{jt}(\alpha,\delta_t,X_t) &=  \frac{1}    {\sqrt{2 \pi} \alpha}
  \int \, \frac {\exp\left( \delta_{jt} + p_{jt} v \right) }
                          { 1 + \sum_{l=1}^J \, \exp\left( \delta_{lt} + p_{lt} v \right) }  \, \exp\left( - \frac {v^2} {2 \alpha^2} \right) \,
    d v .
   \label{DefSspecial}
\end{align}
For $x\geq 0$ we have the general inequalities
$ 1 \geq (1+x)^{-1} \geq 1 -x $. Applying this to \eqref{DefSspecial}
with $x=\sum_{l=1}^J \, \exp\left( \delta_{lt} + p_{lt} v \right)$
one obtains  $s^{\rm up}_{jt}(\alpha,\delta_t,X_t) \geq  s_{jt}(\alpha,\delta_t,X_t)
   \geq  s^{\rm low}_{jt}(\alpha,\delta_t,X_t)$, where
\begin{align}
    s^{\rm up}_{jt}(\alpha,\delta_t,X_t)
    &=  \frac{1}    {\sqrt{2 \pi} \alpha}
  \int \,  \exp\left( \delta_{jt} + p_{jt} v \right)
                       \, \exp\left( - \frac {v^2} {2 \alpha^2} \right) \,
    d v
  \nonumber \\
     &= \exp\left(  \delta_{jt} + \alpha^2 p_{jt}^2 /2 \right) ,
  \nonumber \\
   s^{\rm low}_{jt}(\alpha,\delta_t,X_t)
     &=   \frac{1}    {\sqrt{2 \pi} \alpha}
  \int \,  \exp\left( \delta_{jt} + p_{jt} v \right)
           \left[1 - \sum_{l=1}^J \, \exp\left( \delta_{lt} + p_{lt} v \right)
           \right]  \, \exp\left( - \frac {v^2} {2 \alpha^2} \right) \,
    d v
  \nonumber \\
    &=  s^{\rm up}_{jt}(\alpha,\delta_t,X_t)
       \bigg[ 1 -
      \underbrace{ \sum_{l=1}^J  \exp\left(  \delta_{lt} + \alpha^2 p_{lt}^2 /2
               + \alpha^2  p_{jt} p_{lt} \right) }_{ = \nu_{jt}(\alpha,\delta_t)}  \bigg] .
    \label{s_up_low}
\end{align}
 Here, the integrals over $v$ that appear in the upper and lower bound
are solvable analytically, so that we obtain convenient expressions for
$s^{\rm up}_{jt}(\alpha,\delta_t,X_t)$ and $s^{\rm low}_{jt}(\alpha,\delta_t,X_t)$.

Consider the specification~\eqref{MC_DGP} for $\beta$
negative and large (in absolute value) relative to $\alpha^2$.
 Then $\delta_{jt}$ is also negative
and large in absolute value, which implies that
the $\nu_{jt}=\nu_{jt}(\alpha,\delta_t)$ defined in \eqref{s_up_low} is small.
For $\nu_{jt} \ll 1$, as here,
the  lower and upper bounds are almost identical, which implies
$s_{jt}(\alpha,\delta_t,X_t) \approx
     \exp\left(  \delta_{jt} + \alpha^2 p_{jt}^2 /2 \right)$,
 where $\approx$ means almost equal under that
approximation. Solving for the mean utility yields
$\delta_{jt}(\alpha,s_t,X_t)  \approx \log    s_{jt}(\alpha,\delta_t,X_t)
     -  \alpha^2 p_{jt}^2 /2$. The difference between
     $\delta_{jt}(\alpha,s_t,X_t)$ and $\delta^0_{jt} = \delta_{jt}(\alpha^0,s_t,X_t)$
can then be approximated by
\begin{align}
     \delta_{jt}(\alpha,s_t,X_t) - \delta^0_{jt}
     \; &\approx  \; - \frac{p_{jt}^2} 2 \; \left[\alpha^2- (\alpha^0)^2 \right] .
     \label{delta_approx}
\end{align}
This shows that whenever the approximation $\nu_{jt} \ll 1$ is justified, then
the squared price $p_{jt}^2$ is a valid instrument to identify $\alpha$.
More precisely, equation \eqref{delta_approx}
implies that the LS-MD estimator with instrument $p_{jt}^2$
is approximately equivalent to the least squares estimator
for the linear model with outcome variable
$Y_{jt} = \beta p_{jt} + \alpha^2 p_{jt}^2 + \lambda_j' f_t + e_{jt}$.
Consistency of this least squared estimator
for $\beta$ and $\alpha^2$ in the presence of
the parameters $\lambda_j$ and $f_t$
is discussed in Bai \cite*{Bai2009} and Moon and Weidner \cite*{MoonWeidner2015}.

We have thus shown that
 $\nu_{jt} \ll 1$ is a sufficient condition for validity of the instrument $p_{jt}^2$.
However,
for the data-generating process
with parameters $\alpha^0=1$ and $\beta^0=-3$ used in
the Monte Carlo simulation
this is not a good approximation ---
when calculating $\nu_{jt}$ in that setup one
typically finds values much larger than one.
Therefore, we next confirm by numerical methods
that $p_{jt}^2$ is also a valid instrument when $\nu_{jt}\ll 1$ does not hold.

\subsubsection*{The Instrument Relevance Condition: Some Numerical Evidence}

We want to verify the instrument relevance Assumption~\ref{ass:con}$(v)$
for the data generating process
\eqref{MC_DGP} in the Monte Carlo Simulations with
parameters $\beta^0=-3$, and $\alpha^0=1$.
For this purpose we define
 \begin{align}
      \rho_{\rm IV}(\alpha,\beta)
      &=
       \frac{ \left[ \ft 1 {JT} \Delta \xi_{\alpha,\beta}' \, (x,z) \right]
        \left[ \ft 1 {JT} (x,z)' (x,z) \right]^{-1}
        \left[ \ft 1 {JT} (x,z)' \, \Delta \xi_{\alpha,\beta}  \right] }
       {   \ft 1 {JT} \Delta \xi_{\alpha,\beta}'  \Delta \xi_{\alpha,\beta}   }  ,
    \nonumber \\
     \rho_{\rm F}(\alpha,\beta)
       &= \frac{  \max_{\lambda \in \mathbbm{R}^{J \times R}} \left[
         \ft 1 {JT} \Delta \xi_{\alpha,\beta}'
     \left( \mathbbm{1}_T \otimes P_{(\lambda,\lambda^0)}\right)
       \Delta \xi_{\alpha,\beta} \right] }
     {   \ft 1 {JT} \Delta \xi_{\alpha,\beta}'  \Delta \xi_{\alpha,\beta}  }  ,
   \nonumber \\
     \Delta \rho(\alpha,\beta) &=   \rho_{\rm IV}(\alpha,\beta)  -  \rho_{\rm F}(\alpha,\beta) .
   \label{DefRho}
\end{align}
$\rho_{\rm IV}(\alpha,\beta)$ is the amount of $\Delta \xi_{\alpha,\beta}$ explained
by the instruments and regressors relative to the total variation of $\Delta \xi_{\alpha,\beta}$,
i.e. the relative explanatory power of the instruments.
$\rho_{\rm F}(\alpha,\beta)$ is the maximum amount of $\Delta \xi_{\alpha,\beta}$
explained
by $R$ factor loadings relative the total variation of $\Delta \xi_{\alpha,\beta}$, i.e.
the relative explanatory power of the factors. 
$\rho_{\rm IV}(\alpha,\beta)$ and $\rho_{\rm F}(\alpha,\beta)$
take values betweens 0 and 1.

The difference between the explanatory
power of the instruments and regressors and the explanatory
power of the factors is
given by
$\Delta \rho(\alpha,\beta)$.
Assumption~\ref{ass:con}$(v)$ requires that $\Delta \rho(\alpha,\beta)>0$
for all $\alpha \in {\cal B}_\alpha$ and $\beta \in \mathbbm{R}^K$.

Figure~\ref{fig:rho_iv} contains plots of
$\rho_{\rm IV}(\alpha,\beta)$,  $\rho_{\rm F}(\alpha,\beta)$
    and $ \Delta \rho(\alpha,\beta)$
    as a function of $\alpha$ and $\beta$
for one particular draw of the data generating process with $J=T=80$.
The sample
size is sufficiently large that for different draws the plots in Figure~\ref{fig:rho_iv}
look essentially identical.\footnote{
The appendix contains additional details on the numerical calculation of $\rho_F(\alpha,\beta)$.
}
 Although the data generating process only contains
one factor, we used $R=2$ factors in the calculation of $\rho_{\rm F}(\alpha,\beta)$ and $ \Delta \rho(\alpha,\beta)$
in Figure~\ref{fig:rho_iv}, in order to verify Assumption~\ref{ass:con}$(v)$
also for the case where the number of factors is overestimated
(denoted $R_{\rm EST}$=2 above) --- since
$\rho_{\rm F}(\alpha,\beta)$ is an increasing function of $R$, we thus
also verify the conditions of $R=1$.

For the given draw and
within the examined parameter range one finds that $\rho_{\rm IV}(\alpha,\beta)$
varies between
$0.69$ and $1.00$,
$\rho_{\rm F}(\alpha,\beta)$ varies between $0.34$
and $0.87$, and $\Delta \rho(\alpha,\beta)$ varies between
$0.03$ and $0.49$, in particular
$\Delta \rho(\alpha,\beta)>0$, which is what we wanted to verify.

The variation in $\Delta \rho(\alpha,\beta)$ in this example is mostly
driven by the variation in $\rho_{\rm F}(\alpha,\beta)$, since
$\rho_{\rm IV}(\alpha,\beta)$ for the most part is quite close to one, that is, the explanatory power of the instruments and regressors is very large. The
analytical approximation above showed that for $\nu_{jt} \ll 1$
the regressor $p_{jt}$ and the instrument $p_{jt}^2$
perfectly predict $\Delta \xi_{\alpha,\beta}$, that is, we have  $\rho_{\rm IV}(\alpha,\beta) \approx 1$
under that approximation. Our numerical result now shows that
$p_{jt}^2$  can be a sufficiently powerful instrument
also outside the validity range of this approximation.

\section{Empirical illustration: estimation of demand for new automobiles}

As an illustration of our procedure, we estimate an aggregate random
coefficients logit model
of demand for new automobiles, modeled after the analysis in BLP (1995).
We compare specifications with and without factors, and with and without
price endogeneity. Throughout, we allow for one
normally-distributed random coefficient, attached to price.\footnote{
In principle, multiple random coefficients could be accommodated in a straightforward manner; as this application is primarily illustrative, we do not consider this here.}

For this empirical illustration, we use the same data as was used in BLP
(1995), which are new automobile sales from 1971-1990.\footnote{
 In such a setting, where we have a single national market
  evolving over time, we can interpret
$\lambda_j$ as (unobserved) national advertising for
  brand $j$, which may be roughly constant across time,
and  $f_t$ represents the effectiveness or ``success'' of the advertising,
which varies over time.   Indeed, for the automobile sector (which is the
topic of our empirical example), the dollar amount of national brand-level
advertising does not
vary much across years, but the success of the ad campaign does vary.
}
However, our
estimation procedure requires a balanced panel for the principal components
step.   Since there is substantial entry and exit of individual car models,
we aggregate up to manufacturer-size level, and assume that consumers
choose between aggregate composites of cars.\footnote{
This resembles the treatment in Esteban and Shum's
\citeyear{esteban_shum1} empirical study of the new
and used car markets, which likewise required a balanced panel.
}   Furthermore, we also reduce
our sample window to the sixteen years 1973-1988.  In Table \ref{tab:summary_stats}, we list the
23 car aggregates employed in our analysis, along with the across-year
averages of the variables.

Except from the aggregation, our variables are the same as in BLP.
Market share is given by total sales divided by the number of households in that year.
Price is measured in \$1000 of 1983/84 dollars. Our unit for ``horse
power over weight'' (hp/weight)
is 100 times horse power over pound. ``Miles per dollar'' (mpd) is obtained from
miles per gallons divided by real price per gallon, and measured in miles
over 1983/84 dollars. Size is given by length times width, and measured in
$10^{-4} \, {\rm inch}^2$.

We construct instruments using the idea of Berry \cite*{Berry1994}.
The instruments for a particular aggregated model and year are given by the
averages of hp/weight, mpd and size, over all cars produced by different
manufactures in the same year.
As the weight matrix in the second step of the LS-MD procedure
we use
$\smallW_{JT}= \frac 1 {JT} z^{\prime} M_{x} z$,
which is the optimal weight matrix under homoscedasticity
of $e_{jt}$ and for $R=0$.\footnote{
 We do not change the weight matrix when
estimating specifications with $R=1$, because we do not want
differences in the results for different
values of $R$ to be attributed to the change in
$\smallW_{JT}$.

We include a constant regressor in the model,
although this is a ``low-rank'' regressor, which is
ruled out by our identification and consistency
assumptions. However, as discussed in a footnote above
the inclusion of a low-rank regressor does not hamper
the identification
and estimation of the regression coefficients of the other
  (``high-rank'') regressors.
One certainly wants to include a constant regressor
when estimating the model with no factors ($R=0$),
so to make results easily comparable we include it
in all our model specifications.}

\paragraph{Results.}
Table \ref{tab:main_results} contains estimation results from four
specifications of the model.
In specification~A, prices are considered exogenous (wrt $e_{jt}$), but
one factor is present, which captures some degree of price endogeneity
(wrt. $\xi_{jt}$). Specification~B also contains one factor,
but treats prices as endogenous, even conditional on the factor.
Specification~C corresponds to the BLP (1995) model, where
prices are endogenous, but no factor is present.
Finally, in specification~D, we treat prices as
exogenous, and do not allow for a factor. This final specification is clearly
unrealistic, but is included for comparison with the other
specifications.
In
table \ref{tab:main_results} we report the bias corrected LS-MD estimator
(this only makes a difference for specification A and B), which
accounts for bias due to heteroscedasticity in the error terms,
and due to pre-determined regressors
(we choose bandwidth $h=2$ in the construction of $\widehat B_0$).
The estimation results without bias correction are reported
in table \ref{tab:no_bias_correction}. It turns out,
that it makes not much difference, whether
the LS-MD estimator, or its bias corrected version are used. The t-values
of the bias corrected estimators are somewhat larger, but apart from
the constant, which is insignificant anyways, the bias correction changes
neither the sign of the coefficients nor the conclusion whether the coefficients
are significant at $5\%$ level.

\begin{table}[tb]
   \centering
  \begin{tabular}{l|rl|rl|rl|rl}
                & \multicolumn{8}{|c}{Specifications:}
            \\ \cline{2-9}
                & \multicolumn{2}{|l}{A: \; $R=1$}
                & \multicolumn{2}{|l}{B: \; $R=1$}
                & \multicolumn{2}{|l}{C: \; $R=0$}
                & \multicolumn{2}{|l}{D: \; $R=0$}
           \\
                & \multicolumn{2}{|l}{\; \; \; exogenous p}
                & \multicolumn{2}{|l}{\; \; \; endogenous p}
                & \multicolumn{2}{|l}{\; \; \; endogenous p}
                & \multicolumn{2}{|l}{\; \; \; exogenous p}
             \\ \hline
        price
                  & -4.109 & (-3.568)
                  & -3.842 & (-4.023)
                  & -1.518 & (-0.935)
                  & -0.308 & (-1.299)
        \\
        hp/weight
                  &  0.368  & (1.812)
                  &  0.283  & (1.360)
                  &  -0.481 & (-0.314)
                  &  0.510 & (1.981)
        \\
        mpd
                  & 0.088 & (2.847)
                  & 0.117 & (3.577)
                  & 0.157 & (0.870)
                  & 0.030 & (1.323)
        \\
        size      &  5.448 & (3.644)
                  &  5.404 & (3.786)
                  &  0.446 & (0.324)
                  &  1.154 & (2.471)
        \\
        $\alpha$  & 2.092 & (3.472)
                  & 2.089 & (3.837)
                  & 0.894 & (0.923)
                  & 0.171 & (1.613)
        \\
        const
                  & 3.758 & (1.267)
                  & 0.217 & (0.117)
                  & -3.244 & (-0.575)
                  & -7.827 & (-8.984)

  \end{tabular}

  \caption{\label{tab:main_results}
           \footnotesize Parameter estimates
             (and t-values) for four different model
           specifications (no factor $R=0$ vs. one factor $R=1$;
              exogenous price vs. endogenous price). $\alpha$
              is the standard deviation of the random coefficient distribution
              (only price has a random coefficient), and the regressors
              are p (price),
              hp/weight (horse power per weight),
              mpd (miles per dollar),
              size (car length times car width), and a constant.
              }
\end{table}

In Specification A, most of the coefficients are
precisely estimated.   The price coefficient is -4.109, and the
characteristics coefficients take the expected signs.   The $\alpha$
parameter, corresponding to the standard deviation of the random
coefficient on price, is estimated to be 2.092.   These point estimates
imply that, roughly
97\% of the time, the random price coefficient is negative, which is as we
should expect.

Compared to this baseline, Specification B allows price to be endogenous
(even conditional on the factor). The point estimates for this
specifications are virtually
unchanged from those in Specification A, except for the constant term.
Overall, the estimation results for
the specifications A and B are very similar, and show that once factors
are taken into account
it does not make much difference whether price is treated
as exogenous or endogenous. This suggests that the
factors indeed capture most of the price endogeneity in this application.

In contrast, the estimation results for specifications C and D, which are
the two specifications without any factors, are very
different qualitatively.  The t-values for specification C are rather small
(i.e. standard errors are large), so that the difference in the coefficient
estimates in these two specifications are not actually statistically significant. However, the
differences in the t-values themselves shows that it makes a substantial
difference for the no-factor estimation results whether price is treated as
exogenous or endogenous.

Specifically, in Specification C, the key price
coefficient and $\alpha$ are substantially smaller in magnitude;
furthermore, the standard errors are large, so that none of the estimates
are significant at usual significance levels.  Moreover, the coefficient on
{\em hp/weight} is negative, which is puzzling.
In Specification D, which corresponds to
a BLP model, but without price endogeneity, we see that the price
coefficient is reduced dramatically relative to the other specifications,
down to -0.308.

\paragraph{Elasticities.}
The sizeable differences in the magnitudes of the price coefficients
across the specification with and without factors suggest that these
models may imply economically meaningful differences in price
elasticities.   For this reason, we compute the matrices of
own- and cross-price elasticities for Specifications B (see Table~\eqref{tab:elasticities}) 
and C (see Table~\eqref{tab:elasticitiesBLP}).
The elasticities
were computed using the data in 1988, the final year of our sample.
Comparing these two sets of elasticities, the most obvious difference is
that the elasticities -- both own- and
cross-price -- for Specification C, corresponding to the standard BLP model
without factors, are substantially smaller (about one-half in magnitude)
than the Specification B elasticities.   For instance, reading down the
first column of Table  \eqref{tab:elasticities}, we see that a one-percent
increase in the price of a small
Chevrolet car would result in a 28\% reduction in its market share, but
increase the market share for large Chevrolet cars by 1.5\%.
For the results in Table \eqref{tab:elasticitiesBLP}, however, this same
one-percent price increase would reduce the market share for small
Chevrolet cars by only 13\%, and increase the market share for large
Chevrolet cars by less than half a percent.

On the whole, then, this empirical illustration shows that our estimation
procedure is feasible
even for moderate-sized datasets like the one used here.  Including interactive
fixed effects delivers results which are strikingly different than those obtained
from specifications without these fixed
effects.

\section{Conclusion}
In this paper, we consider an extension of the popular
 BLP
random coefficients discrete-choice demand model, which
  underlies much recent empirical work in IO. We add
 interactive fixed effects in the form of a factor structure on the
 unobserved product characteristics.   The interactive fixed effects can be
 arbitrarily correlated with the observed product characteristics
 (including price), which accommodate endogeneity and, at the same time,
captures strong persistence in market shares across products and
markets.  We propose a two-step least squares-minimum distance (LS-MD)
procedure to calculate the estimator.
Our estimator is easy to compute, and Monte Carlo
simulations show that it performs well.

The model in this paper is, to our knowledge, the first application of
factor-modeling to a nonlinear setting with endogenous regressors.
Since many other models used in applied settings
(such as duration models in labor economics, and parametric
auction models in IO) have these features,
we believe that factor-modeling may prove an effective way of
controlling for unobserved heterogeneity in these models.   We are
 exploring these applications in ongoing work.

\newpage

\begin{appendix}

\section{Additional Tables and Figures}

\begin{figure}[h!]
   \centering
   \epsfig{file=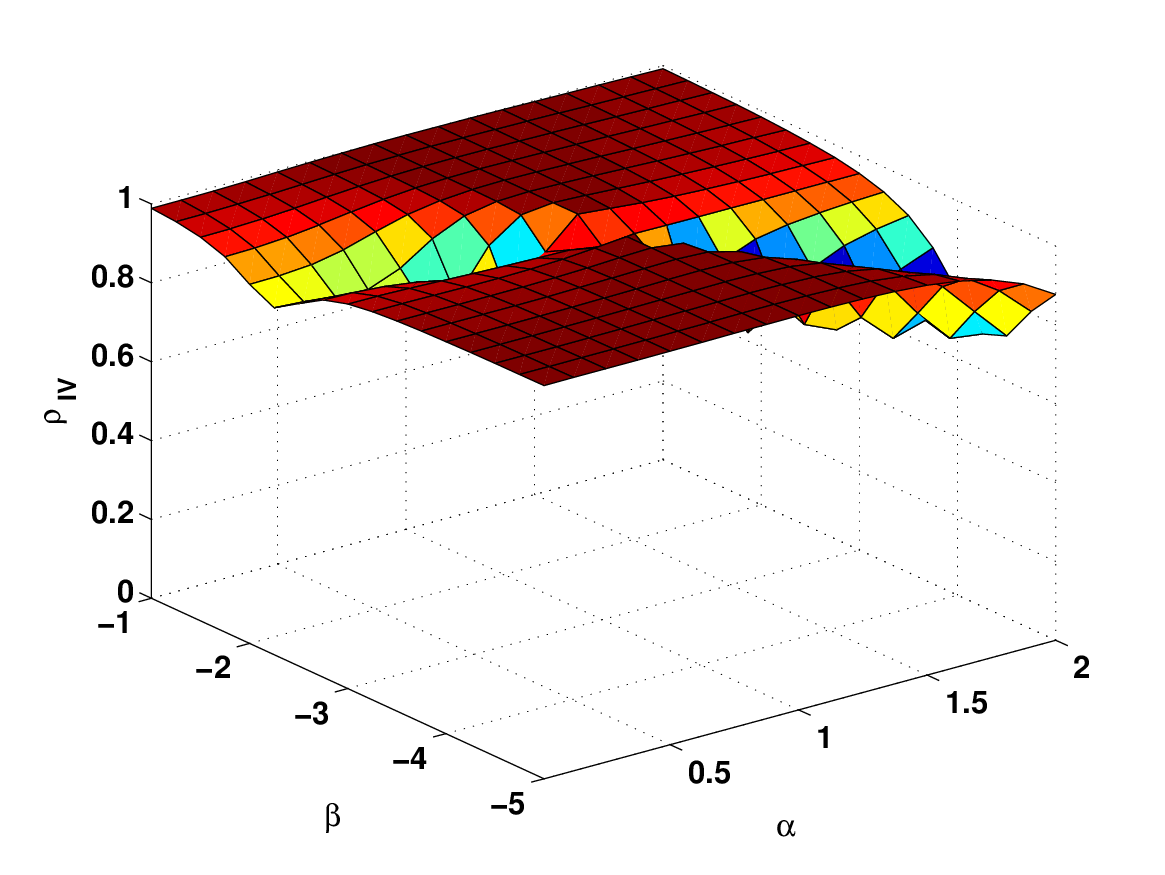,width=0.5\textwidth}
   \epsfig{file=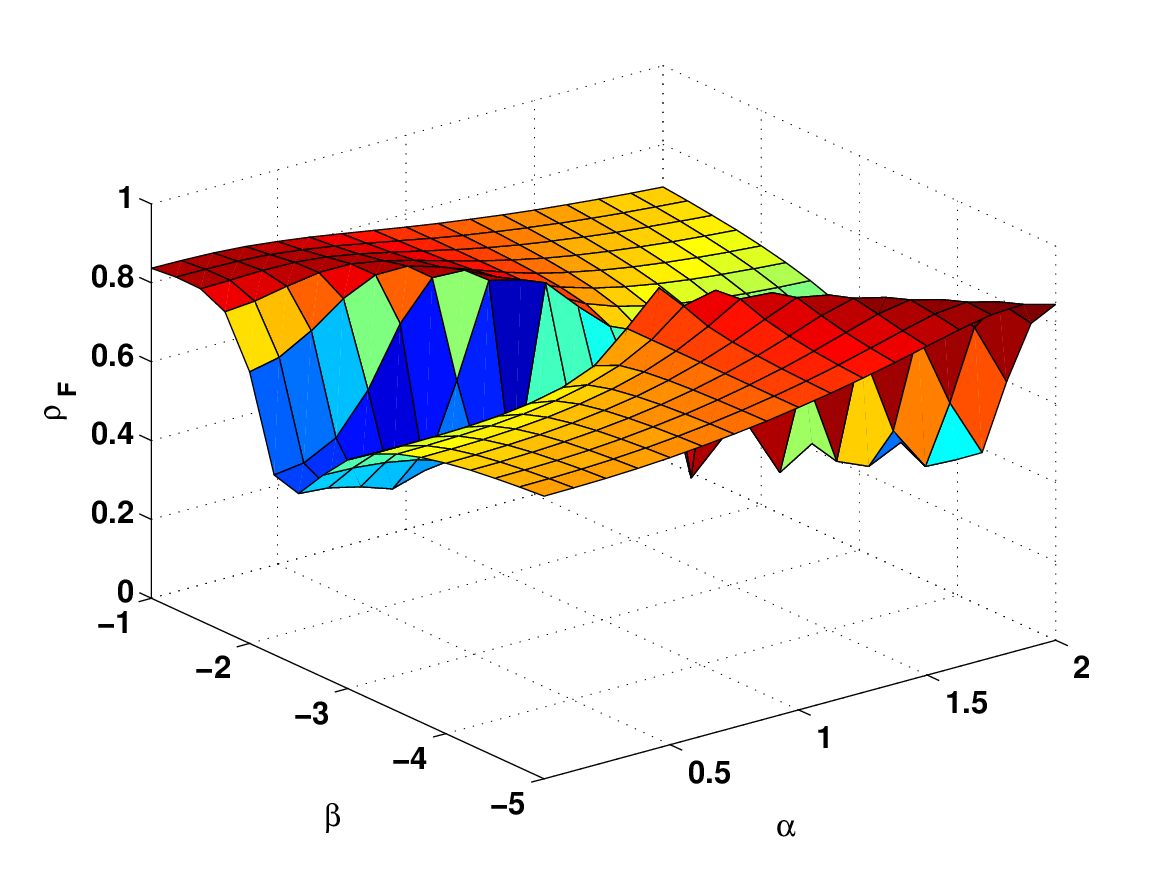,width=0.5\textwidth}
   \epsfig{file=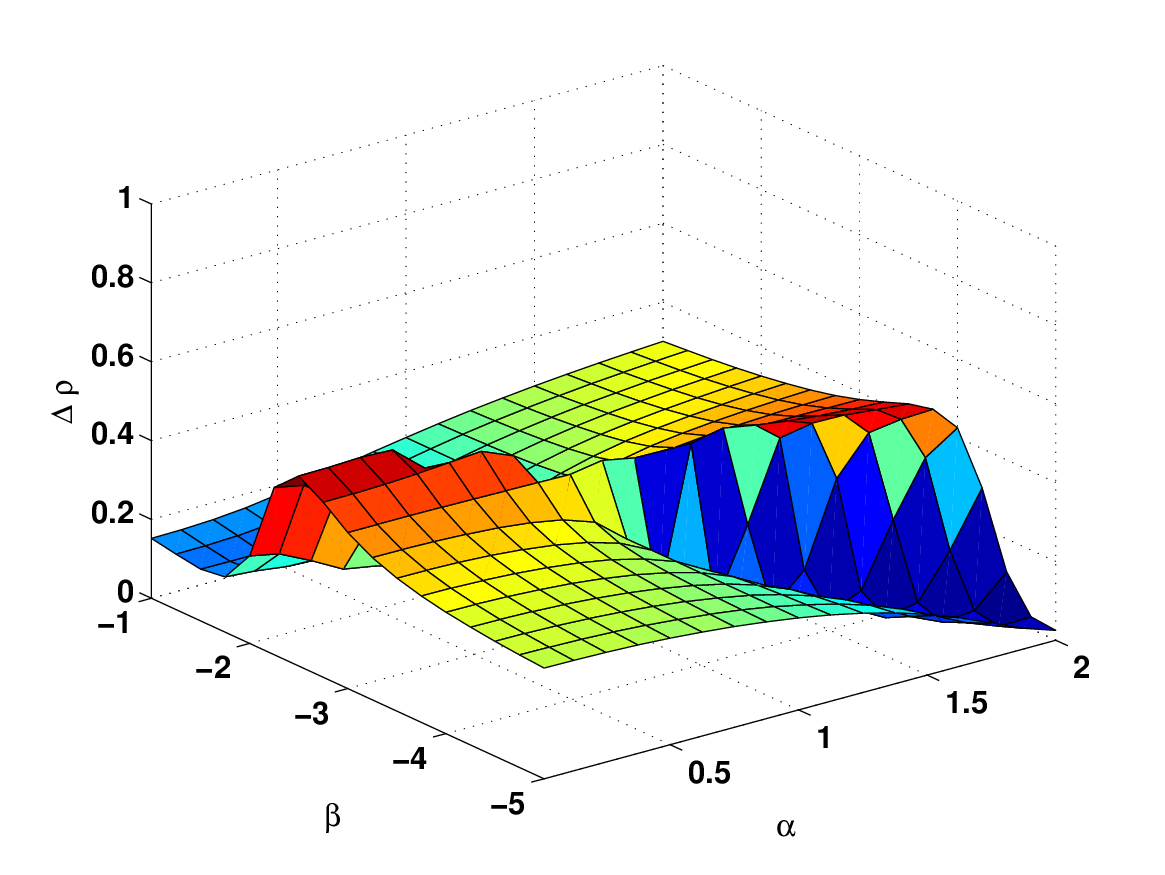,width=0.5\textwidth}
  \caption{\label{fig:rho_iv}
           \footnotesize
     For one draw of the data generating process used in the Monte Carlo design
     with $J=T=80$
    we plot $\rho_{\rm IV}(\alpha,\beta)$,  $\rho_{\rm F}(\alpha,\beta)$
    and $ \Delta \rho(\alpha,\beta)$ defined in \eqref{DefRho}
    as a function of $\alpha$ and $\beta$.
    The number of factors used in the calculation of    $\rho_{\rm F}(\alpha,\beta)$
    is $R=2$, although only one factor is present in the data generating process.
    }
  \vspace{-3cm}
\end{figure}

\begin{table}[h]  
   \center
   \begin{tabular}{ll|rr|rr|rr}
      & & \multicolumn{2}{|c|}{$R_{\rm EST} = 0$}
          & \multicolumn{2}{|c|}{$R_{\rm EST} = 1$}
          & \multicolumn{2}{|c}{$R_{\rm EST} = 2$}
      \\
      J,T & statistics &  \multicolumn{1}{|c}{$\widehat \alpha^*$} &
                                   \multicolumn{1}{c|}{$\widehat \beta^*$} &
                                \multicolumn{1}{|c}{$\widehat \alpha^*$} &
                                   \multicolumn{1}{c|}{$\widehat \beta^*$} &
                                \multicolumn{1}{|c}{$\widehat \alpha^*$} &
                                   \multicolumn{1}{c}{$\widehat \beta^*$}
      \\ \hline \hline
       20,20  & bias
          &  0.4255 & -0.3314
           &  0.0042& -0.0068
           &   0.0001& -0.0023   \\
              & std
           & 0.1644 &  0.1977
           &  0.0759&  0.0981
           &   0.0818&  0.1085   \\
              & rmse
          &  0.4562 &  0.3858
           &   0.0760&  0.0983
           & 0.0817&  0.1084 \\
               & mean(SE)  &
               0.0938   &  0.1300   &  0.0660   &  0.0870   &  0.0632   &  0.0833   \\
              & emp. size &
               0.96   &   0.65   &   0.09   &   0.07   &   0.15   &   0.12   
      \\ \hline
       50,50  & bias
           &   0.4305& -0.3178
           &  0.0000& -0.0006
           &  0.0017& -0.0018    \\
              & std
           &   0.0899&  0.0984
           &  0.0283&  0.0362
           & 0.0293&  0.0368   \\
              & rmse
           &   0.4398&  0.3326
           &  0.0282&  0.0361
           & 0.0293&  0.0369   \\
               & mean(SE)  &
               0.0418   &  0.0551   &  0.0270   &  0.0344   &  0.0265   &  0.0338   \\
              & emp. size &
               1.00   &   0.99   &   0.07   &   0.06   &   0.10   &   0.10   

      \\ \hline
       80,80  & bias
           &   0.4334& -0.3170
           &   -0.0012&  0.0012
           & 0.0001& 0.0000    \\
              & std
           &  0.0686&  0.0731
           &  0.0175&  0.0222
           &  0.0176&  0.0223    \\
              & rmse
           &  0.4388 &  0.3253
           &   0.0175&  0.0222
           & 0.0176&  0.0223 \\
               & mean(SE)  &
               0.0272   &  0.0354   &  0.0171   &  0.0215   &  0.0169   &  0.0213    \\
              & emp. size &
                1.00   &   1.00   &   0.07   &   0.07   &   0.06   &   0.06   
 
    \end{tabular}
    \caption{\label{tab:MCnew2}
           \footnotesize Simulation results for the data generating process
           \eqref{MC_DGP}, using 1000 repetitions.
            We report the bias, standard errors (std), square roots of the
            mean square errors (rmse),
            and the average of the estimated standard error (mean SE)
             of the bias corrected LS-MD estimator
            $(\widehat \alpha^*,\widehat \beta^*)$. 
            In addition, we report the empirical size of a nominal size $5 \%$ t-test
            based on $\widehat \alpha^*$ and $\widehat \beta^*$
            for the hypothesis that the parameter equals its true value.
                        The true number of factors
            in the process is $R=1$, but we use
            $R_{\rm EST}=0,1$, and $2$ in the estimation.
              }
\end{table}

\begin{table}[h]
   \centering
  \begin{tabular}{l|rl|rl}
                & \multicolumn{4}{|c}{Specifications:}
            \\ \cline{2-5}
                & \multicolumn{2}{|l}{A: \; $R=1$}
                & \multicolumn{2}{|l}{B: \; $R=1$}
           \\
                & \multicolumn{2}{|l}{\; \; \; exogenous p}
                & \multicolumn{2}{|l}{\; \; \; endogenous p}
             \\ \hline
        price
                  & -3.112 & (-2.703)
                  & -2.943 & (-3.082)
        \\
        hp/weight
                  &  0.340  & (1.671)
                  &  0.248  & (1.190)
        \\
        mpd
                  & 0.102 & (3.308)
                  & 0.119 & (3.658)
        \\
        size      &   4.568 & (3.055)
                  &  4.505 & (3.156)
        \\
        $\alpha$  & 1.613 & (2.678)
                  & 1.633 & (3.000)
        \\
        const
                  & -0.690 & (-0.232)
                  & -2.984 & (-1.615)
  \end{tabular}

  \caption{\label{tab:no_bias_correction}
           \footnotesize Parameter estimates
             (and t-values) for model
           specification A and B.
           Here we report the LS-MD estimators without bias correction,
           while in table \ref{tab:main_results} we report the
           bias corrected LS-MD estimators.
                         }
\end{table}

\begin{landscape}

\begin{table}[tb]
   \centering
  \begin{tabular}{l||l|l|l|cc|ccc}
    Product\#   & Make & Size Class & Manuf. & Mkt Share \% & Price &
    hp/weight & mpd & size \\
 & & & & (avg) & (avg) & (avg) & (avg) & (avg)
\\\hline\hline
1  & CV (Chevrolet) & small  & GM & 1.39 & 6.8004 & 3.4812 & 20.8172 & 1.2560 \\
2  & CV & large  & GM & 0.49 & 8.4843 & 3.5816 & 15.9629 & 1.5841 \\
3  & OD (Oldsmobile) & small  & GM & 0.25 & 7.6786 & 3.4789 & 19.1946 & 1.3334 \\
4  & OD & large  & GM & 0.69 & 9.7551 & 3.6610 & 15.7762 & 1.5932 \\
5  & PT (Pontiac) & small  & GM & 0.46 & 7.2211 & 3.4751 & 19.3714 & 1.3219 \\
\hline
6  & PT & large  & GM & 0.31 & 8.6504 & 3.5806 & 16.6192 & 1.5686  \\
7  & BK (Buick) & all & GM & 0.84 & 9.2023 & 3.6234 & 16.9960 & 1.5049     \\
8  & CD (Cadillac) & all & GM & 0.29 & 18.4098 & 3.8196 & 13.6894 & 1.5911    \\
9  & FD (Ford) & small  & Ford & 1.05 & 6.3448 & 3.4894 & 21.7885 & 1.2270\\
10 & FD & large  & Ford & 0.63 & 8.9530 & 3.4779 & 15.7585 & 1.6040\\
\hline
11 & MC (Mercury) & small  & Ford & 0.19 & 6.5581 & 3.6141 & 22.2242 & 1.2599\\
12 & MC & large  & Ford & 0.32 & 9.2583 & 3.4610 & 15.9818 & 1.6053\\
13 & LC (Lincoln) & all & Ford & 0.16 & 18.8322 & 3.7309 & 13.6460 & 1.7390  \\
14 & PL (Plymouth) & small  & Chry & 0.31 & 6.2209 & 3.5620 & 22.7818 & 1.1981\\
15 & PL & large  & Chry & 0.17 & 7.7203 & 3.2334 & 15.4870 & 1.5743\\
\hline
16 & DG (Dodge) & small  & Chry & 0.35 & 6.5219 & 3.6047 & 23.2592 & 1.2031 \\
17 & DG & large  & Chry & 0.17 & 7.8581 & 3.2509 & 15.4847 & 1.5681 \\
18 & TY (Toyota) & all & Other & 0.54 & 7.1355 & 3.7103 & 24.3294 & 1.0826   \\
19 & VW (Volkswagen) & all & Other & 0.17 & 8.2388 & 3.5340 & 24.0027 & 1.0645   \\
20 & DT/NI (Datsen/Nissan) & all & Other & 0.41 & 7.8120 & 4.0226 & 24.5849 & 1.0778\\
\hline
21 & HD (Honda) & all & Other & 0.41 & 6.7534 & 3.5442 & 26.8501 & 1.0012    \\
22 & SB (Subaru) & all & Other & 0.10 & 5.9568 & 3.4718 & 25.9784 & 1.0155    \\
23 & REST & all & Other & 1.02 & 10.4572 & 3.6148 & 19.8136 & 1.2830 \\

  \end{tabular}

  \caption{\label{tab:summary_stats}
           \footnotesize Summary statistics for the 23 product-aggregates
           used in estimation.
              }
\end{table}

\end{landscape}

\begin{landscape}

\begin{table}

\footnotesize
\hspace{-18pt}%
\begin{tabular}{l|r@{\,}r@{\,}r@{\,}r@{\,}r@{\,}r@{\,}r@{\,}r@{\,}r@{\,}r@{\,}r@{\,}r@{\,}r@{\,}r@{\,}r@{\,}r@{\,}r@{\,}r@{\,}r@{\,}r@{\,}r@{\,}r@{\,}r}
& \, CV s& \, CV l& \, OD s& \, OD l& \, PT s& \, PT l& \, BK& \, CD& \, FD s& \, FD l& \, MC s& \, MC l& \, LC& \, PL s& \, PL l& \, DG s& \, DG l& \, TY& \, VW& DT/NI& \, HD& \, SB& \, REST\\ \hline
CV s & -28.07 & 0.82 & 0.70 & 1.70 & 0.96 & 0.31 & 2.77 & 0.14 & 1.32 & 2.38 & 0.41 & 1.45 & 0.03 & 0.32 & 0.22 & 0.44 & 0.31 & 1.57 & 0.57 & 1.74 & 0.91 & 0.15 & 6.58\\
CV l & 1.50 & -34.54 & 0.72 & 2.02 & 0.79 & 0.21 & 3.27 & 0.73 & 0.97 & 3.54 & 0.37 & 2.15 & 0.16 & 0.21 & 0.21 & 0.30 & 0.30 & 1.21 & 0.40 & 1.62 & 0.71 & 0.10 & 10.17\\
OD s & 1.29 & 0.72 & -35.78 & 2.08 & 0.72 & 0.18 & 3.36 & 1.15 & 0.84 & 3.90 & 0.35 & 2.37 & 0.25 & 0.17 & 0.20 & 0.25 & 0.28 & 1.06 & 0.34 & 1.53 & 0.63 & 0.08 & 11.35\\
OD l & 0.98 & 0.64 & 0.65 & -35.80 & 0.59 & 0.13 & 3.37 & 2.09 & 0.64 & 4.34 & 0.30 & 2.63 & 0.45 & 0.12 & 0.17 & 0.18 & 0.25 & 0.84 & 0.25 & 1.36 & 0.51 & 0.06 & 12.86\\
PT s & 1.76 & 0.80 & 0.72 & 1.90 & -32.51 & 0.26 & 3.09 & 0.38 & 1.14 & 3.02 & 0.39 & 1.84 & 0.08 & 0.26 & 0.22 & 0.37 & 0.31 & 1.39 & 0.48 & 1.70 & 0.81 & 0.12 & 8.56\\
PT l & 2.17 & 0.81 & 0.68 & 1.55 & 0.98 & -26.85 & 2.53 & 0.06 & 1.40 & 1.97 & 0.41 & 1.21 & 0.01 & 0.35 & 0.22 & 0.48 & 0.31 & 1.65 & 0.61 & 1.72 & 0.94 & 0.16 & 5.37\\
BK & 0.99 & 0.64 & 0.66 & 2.09 & 0.60 & 0.13 & -34.47 & 2.04 & 0.65 & 4.33 & 0.30 & 2.62 & 0.44 & 0.12 & 0.18 & 0.18 & 0.25 & 0.84 & 0.25 & 1.36 & 0.51 & 0.06 & 12.81\\
CD & 0.00 & 0.01 & 0.01 & 0.08 & 0.00 & 0.00 & 0.12 & -6.97 & 0.00 & 0.36 & 0.00 & 0.21 & 3.67 & 0.00 & 0.00 & 0.00 & 0.00 & 0.00 & 0.00 & 0.02 & 0.00 & 0.00 & 1.19\\
FD s & 2.03 & 0.82 & 0.71 & 1.71 & 0.95 & 0.31 & 2.79 & 0.15 & -28.99 & 2.41 & 0.41 & 1.47 & 0.03 & 0.32 & 0.22 & 0.44 & 0.31 & 1.56 & 0.57 & 1.74 & 0.90 & 0.15 & 6.67\\
FD l & 0.61 & 0.50 & 0.55 & 1.95 & 0.42 & 0.07 & 3.13 & 4.23 & 0.40 & -34.69 & 0.23 & 2.80 & 0.90 & 0.06 & 0.14 & 0.10 & 0.20 & 0.56 & 0.15 & 1.07 & 0.34 & 0.04 & 14.05\\
MC s & 1.57 & 0.77 & 0.72 & 1.99 & 0.81 & 0.22 & 3.24 & 0.63 & 1.02 & 3.41 & -34.49 & 2.07 & 0.14 & 0.23 & 0.21 & 0.32 & 0.30 & 1.26 & 0.42 & 1.64 & 0.74 & 0.11 & 9.77\\
MC l & 0.62 & 0.50 & 0.55 & 1.95 & 0.43 & 0.07 & 3.14 & 4.15 & 0.41 & 4.64 & 0.23 & -36.50 & 0.88 & 0.06 & 0.14 & 0.11 & 0.20 & 0.56 & 0.15 & 1.08 & 0.35 & 0.04 & 14.03\\
LC & 0.00 & 0.01 & 0.02 & 0.09 & 0.01 & 0.00 & 0.15 & 20.15 & 0.00 & 0.41 & 0.00 & 0.24 & -23.81 & 0.00 & 0.00 & 0.00 & 0.00 & 0.00 & 0.00 & 0.02 & 0.00 & 0.00 & 1.39\\
PL s & 2.21 & 0.79 & 0.64 & 1.40 & 0.98 & 0.34 & 2.29 & 0.03 & 1.42 & 1.65 & 0.40 & 1.01 & 0.01 & -23.54 & 0.21 & 0.49 & 0.30 & 1.66 & 0.63 & 1.67 & 0.95 & 0.16 & 4.42\\
PL l & 1.47 & 0.75 & 0.72 & 2.03 & 0.78 & 0.21 & 3.29 & 0.78 & 0.96 & 3.59 & 0.37 & 2.18 & 0.17 & 0.21 & -35.26 & 0.30 & 0.29 & 1.19 & 0.39 & 1.61 & 0.70 & 0.10 & 10.33\\
DG s & 2.17 & 0.81 & 0.68 & 1.55 & 0.98 & 0.33 & 2.54 & 0.06 & 1.40 & 1.99 & 0.41 & 1.22 & 0.01 & 0.35 & 0.22 & -26.80 & 0.31 & 1.64 & 0.61 & 1.72 & 0.94 & 0.16 & 5.41\\
DG l & 1.47 & 0.75 & 0.72 & 2.03 & 0.78 & 0.21 & 3.29 & 0.78 & 0.96 & 3.59 & 0.37 & 2.18 & 0.17 & 0.21 & 0.20 & 0.30 & -35.18 & 1.19 & 0.39 & 1.61 & 0.70 & 0.10 & 10.33\\
TY & 1.94 & 0.81 & 0.72 & 1.79 & 0.93 & 0.29 & 2.91 & 0.22 & 1.25 & 2.65 & 0.41 & 1.62 & 0.05 & 0.30 & 0.22 & 0.41 & 0.31 & -30.16 & 0.54 & 1.73 & 0.87 & 0.14 & 7.41\\
VW & 2.13 & 0.82 & 0.69 & 1.61 & 0.97 & 0.32 & 2.63 & 0.09 & 1.37 & 2.13 & 0.41 & 1.31 & 0.02 & 0.34 & 0.22 & 0.47 & 0.31 & 1.62 & -27.86 & 1.73 & 0.93 & 0.15 & 5.85\\
DT/NI & 1.49 & 0.76 & 0.72 & 2.02 & 0.79 & 0.21 & 3.28 & 0.75 & 0.97 & 3.55 & 0.37 & 2.16 & 0.16 & 0.21 & 0.21 & 0.30 & 0.29 & 1.20 & 0.40 & -33.74 & 0.71 & 0.10 & 10.22\\
HD & 1.88 & 0.81 & 0.72 & 1.83 & 0.91 & 0.28 & 2.97 & 0.26 & 1.22 & 2.77 & 0.40 & 1.69 & 0.06 & 0.29 & 0.22 & 0.40 & 0.31 & 1.47 & 0.52 & 1.72 & -31.39 & 0.13 & 7.77\\
SB & 2.16 & 0.82 & 0.68 & 1.58 & 0.98 & 0.33 & 2.57 & 0.07 & 1.39 & 2.04 & 0.41 & 1.25 & 0.02 & 0.35 & 0.22 & 0.47 & 0.31 & 1.64 & 0.61 & 1.73 & 0.94 & -27.60 & 5.58\\
REST & 0.56 & 0.47 & 0.53 & 1.91 & 0.40 & 0.07 & 3.07 & 4.71 & 0.37 & 4.65 & 0.22 & 2.80 & 1.00 & 0.06 & 0.13 & 0.09 & 0.19 & 0.51 & 0.13 & 1.02 & 0.32 & 0.03 & -25.42\\
\end{tabular}

  \caption{\label{tab:elasticities}
           \footnotesize
    Estimated price elasticities for specification B in $t=1988$.
    Rows ($i$) correspond to market shares ($s_{jt}$), and columns ($j$)
    correspond to prices ($p_{jt}$)
    with respect to which elasticities are calculated.}
\end{table}

\end{landscape}

\newpage

\begin{landscape}

\begin{table}

\footnotesize
\hspace{-18pt}%
\begin{tabular}{l|r@{\,}r@{\,}r@{\,}r@{\,}r@{\,}r@{\,}r@{\,}r@{\,}r@{\,}r@{\,}r@{\,}r@{\,}r@{\,}r@{\,}r@{\,}r@{\,}r@{\,}r@{\,}r@{\,}r@{\,}r@{\,}r@{\,}r}
& \, CV s& \, CV l& \, OD s& \, OD l& \, PT s& \, PT l& \, BK& \, CD& \, FD s& \, FD l& \, MC s& \, MC l& \, LC& \, PL s& \, PL l& \, DG s& \, DG l& \, TY& \, VW& DT/NI& \, HD& \, SB& \, REST\\ \hline
CV s & -12.95 & 0.46 & 0.46 & 0.48 & 0.46 & 0.47 & 0.48 & 1.45 & 0.46 & 0.51 & 0.46 & 0.51 & 1.41 & 0.48 & 0.46 & 0.47 & 0.46 & 0.46 & 0.46 & 0.46 & 0.46 & 0.47 & 0.51\\
CV l & 0.43 & -15.20 & 0.49 & 0.53 & 0.45 & 0.41 & 0.53 & 2.46 & 0.43 & 0.60 & 0.46 & 0.59 & 2.39 & 0.40 & 0.47 & 0.41 & 0.47 & 0.43 & 0.42 & 0.47 & 0.44 & 0.41 & 0.61\\
OD s & 0.41 & 0.47 & -15.79 & 0.53 & 0.44 & 0.39 & 0.53 & 2.83 & 0.41 & 0.61 & 0.46 & 0.61 & 2.73 & 0.37 & 0.47 & 0.39 & 0.47 & 0.42 & 0.40 & 0.47 & 0.42 & 0.39 & 0.63\\
OD l & 0.38 & 0.45 & 0.48 & -16.57 & 0.41 & 0.35 & 0.53 & 3.40 & 0.38 & 0.63 & 0.44 & 0.63 & 3.28 & 0.33 & 0.45 & 0.35 & 0.45 & 0.39 & 0.36 & 0.45 & 0.40 & 0.36 & 0.65\\
PT s & 0.44 & 0.47 & 0.49 & 0.51 & -14.32 & 0.44 & 0.51 & 2.01 & 0.45 & 0.56 & 0.47 & 0.56 & 1.95 & 0.44 & 0.47 & 0.44 & 0.47 & 0.45 & 0.44 & 0.47 & 0.45 & 0.44 & 0.58\\
PT l & 0.46 & 0.44 & 0.43 & 0.44 & 0.44 & -11.76 & 0.44 & 1.09 & 0.46 & 0.45 & 0.44 & 0.45 & 1.07 & 0.51 & 0.44 & 0.48 & 0.44 & 0.45 & 0.47 & 0.44 & 0.45 & 0.48 & 0.46\\
BK & 0.38 & 0.45 & 0.48 & 0.53 & 0.42 & 0.35 & -16.54 & 3.38 & 0.38 & 0.63 & 0.44 & 0.63 & 3.26 & 0.33 & 0.45 & 0.35 & 0.45 & 0.39 & 0.36 & 0.45 & 0.40 & 0.36 & 0.65\\
CD & 0.03 & 0.06 & 0.07 & 0.10 & 0.05 & 0.03 & 0.10 & -7.85 & 0.03 & 0.15 & 0.06 & 0.15 & 5.14 & 0.02 & 0.06 & 0.03 & 0.06 & 0.04 & 0.03 & 0.06 & 0.04 & 0.03 & 0.16\\
FD s & 0.46 & 0.46 & 0.47 & 0.48 & 0.46 & 0.47 & 0.48 & 1.48 & -13.03 & 0.51 & 0.46 & 0.51 & 1.44 & 0.48 & 0.46 & 0.47 & 0.46 & 0.46 & 0.46 & 0.46 & 0.46 & 0.47 & 0.52\\
FD l & 0.33 & 0.42 & 0.45 & 0.52 & 0.37 & 0.30 & 0.51 & 4.22 & 0.33 & -17.46 & 0.41 & 0.63 & 4.06 & 0.27 & 0.42 & 0.30 & 0.42 & 0.35 & 0.31 & 0.42 & 0.36 & 0.30 & 0.65\\
MC s & 0.43 & 0.47 & 0.49 & 0.53 & 0.45 & 0.42 & 0.52 & 2.35 & 0.43 & 0.59 & -14.99 & 0.59 & 2.28 & 0.41 & 0.47 & 0.42 & 0.47 & 0.44 & 0.42 & 0.47 & 0.44 & 0.42 & 0.60\\
MC l & 0.33 & 0.42 & 0.45 & 0.52 & 0.38 & 0.30 & 0.52 & 4.20 & 0.33 & 0.63 & 0.41 & -17.44 & 4.03 & 0.27 & 0.42 & 0.30 & 0.42 & 0.35 & 0.31 & 0.42 & 0.36 & 0.30 & 0.66\\
LC & 0.04 & 0.07 & 0.08 & 0.11 & 0.05 & 0.03 & 0.10 & 5.75 & 0.04 & 0.16 & 0.06 & 0.16 & -8.59 & 0.02 & 0.07 & 0.03 & 0.07 & 0.04 & 0.03 & 0.07 & 0.04 & 0.03 & 0.17\\
PL s & 0.45 & 0.40 & 0.40 & 0.39 & 0.42 & 0.48 & 0.39 & 0.79 & 0.45 & 0.39 & 0.41 & 0.39 & 0.77 & -10.42 & 0.40 & 0.48 & 0.40 & 0.43 & 0.47 & 0.40 & 0.43 & 0.47 & 0.39\\
PL l & 0.42 & 0.47 & 0.49 & 0.53 & 0.45 & 0.41 & 0.53 & 2.51 & 0.42 & 0.60 & 0.46 & 0.60 & 2.43 & 0.40 & -15.28 & 0.41 & 0.47 & 0.43 & 0.41 & 0.47 & 0.44 & 0.41 & 0.61\\
DG s & 0.46 & 0.44 & 0.44 & 0.44 & 0.44 & 0.48 & 0.44 & 1.10 & 0.46 & 0.45 & 0.44 & 0.45 & 1.08 & 0.51 & 0.44 & -11.80 & 0.44 & 0.45 & 0.47 & 0.44 & 0.45 & 0.48 & 0.46\\
DG l & 0.42 & 0.47 & 0.49 & 0.53 & 0.45 & 0.41 & 0.53 & 2.51 & 0.42 & 0.60 & 0.46 & 0.60 & 2.43 & 0.40 & 0.47 & 0.41 & -15.28 & 0.43 & 0.41 & 0.47 & 0.44 & 0.41 & 0.61\\
TY & 0.46 & 0.47 & 0.48 & 0.49 & 0.46 & 0.46 & 0.49 & 1.69 & 0.46 & 0.53 & 0.46 & 0.53 & 1.64 & 0.46 & 0.47 & 0.46 & 0.47 & -13.58 & 0.46 & 0.47 & 0.46 & 0.46 & 0.54\\
VW & 0.46 & 0.45 & 0.45 & 0.46 & 0.45 & 0.48 & 0.45 & 1.24 & 0.46 & 0.48 & 0.45 & 0.47 & 1.21 & 0.50 & 0.45 & 0.48 & 0.45 & 0.46 & -12.28 & 0.45 & 0.45 & 0.47 & 0.48\\
DT/NI & 0.42 & 0.47 & 0.49 & 0.53 & 0.45 & 0.41 & 0.53 & 2.48 & 0.43 & 0.60 & 0.46 & 0.60 & 2.40 & 0.40 & 0.47 & 0.41 & 0.47 & 0.43 & 0.42 & -15.22 & 0.44 & 0.41 & 0.61\\
HD & 0.45 & 0.47 & 0.48 & 0.50 & 0.46 & 0.45 & 0.50 & 1.79 & 0.45 & 0.54 & 0.47 & 0.54 & 1.74 & 0.46 & 0.47 & 0.45 & 0.47 & 0.45 & 0.45 & 0.47 & -13.83 & 0.45 & 0.55\\
SB & 0.46 & 0.44 & 0.44 & 0.45 & 0.45 & 0.48 & 0.45 & 1.15 & 0.46 & 0.46 & 0.44 & 0.46 & 1.13 & 0.50 & 0.44 & 0.48 & 0.44 & 0.45 & 0.47 & 0.44 & 0.45 & -12.00 & 0.47\\
REST & 0.32 & 0.41 & 0.45 & 0.51 & 0.37 & 0.29 & 0.51 & 4.37 & 0.32 & 0.63 & 0.40 & 0.63 & 4.19 & 0.26 & 0.41 & 0.29 & 0.41 & 0.34 & 0.30 & 0.41 & 0.35 & 0.29 & -17.59\\
\end{tabular}

  \caption{\label{tab:elasticitiesBLP}
           \footnotesize
    Estimated price elasticities for specification C (BLP case) in $t=1988$.
    Rows ($i$) correspond to market shares ($s_{jt}$), and columns ($j$)
    correspond to prices ($p_{jt}$)
    with respect to which elasticities are calculated.}
\end{table}

\end{landscape}

\newpage

\section{Alternative GMM approach}
\label{app:GMM}

In this section we show that in the presence of factors a moment based
estimation approach along the lines originally proposed by BLP
is inadequate.
The moment conditions imposed by the model are
\begin{align}
   \mathbb{E} \left[  e_{jt}\left(\alpha^0,\, \beta^0,\, \lambda^0 f^{0\prime} \right) X_{k,jt} \right] &= 0  \; ,
        &   k &= 1,\ldots,K \; ,
     \nonumber \\
   \mathbb{E} \left[  e_{jt}\left(\alpha^0,\, \beta^0,\, \lambda^0 f^{0\prime} \right) Z_{m,jt} \right] &= 0  \; ,
        &   m &= 1,\ldots,M \; ,
\end{align}
where $e_{jt}(\alpha,\, \beta,\, \lambda f') = \delta_{jt}(\alpha,\, s_t,\, X_t)
              -  \sum_{k=1}^K \, \beta_k \, X_{k,jt} -  \sum_{r=1}^R \, \lambda_{jr} f_{tr}$.
Note that we write the residuals $e_{jt}$ as a function of the $J\times T$ matrix $\lambda f'$
in order to avoid the ambiguity of the decomposition into $\lambda$ and $f$.
The corresponding sample moments read
\begin{align}
   m^X_k(\alpha,\, \beta,\, \lambda f')  &=
     \frac 1 {JT} \, {\rm Tr} \left(
         e(\alpha,\, \beta,\, \lambda f') \, X'_k \right) \; ,
   \nonumber \\
   m^Z_m(\alpha,\, \beta,\, \lambda f')  &= \frac 1 {JT} \, {\rm Tr} \left(
      e(\alpha,\, \beta,\, \lambda f') \, Z'_m \right) \; .
\end{align}
We also define the sample moment vectors
$m^X(\alpha,\, \beta,\, \lambda f')=\left( m^X_1, \ldots, m^X_K \right)'$
and $m^Z(\alpha,\, \beta,\, \lambda f')=\left( m^Z_1, \ldots, m^Z_M \right)'$.
An alternative estimator for
$\alpha$, $\beta$, $\lambda$ and $f$
is then given by\footnote{
              The minimizing $\hat \lambda_{\alpha,\beta}$ and $\hat f_{\alpha,\beta}$
are the least squares estimators, or equivalently, the principal components estimators, e.g. $\hat \lambda_{\alpha,\beta}$
consists of the eigenvectors corresponding to the $R$ largest eigenvalues
of the $J\times J$ matrix
\begin{align*}
 \left(  \delta(\alpha,\, s,\, X) -  \sum_{k=1}^K \, \beta_k \, X_{k} \right)
 \left(  \delta(\alpha,\, s,\, X) -  \sum_{k=1}^K \, \beta_k \, X_{k} \right)' \; .
\end{align*}
}
\begin{align}
   \left( \hat \lambda_{\alpha,\beta} \, , \; \hat f_{\alpha,\beta} \right)
     &= \argmin_{ \{ \lambda, \, f \} } \, \sum_{j=1}^J \, \sum_{t=1}^T \,
                e^2_{jt}(\alpha,\, \beta,\, \lambda f') \; .
 \nonumber \\
   \left( \hat \alpha^{\rm GMM}, \, \hat \beta^{\rm GMM} \right)
      &= \argmin_{ \{ \alpha \in {\cal B}_\alpha, \, \beta \} }
            \left( \begin{array}{c} m^X(\alpha,\, \beta,\, \hat \lambda_{\alpha,\beta}
                                                           \hat f_{\alpha,\beta}')
            \\ m^Z(\alpha,\, \beta,\, \hat \lambda_{\alpha,\beta}
                                                           \hat f_{\alpha,\beta}') \end{array} \right)'
            \; \bigW_{JT} \;
            \left( \begin{array}{c} m^X(\alpha,\, \beta,\, \hat \lambda_{\alpha,\beta}
                                                           \hat f_{\alpha,\beta}') \\
                                    m^Z(\alpha,\, \beta,\, \hat \lambda_{\alpha,\beta}
                                                           \hat f_{\alpha,\beta}') \end{array} \right) \; ,
   \label{estimatorGMM}
\end{align}
where $\bigW_{JT}$
is a positive definite $(K+M) \times (K+M)$ weight matrix.
  The main difference between this alternative estimator
  and our estimator (\ref{estimator}) is that the least-squares step
  is used solely to recover estimates of the factors and factor
  loadings (principal components estimator),
  while the structural parameters $(\alpha,\beta)$ are
  estimated in the GMM second step.
The relation between $\hat \alpha$ and $\hat \beta$ defined in
\eqref{estimator} and $\hat \alpha^{\rm GMM}$ and $\hat \beta^{\rm GMM}$ defined
in \eqref{estimatorGMM} is as follows
\begin{itemize}
   \item[(i)] Let $R=0$ (no factors) and set
          \begin{align}
             \bigW_{JT} &=
        \left( \begin{array}{cc} \left(\frac 1 {JT} x'x \right)^{-1}& 0_{K\times M} \\
                              0_{M\times K} & 0_{M\times M} \end{array}  \right)
      + \left(\begin{array}{c} -  (x'x)^{-1} x' \, z \\ \mathbbm{1}_M \end{array} \right)
         \left(\frac 1 {JT} z' M_x z \right)^{-1}
         \nonumber \\ & \qquad\ \qquad \qquad\ \qquad\qquad\ \qquad
          \, \smallW_{JT} \,
         \left(\frac 1 {JT} z' M_x z \right)^{-1}
         \left(\begin{array}{c} -  (x'x)^{-1} x' \, z \\ \mathbbm{1}_M \end{array} \right)' \; ,
          \label{RelationW}
          \end{align}
          where
          $x$ is a $JT \times K$ matrix and $z$ is a $JT \times M$ matrix,
          given by $x_{.,k} = {\rm vec}\left( X_{k} \right)$,
          $k=1,\ldots,K$, and
          $z_{.,m} = {\rm vec}\left( Z_{m} \right)$,
          $m=1,\ldots,M$.
         Then $\hat \alpha$ and $\hat \beta$
         solve \eqref{estimator} with weight matrix $\smallW_{JT}$
         if and only if they solve \eqref{estimatorGMM} with this weight matrix
         $\bigW_{JT}$,\footnote{
          With this weight matrix $\bigW_{JT}$ the second
          stage objective function in \eqref{estimatorGMM} becomes
          \begin{align*}
             & \left( d(\alpha) - x \beta \right)' \, x \, (x'x)^{-1} \,
                        x' \,  \left( d(\alpha) - x \beta \right) /JT
              + d'(\alpha) \, M_x  \, z \, (z'M_x z)^{-1} \smallW_{JT} (z'M_x z)^{-1}
              \, z' \, M_x \, d(\alpha)
           \\ &=  \left( d(\alpha) - x \beta \right)'  P_x
            \left( d(\alpha) - x \beta \right)/JT
                   +  \tilde \gamma'_\alpha \, \smallW_{JT} \, \tilde \gamma_\alpha \; ,
          \end{align*}
          where $d(\alpha)={\rm vec}(\delta(\alpha,\, s,\, X)-\delta(\alpha^0,\, s,\, X))$.
          Here, $\beta$ only appears in the first term, and by choosing
          $\beta=\hat \beta = (x' x)^{-1} x' d(\alpha)$
          this term becomes zero.
          Thus, we are left
          with the second term, which is exactly the second stage objective function
          in \eqref{estimator} in this case, since for $R=0$
          by the Frisch-Waugh theorem
          we have $\tilde \gamma_\alpha = (z'M_x z)^{-1}
              \, z' \, M_x \, d(\alpha)$.
          }         i.e. in this case we have
         $(\hat \alpha,\hat \beta) = (\hat \alpha^{\rm GMM},\hat \beta^{\rm GMM})$.
   \item[(ii)] Let $R>0$ and $M=L$ (exactly identified case). Then a solution of \eqref{estimator}
         also is a solution of \eqref{estimatorGMM}, but not every solution
         of \eqref{estimatorGMM} needs to be a solution of \eqref{estimator}.
   \item[(iii)] For $M>L$ and $R>0$ there is no straightforward
     characterization of the relationship between
          the estimators in \eqref{estimator} and \eqref{estimatorGMM}.
\end{itemize}

We want to discuss the exactly identified case $M=L$ a bit further. The reason why in this
case every solution of \eqref{estimator} also solves \eqref{estimatorGMM}
is that the first order conditions (FOC's)
wrt to $\beta$ and $\gamma$ of the first stage optimization
in \eqref{estimator} read
$m^X(\hat \alpha,\, \hat \beta,\, \hat \lambda_{\hat \alpha,\hat \beta} \hat f_{\hat \alpha,\hat\beta}')=0$
and
$m^Z(\hat \alpha,\, \hat \beta,\, \hat \lambda_{\hat \alpha,\hat \beta} \hat f_{\hat \alpha,\hat\beta}')=0$, which implies that the GMM objective function of \eqref{estimatorGMM}
is zero, i.e. minimized.
The reverse statement is not true, because for $R>0$
the first stage objective function
in \eqref{estimator} is not a quadratic function of $\beta$ and $\gamma$
anymore once one concentrates out $\lambda$ and $f$, and it can have multiple local
minima that satisfy the FOC.
Therefore, $\hat \alpha^{\rm GMM}$ and $\hat \beta^{\rm GMM}$ can be
inconsistent, while $\hat \alpha$ and $\hat \beta$ are consistent,
which is the main reason to consider the latter in this paper.

To illustrate this important difference between $\hat \alpha^{\rm GMM}$, $\hat \beta^{\rm GMM}$  and $\hat \alpha$, $\hat \beta$, we want to give a simple
example for a linear model in which the least squares objective function has multiple local
minima. Consider a DGP where $Y_{jt} = \beta^0 X_{jt} + \lambda^0_j f^0_t + e_{jt}$,
with $X_{jt} = 1 + 0.5 \tilde X_{jt} + \lambda^0_j f^0_t$, and 
$\tilde X_{jt}$, $e_{jt}$, $\lambda^0_j$ and $f^0_t$ are all identically distributed
as ${\cal N}(0,1)$, mutually independent, and independent across $j$ and $t$.
Here, the number of factors $R=1$, and we
assume that $Y_{jt}$ and $X_{jt}$ are observed and that $\beta^0=0$. The profiled least
squares objective function in this model, which corresponds to our inner loop,
is given by $L(\beta) = \sum_{r=2}^{T} \mu_r \left[ (Y-\beta X)' (Y-\beta X) \right]$.
For $J=T=100$ and a concrete draw of $Y$ and $X$, this objective function is plotted
in figure~\ref{fig:ambiguity}. The shape of this objective function is
qualitatively unchanged for other draws of $Y$ and $X$, or larger values of $J$ and $T$. As predicted by our consistency result, the global minimum of $L(\beta)$
is close to $\beta^0=0$, but another local minimum is present, which does neither vanish nor converge to $\beta^0=0$
when $J$ and $T$ grow to infinity. Thus, the global minimum of $L(\beta)$
gives a consistent estimator, but the solution to the FOC
$\partial L(\beta) / \partial \beta =0$ gives not. In this example, the
principal components estimator of $\lambda(\beta)$ and $f(\beta)$,
which are derived from
$Y-\beta X$, become very bad approximations for $\lambda^0$ and $f^0$ for
$\beta\gtrsim 0.5$. Thus, for $\beta\gtrsim 0.5$, the fixed effects are
essentially not controlled for anymore in the objective function, and the local minimum
around $\beta \approx 0.8$ reflects the resulting endogeneity problem.

\begin{figure}[tbh]
   \centering
   \epsfig{file=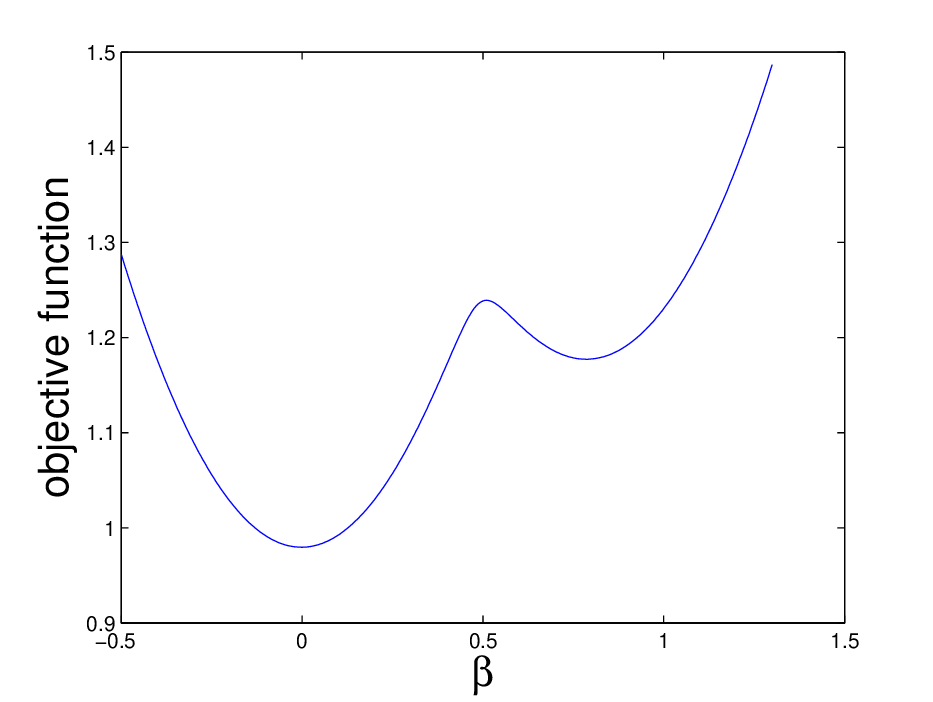,width=0.5\textwidth}
  \caption{\label{fig:ambiguity}
           \footnotesize Example for multiple local minima in the 
           least squares objective function $L(\beta)$. The global minimum can be found close to the true value
                  $\beta^0=0$, but another local minimum exists around
                  $\beta \approx 0.8$, which renders the FOC inappropriate for
                  defining the estimator $\hat \beta$.}  
\end{figure}

\section{Details for Theorems~\ref{th:asympt_dist}
   and \ref{th:biascorrection} }
\label{app:DetailsTh}
\subsection{Formulas for Asymptotic Variance Terms}

We define the $JT \times K$ matrix $x^{\lambda f}$,
the $JT \times M$ matrix $z^{\lambda f}$,
and the $JT \times L$ matrix $g$ by
\begin{align}
   x^{\lambda f}_{.,k} &= {\rm vec}\left( M_{\lambda^0} X_{k} M_{f^0} \right) \; , &
   z^{\lambda f}_{.,m} &= {\rm vec}\left( M_{\lambda^0} Z_{m} M_{f^0} \right) \; , &
   g_{.,l} &= - {\rm vec}\left( \nabla_l \, \delta(\alpha^0) \right) \; , &
\end{align}
where $k=1,\ldots,K$, $m=1,\ldots,M$, and $l=1,\ldots,L$.
Note that $x^{\lambda f} = (\mathbbm{1}_T \otimes M_{\lambda^0}) x^f$,
$z^{\lambda f} = (\mathbbm{1}_T \otimes M_{\lambda^0}) z^f$, and
$g$ is the vectorization of the gradient of $\delta(\alpha)$, evaluated at the
true parameter.
We introduce the  $(L+K)\times(L+K)$ matrix $G$
and the $(K+M)\times(K+M)$ matrix $\Omega$ as follows
\begin{align}
        G & = \hspace{-2pt}
          \plim_{J,T \rightarrow \infty}
              \frac 1 {JT}
               \left( \begin{array}{cc} g'  x^{\lambda f} & g' z^{\lambda f} \\
                x^{\lambda f \prime}  x^{\lambda f} &
                    x^{\lambda f \prime} z^{\lambda f} \end{array} \right)  ,
     &
   \Omega &= \hspace{-2pt}\plim_{J,T \rightarrow \infty}  \frac 1 {JT}
          \left(x^{\lambda f}, z^{\lambda f}\right)'
          {\rm diag}(\Sigma^{\rm vec}_e)
                \left(x^{\lambda f}, z^{\lambda f}\right),
   \label{DefGOmega}
\end{align}
where $ \Sigma_e^{\rm vec} = {\rm vec}\left\{ \left[ \mathbbm{E} \left(e_{jt}^2\right)
           \right]_{\begin{minipage}{1.2cm} \tiny
            j=1,\ldots,J\\t=1,\ldots,T\end{minipage}}  \right\}$ is
the $JT$-vector of vectorized variances of $e_{jt}$.
Finally, we define the  $(K+M) \times (K+M)$ weight matrix $\bigW$ by
\begin{align}
                \bigW &=  \plim_{J,T \rightarrow \infty} \Bigg[
        \left( \begin{array}{cc} \left( \frac 1 {JT} x^{\lambda f \prime}
                                           x^{\lambda f} \right)^{-1} & 0_{K\times M} \\
                              0_{M\times K} & 0_{M\times M} \end{array}  \right)
  + \left(\begin{array}{c} -  (x^{\lambda f \prime} x^{\lambda f})^{-1} x^{\lambda f \prime} \, z^{\lambda f} \\ \mathbbm{1}_M \end{array} \right)
      \nonumber \\ & \quad
       \times  \left(\frac 1 {JT} z^{\lambda f \prime} M_{x^{\lambda f}}
       z^{\lambda f } \right)^{-1}
         \, \smallW_{JT} \,
         \left(\frac 1 {JT} z^{\lambda f  \prime} M_{x^{\lambda f}}
             z^{\lambda f } \right)^{-1}
         \left(\begin{array}{c} -  (x^{\lambda f \prime} x^{\lambda f})^{-1} x^{\lambda f \prime} \, z^{\lambda f} \\ \mathbbm{1}_M \end{array} \right)' \Bigg].
      \label{DefBigW}
\end{align}
Existence of these probability limits is imposed by Assumption \ref{ass:A3} below.

\subsection{Formulas for Asymptotic Bias Terms}
\label{app:biasterms}

Here we provide the formulas for the asymptotic bias terms $B_0$, $B_1$
and $B_2$ that enter into Theorem~\ref{th:asympt_dist}.
Let the $J\times 1$ vector
$\Sigma_e^{(1)}$, the $T\times 1$ vector $\Sigma_e^{(2)}$,
and the $T\times T$
matrices $\Sigma^{X,e}_k$, $k=1,\ldots,K$, and $\Sigma^{Z,e}_m$, $m=1,\ldots,M$,
be defined by
\begin{align}
   \Sigma_{e,j}^{(1)} &=
     \frac 1 T \, \sum_{t=1}^T \, \mathbbm{E} \left(e_{jt}^2\right)
      \; ,
  &
   \Sigma_{e,t}^{(2)} &=
      \frac 1 J \, \sum_{j=1}^J \, \mathbbm{E} \left(e_{jt}^2\right)  \; ,
 \nonumber \\
   \Sigma^{X,e}_{k,t\tau} \, &= \, \frac 1 J \, \sum_{j=1}^J \,
                                \mathbbm{E}\left( X_{k,jt} \, e_{j\tau} \right) \; ,
  &
   \Sigma^{Z,e}_{m,t\tau} \, &= \, \frac 1 J \, \sum_{j=1}^J \,
                                \mathbbm{E}\left( Z_{m,jt} \, e_{j\tau} \right) \; ,
\end{align}
where $j=1,\ldots,J$ and
$t,\tau=1,\ldots,T$. Furthermore, let
\begin{align}
        b^{(x,0)}_k &= \plim_{J,T \rightarrow \infty} \,
                      {\rm Tr}\left(  P_{f^0} \, \Sigma^{X,e}_k \right) \; ,
        \nonumber \\
        b^{(x,1)}_k &= \plim_{J,T \rightarrow \infty} \,
       {\rm Tr}\left[ {\rm diag}\left(\Sigma_e^{(1)} \right) \, M_{\lambda^0} \, X_k \,
              f^0 \, (f^{0\prime}f^0)^{-1} \, (\lambda^{0\prime}\lambda^0)^{-1} \, \lambda^{0\prime} \right] \; ,
        \nonumber \\
        b^{(x,2)}_k &= \plim_{J,T \rightarrow \infty} \,
      {\rm Tr}\left[ {\rm diag}\left(\Sigma_e^{(2)} \right) \, M_{f^0} \, X^{\prime}_k \,
              \lambda^0 \, (\lambda^{0\prime}\lambda^0)^{-1} \, (f^{0\prime}f^0)^{-1} \, f^{0\prime} \right] \; ,
        \nonumber \\
        b^{(z,0)}_m &=  \plim_{J,T \rightarrow \infty} \,
               {\rm Tr}\left(  P_{f^0} \, \Sigma^{Z,e}_m \right) \; ,
        \nonumber \\
        b^{(z,1)}_m &= \plim_{J,T \rightarrow \infty} \,
       {\rm Tr}\left[ {\rm diag}\left(\Sigma_e^{(1)} \right) \, M_{\lambda^0} \, Z_m \,
              f^0 \, (f^{0\prime}f^0)^{-1} \, (\lambda^{0\prime}\lambda^0)^{-1} \, \lambda^{0\prime} \right] \; ,
        \nonumber \\
        b^{(z,2)}_m &= \plim_{J,T \rightarrow \infty} \,
      {\rm Tr}\left[ {\rm diag}\left(\Sigma_e^{(2)} \right) \, M_{f^0} \, Z^{\prime}_m \,
              \lambda^0 \, (\lambda^{0\prime}\lambda^0)^{-1} \, (f^{0\prime}f^0)^{-1} \, f^{0\prime} \right] \; ,
    \label{Defbzx012}          
\end{align}
and we set
$b^{(x,i)}=\left(b^{(x,i)}_1,\ldots,b^{(x,i)}_K \right)'$ and
$b^{(z,i)}=\left(b^{(z,i)}_1,\ldots,b^{(z,i)}_M \right)'$,
for $i=0,1,2$.
With these definitions we can now give the expression for the asymptotic bias terms
which appear in Theorem \ref{th:asympt_dist}, namely
\begin{align}
     B_i  &= - \left( G \bigW G' \right)^{-1}
     G \bigW \,
     \left( \begin{array}{c} b^{(x,i)} \\ b^{(z,i)} \end{array} \right)  ,
    \label{DefBias}
\end{align}
where $i=0,1,2$.

\subsection{Additional Assumptions for Asymptotic Distribution and Bias Correction}

In addition to Assumption~\ref{ass:con}, which guarantees consistency
of the LS-MD estimator, we also require the Assumptions~\ref{ass:A2}, \ref{ass:A3}
and \ref{ass:A4} to derive the limiting distribution of the estimator in
Theorem~\ref{th:asympt_dist},
and Assumption~\ref{ass:A5} to provide consistent estimators for
the asymptotic bias and asymptotic covariance matrix in Theorem~\ref{th:biascorrection}.
These additional assumptions are presented below.

\begin{assumption}   $\phantom{a}$
   \label{ass:A2}
    We assume that the limits of $\lambda^{0\prime} \lambda^0 / J$
                 and $f^{0\prime} f^0 / T$ are finite and have full rank, i.e.
  (a)~$\lim_{J,T \rightarrow \infty}\left(\lambda^{0\prime} \lambda^0/J\right) > 0$,
  (b)~$\lim_{J,T \rightarrow \infty} \left( f^{0\prime} f^0 / T \right) > 0$\;.
\end{assumption}

Assumption \ref{ass:A2} guarantees that $\|\lambda^0\|$
and $\|f^0\|$ grow at a rate of $\sqrt{J}$ and $\sqrt{T}$, respectively.
This is a so called ``strong factor'' assumption that makes sure that the influence of the
factors is sufficiently large, so that the principal components estimators
$\widehat \lambda$ and $\widehat f$ can pick up the correct factor loadings and factors.

\begin{assumption}
    \label{ass:A3}
     We assume existence of the probability limits
     $G$, $\Omega$, $\bigW$,
     $b^{(x,i)}$ and $b^{(z,i)}$, $i=0,1,2$.
     In addition, we assume
     $G \bigW G' > 0$ and
     $G \bigW \Omega \bigW G'>0$.
\end{assumption}

\begin{assumption}
    \label{ass:A4}
    $\phantom{a}$
   \begin{itemize}
      \item[(i)] There exist $J\times T$ matrices $r^{\Delta}(\alpha)$
            and $\nabla_l \delta(\alpha^0)$,
            $l=1,\ldots,L$, such that
       \begin{align*}
  \delta(\alpha)-\delta(\alpha^0) &=
               \sum_{l=1}^L (\alpha_l-\alpha_l^0) \,
                \nabla_l \delta(\alpha^0) + r^{\Delta}(\alpha) \; ,
\end{align*}
and
\begin{align*}
   \frac 1 {\sqrt{JT}} \| \nabla_l \delta(\alpha^0) \|_F &= {\cal O}_p(1) \; ,
   && \text{for} \;\; l=1,\ldots,L \; ,
  \nonumber \\
   \sup_{\{\alpha: \, \sqrt{J} \|\alpha-\alpha^0\| < c, \;
                     \alpha \neq \alpha^0 \}}
       \frac{ \frac 1 {\sqrt{JT}} \| r^{\Delta}(\alpha) \|_F }
            {\| \alpha - \alpha^0 \|} &= o_p(1) \; ,
   && \text{for all} \; \; c>0 \; .
\end{align*}

       \item[(ii)] $\|\lambda^0_j\|$ and $\|f^0_t\|$ are uniformly bounded
                 across $j$, $t$, $J$ and $T$.

       \item[(iii)] The errors $e_{jt}$ are independent across $j$ and $t$, they satisfy
                  $\mathbb{E} e_{jt}=0$, and $\mathbb{E} (e_{jt})^{8+\epsilon}$ is bounded
                   uniformly across $j,t$ and $J,T$, for some $\epsilon>0$.
       \item[(iv)] The regressors $X_k$, $k=1,\ldots,K$,
                   (both high- and low rank regressors) and the instruments
                   $Z_m$, $m=1,\ldots,M$,  can be decomposed as
                 $X_k = X^{\rm str}_k + X^{\rm weak}_k$
                 and $Z_m = Z^{\rm str}_m + Z^{\rm weak}_m$.
                 The components $X^{\rm str}_k$ and $Z^{\rm str}_m$ are strictly exogenous,
                 i.e. $X^{\rm str}_{k,jt}$ and $Z^{\rm str}_{m,jt}$ are independent
                 of $e_{j\tau}$ for all $j,i,t,\tau$.
                 The components $X^{\rm weak}_k$ and $Z^{\rm weak}_m$ are weakly exogenous,
                  and we assume
                 \begin{align*}
                    X^{\rm weak}_{k,jt} &= \sum_{\tau=1}^{t-1} \, c_{k,j\tau} \, e_{j,t-\tau} \; ,
                    &
                    Z^{\rm weak}_{m,jt} &= \sum_{\tau=1}^{t-1} \, d_{m,j\tau} \, e_{j,t-\tau} \; ,
                    \label{defXweak}
                 \end{align*}
                 for some coefficients $c_{k,j\tau}$ and $d_{m,j\tau}$ that satisfy
                 \begin{align*}
                    |c_{k,j\tau}| \, &< \, \alpha^{\tau} \; ,
                    &
                    |d_{k,j\tau}| \, &< \, \alpha^{\tau} \; ,
                 \end{align*}
                 where $\alpha \in (0,1)$ is a constant that is independent of $\tau=1,\ldots,T-1$,
                 $j=1\ldots J$, $k=1,\ldots,K$ and $m=1,\ldots,M$.
                 We also assume that $\mathbb{E}[(X^{\rm str}_{k,jt})^{8+\epsilon}]$
                 and $\mathbb{E}[(Z^{\rm str}_{m,jt})^{8+\epsilon}]$ are
                 bounded uniformly over $j,t$ and $J,T$, for some $\epsilon>0$.
   \end{itemize}
\end{assumption}

Assumption \ref{ass:con}$(ii)$ and $(iii)$ are implied by Assumption \ref{ass:A4}, so it would not be
necessary to impose those explicitly in Theorem \ref{th:asympt_dist}.
Part $(ii)$, $(iii)$ and $(iv)$ of Assumption \ref{ass:A4} are identical
to Assumption 5 in Moon and Weidner \cite*{MoonWeidner2015,MoonWeidner2015b}, except for the
appearance of the instruments $Z_m$ here,
which need to be included since they appear as additional regressors in the
first step of our estimation procedure. Part $(i)$ of Assumption
\ref{ass:A4} can for example be justified by assuming that
within any  $\sqrt{J}$-shrinking neighborhood of $\alpha^0$
we have wpa1 that
 $\delta_{jt}(\alpha)$
is differentiable, that $| \nabla_l \delta_{jt}(\alpha) |$
is uniformly bounded across $j$, $t$, $J$ and $T$, and that
$\nabla_l \delta_{jt}(\alpha)$ is Lipschitz continuous with a
Lipschitz constant that is uniformly bounded
across $j$, $t$, $J$ and $T$, for all $l=1,\ldots L$. But since the assumption
is only on the Frobenius norm of the gradient and remainder term,
one can also conceive weaker sufficient conditions for
Assumption \ref{ass:A4}$(i)$.

\begin{assumption}
   \label{ass:A5}
   For all $c>0$ and $l=1,\ldots,L$ we have
   \begin{align*}
      \sup_{\{\alpha: \, \sqrt{JT} \|\alpha-\alpha^0\| < c  \}}
     \| \nabla_l \delta(\alpha) - \nabla_l \delta(\alpha^0) \|_F = o_p(\sqrt{JT}).
   \end{align*}
\end{assumption}
This last assumption is needed to guarantee consistency
of the bias and variance estimators that are presented in the following.

\subsection{Bias and Variance Estimators}
\label{app:BiasVarEst}

Here we present consistent estimators for the matrices
$G$, $\Omega$, and $\bigW$, which enter into the asymptotic variance of
the LS-MD estimator, and for the vectors $B_0$, $B_1$ and $B_2$,
which enter into the asymptotic bias of the estimator.
Consistency of these estimators is stated in Theorem~\ref{th:biascorrection}.

Given the LS-MD estimators $\widehat \alpha$ and $\widehat \beta$, we can define the residuals
\begin{align}
   \widehat e &= \delta(\widehat \alpha,\, s,\, X)
   -  \sum_{k=1}^K \, \widehat \beta_k \, X_{k} - \widehat \lambda \widehat f' \; .
\end{align}
We also define the $JT \times K$ matrix $\widehat x^{\lambda f}$,
the $JT \times M$ matrix $\widehat z^{\lambda f}$,
and the $JT \times L$ matrix $\widehat g$ by
\begin{align}
   \widehat x^{\lambda f}_{.,k} &= {\rm vec}\left( M_{\widehat \lambda} X_{k} M_{\widehat f} \right) \; , &
   \widehat z^{\lambda f}_{.,m} &= {\rm vec}\left( M_{\widehat \lambda} Z_{m} M_{\widehat f} \right) \; , &
   \widehat g_{.,l} &= - {\rm vec}\left( \nabla_l \, \delta(\widehat \alpha) \right) \; , &
\end{align}
where $k=1,\ldots,K$, $m=1,\ldots,M$, and $l=1,\ldots,L$.
The definition of $\widehat \Sigma^{\rm vec}_e$, $\widehat \Sigma_e^{(1)}$ and $\widehat \Sigma_e^{(2)}$
is analogous to that of $\Sigma^{\rm vec}_e$, $\Sigma_e^{(1)}$ and $\Sigma_e^{(2)}$,
but with $\mathbbm{E}(e^2_{jt})$ replaced by $\widehat e^2_{jt}$.
The $T\times T$
matrices $\widehat \Sigma^{X,e}_k$, $k=1,\ldots,K$, and $\widehat \Sigma^{Z,e}_m$, $m=1,\ldots,M$,
are defined by
\begin{align}
   \widehat \Sigma^{X,e}_{k,t\tau} \, &= \,
    \left\{ \begin{array}{ll}
    \frac 1 J \, \sum_{j=1}^J \, X_{k,jt} \, \widehat e_{j\tau}
      & \text{for $0 < t-\tau \leq h$}
     \\
       0 & \text{otherwise}
    \end{array} \right.
  \nonumber \\
   \widehat \Sigma^{Z,e}_{m,t\tau} \, &= \,
    \left\{ \begin{array}{ll}
    \frac 1 J \, \sum_{j=1}^J \, Z_{m,jt} \, \widehat e_{j\tau}
      & \text{for $0 < t-\tau \leq h$}
     \\
       0 & \text{otherwise}
    \end{array} \right.
\end{align}
where $t,\tau=1,\ldots,T$, and $h \in \mathbbm{N}$ is a bandwidth parameter.
Using these objects we define
\begin{align}
        \widehat G
          & =
           \,
              \frac 1 {JT}
               \left( \begin{array}{cc} \widehat g' \, \widehat x^{\lambda f} & \widehat g' \,\widehat z^{\lambda f} \\
                \widehat x^{\lambda f \prime} \, \widehat x^{\lambda f} &
                    \widehat x^{\lambda f \prime} \,\widehat z^{\lambda f} \end{array} \right) \; ,
         \nonumber \\
         \widehat \Omega &=  \, \frac 1 {JT}
          \left(\widehat x^{\lambda f}, \widehat z^{\lambda f}\right)'
          {\rm diag}(\widehat \Sigma^{\rm vec}_e)
                \left(\widehat x^{\lambda f}, \widehat z^{\lambda f}\right)   \; ,
        \nonumber \\
        \widehat b^{(x,0)}_k &=  \,
                      {\rm Tr}\left(  P_{\widehat  f} \, \widehat \Sigma^{X,e}_k \right) \; ,
        \nonumber \\
        \widehat b^{(x,1)}_k &=  \,
       {\rm Tr}\left[ {\rm diag}\left(\widehat \Sigma_e^{(1)} \right) \, M_{\widehat \lambda} \, X_k \,
              \widehat  f \, (\widehat f^{\prime}\widehat  f)^{-1} \, (\widehat \lambda^{\prime}\widehat \lambda)^{-1} \, \widehat \lambda^{\prime} \right] \; ,
        \nonumber \\
        \widehat b^{(x,2)}_k &= \,
      {\rm Tr}\left[ {\rm diag}\left(\widehat \Sigma_e^{(2)} \right) \, M_{\widehat  f} \, X^{\prime}_k \,
              \widehat \lambda \, (\widehat \lambda^{\prime}\widehat \lambda)^{-1} \, (\widehat f^{\prime}\widehat  f)^{-1} \, \widehat f^{\prime} \right] \; ,
        \nonumber \\
        \widehat b^{(z,0)}_m &=   \,
               {\rm Tr}\left(  P_{\widehat  f} \, \widehat \Sigma^{Z,e}_m \right) \; ,
        \nonumber \\
        \widehat b^{(z,1)}_m &= \,
       {\rm Tr}\left[ {\rm diag}\left(\widehat \Sigma_e^{(1)} \right) \, M_{\widehat \lambda} \, Z_m \,
              \widehat  f \, (\widehat f^{\prime}\widehat  f)^{-1} \, (\widehat \lambda^{\prime}\widehat \lambda)^{-1} \, \widehat \lambda^{\prime} \right] \; ,
        \nonumber \\
        \widehat b^{(z,2)}_m &= \,
      {\rm Tr}\left[ {\rm diag}\left(\widehat \Sigma_e^{(2)} \right) \, M_{\widehat  f} \, Z^{\prime}_m \,
              \widehat \lambda \, (\widehat \lambda^{\prime}\widehat \lambda)^{-1} \, (\widehat f^{\prime}\widehat  f)^{-1} \, \widehat f^{\prime} \right] \; ,
\end{align}
for $k=1,\ldots,K$ and $m=1,\ldots,M$.
We set
$\widehat b^{(x,i)}=\left(\widehat b^{(x,i)}_1,\ldots,\widehat b^{(x,i)}_K \right)'$ and
$\widehat b^{(z,i)}=\left(\widehat b^{(z,i)}_1,\ldots,\widehat b^{(z,i)}_M \right)'$,
for $i=0,1,2$. The estimator of $\bigW$ is given by
\begin{align}
   \bigWest &=
        \left( \begin{array}{cc} \left( \frac 1 {JT} \widehat x^{\lambda f \prime}
                                           \widehat x^{\lambda f} \right)^{-1} & 0_{K\times M} \\
                              0_{M\times K} & 0_{M\times M} \end{array}  \right)
  + \left(\begin{array}{c} -  (\widehat x^{\lambda f \prime} \widehat x^{\lambda f})^{-1} \widehat x^{\lambda f \prime} \, \widehat z^{\lambda f} \\ \mathbbm{1}_M \end{array} \right)
         \left(\frac 1 {JT} \widehat z^{\lambda f \prime} M_{\widehat x^{\lambda f}} \widehat z^{\lambda f} \right)^{-1}
      \nonumber \\ & \qquad\qquad\qquad\qquad
         \, \smallW_{JT} \,
         \left(\frac 1 {JT} \widehat z^{\lambda f \prime} M_{\widehat x^{\lambda f}} \widehat z^{\lambda f} \right)^{-1}
         \left(\begin{array}{c} -  (\widehat x^{\lambda f \prime} \widehat x^{\lambda f})^{-1} \widehat x^{\lambda f \prime} \, \widehat z^{\lambda f} \\ \mathbbm{1}_M \end{array} \right)' \; .
\end{align}
Finally, for $i=0,1,2$, we have
\begin{align}
     \widehat B_i  &= - \left( \widehat G  \bigWest  \widehat G' \right)^{-1}
     \widehat G  \bigWest \,
     \left( \begin{array}{c} \widehat b^{(x,i)} \\ \widehat b^{(z,i)} \end{array} \right) \; .
\end{align}
The only subtlety here lies in the definition of $\widehat \Sigma^{X,e}_{k}$
and $\widehat \Sigma^{Z,e}_{m}$, where we explicitly impose the constraint that
$\widehat \Sigma^{X,e}_{k,t\tau}=\widehat \Sigma^{Z,e}_{m,t\tau}=0$ for $t-\tau \leq 0$
and for $t-\tau > h$, where $h \in \mathbbm{N}$ is a bandwidth parameter.
On the one side ($t-\tau \leq 0$) this constraint stems from the assumption that $X_k$ and $Z_m$ are only correlated with past values of the errors $e$, not with present and future values,
on the other side ($t-\tau > h$) we need the bandwidth cutoff to guarantee that the variance
of our estimator for $B_0$ converges to zero. Without imposing
this constraint and introducing the bandwidth parameter, our estimator for $B_0$ would be inconsistent.

\subsection{Correlation of $e_{jt}$ across $j$}

Our assumptions impose that the error term $e_{jt}$ is independent both across products $j$ and over markets/time $t$.
We allow for the regressors $X_{jt}$ to be pre-determined,  for example, lagged dependent variables are allowed, 
and would therefore run into identification problems if we also allowed  $e_{jt}$ to be to be 
correlated over time. However, it would not cause any conceptual problem to allow weak correlation of $e_{jt}$
across products $j$. The above formulas for the asymptotic variance and bias of the LS-MD estimator would need
to be modified as follows:
\begin{itemize}
    \item The diagonal matrix $ {\rm diag}(\Sigma^{\rm vec}_e)$ that appears in the definition of $\Omega$
       in equation~\eqref{DefGOmega} needs to be replaced by the potentially non-diagonal $JT \times JT$
       variance-covariance matrix of the $JT$-vector of error terms ${\rm vec}(e)$. 
       Otherwise, the formula for the asymptotic variance-covariance matrix
       $\left(G  \bigW  G' \right)^{-1}  G  \bigW  \Omega
          \bigW G'  \left(G  \bigW  G' \right)^{-1} $ of the LS-MD estimator is unchanged.
          
     \item The diagonal matrix  ${\rm diag}\left( \Sigma_e^{(1)} \right)$ that enters into the definition
     of $ b^{(x,1)}_k$ and $ b^{(z,1)}_m$ in equation \eqref{Defbzx012} needs to be replaced with 
     the potentially non-diagonal $J \times J$  matrix
      $\frac 1 T \sum_{t=1}^T \Sigma^{(e)}_t$, where $\Sigma^{(e)}_t$ is the variance-covariance matrix of
      the $J$-vector $(e_{jt} \, : \, j=1,\ldots,J)$. After this change and the change of $\Omega$ already described
      above, the asymptotic bias terms $B_0$, $B_1$ and $B_2$ are still given by equation \eqref{DefBias}.
\end{itemize}
 Those two modifications to the asymptotic variance and bias are very much in line with the results in
Bai \cite*{Bai2009}, who allows for cross-sectional dependence in the error term in a linear model with
interactive fixed effects.  We leave the question of bias correction and inference for the case of
cross-sectional dependence in $e_{jt}$ for future work.

\section{Proofs}

\footnotesize

In addition to the vectorizations $x$, $x^{\lambda f}$,
$z$, $z^{\lambda f}$, $g$, and $d(\alpha)$, which were
already defined above, we also introduce
the $JT \times K$ matrix $x^f$, the $JT \times M$ matrix $z^f$, and
the $JT \times 1$ vector $\varepsilon$
by
\begin{align*}
   x^f_{.,k} &= {\rm vec}\left( X_{k} M_{f^0} \right) ,  &
   z^f_{.,m} &= {\rm vec}\left( Z_{m} M_{f^0} \right) , &
   \varepsilon &= {\rm vec}\left( e \right) ,
\end{align*}
where $k=1,\ldots,K$ and $m=1,\ldots,M$.

\subsection{Proof of Identification}

\begin{proof}[\bf Proof of Theorem~\ref{th:id-new}]
To show that any two different \textit{parameters} cannot be observational
equivalent, we introduce the following functional
\[
  Q\left( \alpha ,\beta ,\gamma ,\lambda ,f ; F^0_{s,X,Z} \right) =
  \mathbbm{E}_0 \left\| \delta(\alpha^{})
              - \beta^{} \cdot X
              - \gamma^{} \cdot Z  -  \lambda^{}f^{'} \right\|_F^2 ,
\]
where $\mathbbm{E}_0$ refers to the expectation under the distribution of observables
$F^0_{s,X,Z}$, which is implied by the model, i.e.
$F^0_{s,X,Z}= \Gamma(\alpha^0, \beta^0, \lambda^0 f^{0 \prime}, F^0_{e,X,Z})$.

First, we show that under Assumption~\ref{ass:ID}(i)-(iv), the minima of the
       function $Q\left( \alpha^0 ,\beta ,\gamma ,\lambda ,f; F^0_{s,X,Z}\right)$
       over ($\beta$, $\gamma$, $\lambda$, $f$)
       satisfies
           $\beta=\beta^0$, $\gamma=0$, and $\lambda f' = \lambda^0 f^{0 \prime}$.
       Using model
       \eqref{model0} and Assumption \ref{ass:ID}$(ii)$ and $(iii)$ we find
       \begin{align}
          &Q\left( \alpha^0 ,\beta ,\gamma ,\lambda ,f;F^0_{s,X,Z}\right)
         \nonumber \\
          &= \mathbbm{E}_0 \,
          {\rm Tr}\left\{
          [\delta(\alpha^0)
              - \beta \cdot X - \gamma \cdot Z -  \lambda f']'
          [\delta(\alpha^0)
              - \beta \cdot X - \gamma \cdot Z - \lambda f'] \right\}
          \nonumber \\
          &= \mathbbm{E}_0 \,
          {\rm Tr}\left\{
          [ (\beta^0 - \beta) \cdot X - \gamma \cdot Z
        + \lambda^0 f^{0 \prime} - \lambda f' + e ]'
           [ (\beta^0 - \beta) \cdot X - \gamma \cdot Z
        + \lambda^0 f^{0 \prime} - \lambda f' + e ] \right\}
          \nonumber \\
          &= \mathbbm{E}_0{\rm Tr}( e' e)
            + \underbrace{  \mathbbm{E}_0 \,
          {\rm Tr}\left\{
          [ (\beta^0 - \beta) \cdot X - \gamma \cdot Z
        + \lambda^0 f^{0 \prime} - \lambda f' ]'
           [ (\beta^0 - \beta) \cdot X - \gamma \cdot Z
        + \lambda^0 f^{0 \prime} - \lambda f' ] \right\}
        }_{=Q^*(\beta,\gamma,\lambda,f; F^0_{s,X,Z})}        .
       \end{align}
       Note that $Q^*(\beta,\gamma,\lambda,f; F^0_{s,X,Z}) \geq 0$
       and that $Q^*(\beta^0,0,\lambda^0,f^0; F^0_{s,X,Z}) = 0$.
       Thus, the minimum value of
       $Q\left( \alpha^0 ,\beta , \gamma ,\lambda ,f; F^0_{s,X,Z}\right)$
       equals $\mathbbm{E}_0{\rm Tr}( e' e)$
       and all parameters that minimize
       $Q\left( \alpha^0 ,\beta , \gamma ,\lambda ,f; F^0_{s,X,Z}\right)$
       must satisfy $Q^*(\beta,\gamma,\lambda,f;F^0_{s,X,Z}) = 0$.
       We have for any $\lambda$ and $f$
     \begin{align}
         Q^*(\beta,\gamma,\lambda,f;F^0_{s,X,Z})
           &\geq
       \mathbbm{E}_0 \, {\rm Tr}\{ [(\beta^0 - \beta) \cdot X - \gamma \cdot Z ]'
        M_{(\lambda,\lambda^0)} [(\beta^0 - \beta) \cdot X - \gamma \cdot Z ] \}
       \nonumber \\
           &=
         [(\beta^0-\beta)',\gamma']
    \mathbbm{E}_0[(x,z)' (\mathbbm{1}_T \otimes M_{(\lambda,\lambda^0)}) (x,z)]
        [(\beta^0-\beta)',\gamma']'
        \nonumber \\
            & \geq b \left( \|\beta-\beta^0\|^2 + \|\gamma\|^2 \right)^2,
     \end{align}
      where the last line holds by Assumption \ref{ass:ID}$(iv)$.
     This shows that $\beta=\beta^0$ and $\gamma=0$ is necessary to
       minimize $Q\left( \alpha^0 ,\beta , \gamma ,\lambda ,f;F^0_{s,X,Z}\right)$.
        Since ${\rm Tr}(AA')=0$ for a matrix $A$ implies $A=0$, we find
       that  $Q^*(\beta^0,0,\lambda,f;F^0_{s,X,Z})=0$  implies
        $\beta=\beta^0$, $\gamma=0$ and $\lambda^0 f^{0 \prime} - \lambda f' = 0$.
     We have thus shown that $Q\left( \alpha^0 ,\beta ,\gamma ,\lambda ,f;F^0_{s,X,Z}\right)$
       is minimized if and only if
       $\beta=\beta^0$, $\gamma=0$ and  $\lambda f' = \lambda^0 f^{0 \prime}$.

For the second part, we introduce a second functional; for a given
$\alpha$ we define:
\begin{equation}
   \gamma(\alpha; F^0_{s,X,Z})
   \in \text{argmin}_{\gamma} \min_{\beta, \lambda, f}
         Q\left(  \alpha ,\beta ,\gamma ,\lambda ,f ; F^0_{s,X,Z} \right) .
    \label{DefGammaId}
\end{equation}
We show that under
Assumption~\ref{ass:ID}(i)-(v), $\gamma(\alpha; F^0_{s,X,Z})=0$ implies $\alpha=\alpha^0$.
    From part $(i)$ we already know that $\gamma(\alpha^0; F^0_{s,X,Z})=0$.
      The proof proceeds by contradiction. Assume that $\gamma(\alpha; F^0_{s,X,Z})=0$
      for $\alpha \neq \alpha^0$. By definition of $\gamma(\cdot)$ in
      Eq. (\ref{DefGammaId}), this
      implies that there exists
      $\tilde \beta$, $\tilde \lambda$ and $\tilde f$ such that
\begin{equation}\label{Inequality}
Q\left( \alpha ,\tilde \beta ,0 ,\tilde \lambda , \tilde f; F^0_{s,X,Z}\right)
       \leq \min_{\beta,\gamma,\lambda,f}
         Q\left( \alpha ,\beta ,\gamma ,\lambda ,f;F^0_{s,X,Z}\right).
\end{equation}
      Using model \eqref{model0} and our assumptions
      we obtain the following lower bound
      for the lhs of  inequality~\eqref{Inequality}
      \begin{align}
         & Q \left( \alpha ,\tilde \beta ,0 ,\tilde \lambda , \tilde f;F^0_{s,X,Z}\right)
          =
         \mathbbm{E}_0 {\rm Tr} \left[
           \left( \delta(\alpha)
              - \tilde \beta \cdot X
              -  \tilde \lambda \tilde f' \right)'
           \left( \delta(\alpha)
              - \tilde \beta \cdot X
              -  \tilde \lambda \tilde f' \right) \right]
        \nonumber \\
         & \quad =
         \mathbbm{E}_0 {\rm Tr} \Big[
           \left( \delta(\alpha) - \delta(\alpha^0)
              - (\tilde \beta - \beta^0) \cdot X
              + \lambda^0 f^0
              -  \tilde \lambda \tilde f' + e \right)'
          \nonumber \\  & \qquad \qquad \qquad \qquad
           \left( \delta(\alpha) - \delta(\alpha^0)
              - (\tilde \beta - \beta^0) \cdot X
              + \lambda^0 f^0
              -  \tilde \lambda \tilde f' + e \right) \Big]
        \nonumber \\
         & \quad =
         2 \mathbbm{E}_0 {\rm Tr} \Big[
           \left( \delta(\alpha) - \delta(\alpha^0) + \ft 1 2 e \right)' e \Big]
          +
         \mathbbm{E}_0 {\rm Tr} \Big[
           \left( \delta(\alpha) - \delta(\alpha^0)
              - (\tilde \beta - \beta^0) \cdot X
              + \lambda^0 f^0
              -  \tilde \lambda \tilde f' \right)'
          \nonumber \\  & \qquad \qquad \qquad \qquad
                    \qquad \qquad \qquad  \qquad \qquad  \;
           \left( \delta(\alpha) - \delta(\alpha^0)
              - (\tilde \beta - \beta^0) \cdot X
              + \lambda^0 f^0
              -  \tilde \lambda \tilde f' \right) \Big]
        \nonumber \\
         & \quad \geq
         2 \mathbbm{E}_0 {\rm Tr} \Big[
           \left( \delta(\alpha) - \delta(\alpha^0) + \ft 1 2 e \right)' e \Big]
          +
         \mathbbm{E}_0 {\rm Tr} \Big[
           \left( \delta(\alpha) - \delta(\alpha^0)
              - (\tilde \beta - \beta^0) \cdot X \right)'
              M_{(\tilde \lambda,\lambda^0)}
          \nonumber \\  & \qquad \qquad \qquad \qquad
                    \qquad \qquad \qquad  \qquad \qquad \qquad \qquad  \;
           \left( \delta(\alpha) - \delta(\alpha^0)
              - (\tilde \beta - \beta^0) \cdot X
              \right) \Big]
        \nonumber \\
         & \quad =
         2 \mathbbm{E}_0 {\rm Tr} \Big[
           \left( \delta(\alpha) - \delta(\alpha^0) + \ft 1 2 e \right)' e \Big]
          +  \mathbbm{E}_0\left[ \Delta \xi_{\alpha,\tilde \beta}'
     \left( \mathbbm{1}_T \otimes M_{(\tilde \lambda,\lambda^0)}\right)
     \Delta \xi_{\alpha,\tilde \beta} \right]
        \nonumber \\
         & \quad =
         2 \mathbbm{E}_0 {\rm Tr} \Big[
           \left( \delta(\alpha) - \delta(\alpha^0) + \ft 1 2 e \right)' e \Big]
          + \mathbbm{E}_0\left[ \Delta \xi_{\alpha,\tilde \beta}'
     \Delta \xi_{\alpha,\tilde \beta} \right]
          - \mathbbm{E}_0\left[ \Delta \xi_{\alpha,\tilde \beta}'
     \left( \mathbbm{1}_T \otimes P_{(\tilde \lambda,\lambda^0)}\right)
     \Delta \xi_{\alpha,\tilde \beta} \right]  .
      \end{align}
      Similarly, we obtain the following upper bound for the rhs of the
      above inequality (\ref{Inequality})
      \begin{align}
  &\min_{\beta,\gamma,\lambda,f} Q\left( \alpha ,\beta ,\gamma ,\lambda ,f;F^0_{s,X,Z}\right)
  \leq \min_{\beta,\gamma}
             Q\left( \alpha ,\beta ,\gamma, \lambda^0, f^0; F^0_{s,X,Z}\right)
   \nonumber \\
   & \quad = \min_{\beta,\gamma} \mathbbm{E}_0 {\rm Tr} \left[
           \left( \delta(\alpha)
              - \beta \cdot X
              -  \gamma \cdot Z - \lambda^0 f^{0 \prime}\right)'
           \left( \delta(\alpha)
              - \beta \cdot X
              -  \gamma \cdot Z - \lambda^0 f^{0 \prime} \right) \right]
   \nonumber \\
   & \quad = \min_{\beta,\gamma} \mathbbm{E}_0 {\rm Tr} \Big[
           \left( \delta(\alpha) - \delta(\alpha^0)
              - (\beta-\beta^0) \cdot X
              -  \gamma \cdot Z + e \right)'
         \nonumber \\ & \qquad \qquad \qquad \qquad \qquad
           \left( \delta(\alpha) - \delta(\alpha^0)
              - (\beta - \beta^0) \cdot X
              -  \gamma \cdot Z +e \right) \Big]
  \nonumber \\
  & \quad =
         2 \mathbbm{E}_0 {\rm Tr} \Big[
           \left( \delta(\alpha) - \delta(\alpha^0) + \ft 1 2 e \right)' e \Big]
        + \min_{\beta,\gamma} \mathbbm{E}_0 {\rm Tr} \Big[
           \left( \delta(\alpha) - \delta(\alpha^0)
              - (\beta-\beta^0) \cdot X
              -  \gamma \cdot Z  \right)'
         \nonumber \\ & \qquad \qquad \qquad \qquad \qquad \qquad \qquad \qquad
               \qquad \qquad \qquad
           \left( \delta(\alpha) - \delta(\alpha^0)
              - (\beta - \beta^0) \cdot X
              -  \gamma \cdot Z \right) \Big]
  \nonumber \\
   & \quad =
               2 \mathbbm{E}_0 {\rm Tr} \Big[
           \left( \delta(\alpha) - \delta(\alpha^0) + \ft 1 2 e \right)' e \Big]
        \nonumber \\ & \qquad \qquad
        + \min_{\beta,\gamma}
         \mathbbm{E}_0\left[
         \left(\Delta \xi_{\alpha,\tilde \beta}
         - x (\beta - \tilde \beta) - z \gamma \right)'
          \left(\Delta \xi_{\alpha,\tilde \beta}
         - x (\beta - \tilde \beta) - z \gamma \right) \right]
  \nonumber \\
   & \quad =
               2 \mathbbm{E}_0 {\rm Tr} \Big[
           \left( \delta(\alpha) - \delta(\alpha^0) + \ft 1 2 e \right)' e \Big]
          +\mathbbm{E}_0\left[ \Delta \xi_{\alpha,\tilde \beta}'
     \Delta \xi_{\alpha,\tilde \beta} \right]
        \nonumber \\ & \qquad \qquad
        - \mathbbm{E}_0\big[ \Delta \xi_{\alpha,\tilde \beta}' \, (x,z) \big]
     \mathbbm{E}_0\big[ (x,z)' (x,z) \big]^{-1}
     \mathbbm{E}_0\big[ (x,z)' \, \Delta \xi_{\alpha,\tilde \beta}  \big]  .
     \end{align}
     Plugging these bounds in the original inequality we obtain
     \begin{align}
        \mathbbm{E}_0\big[ \Delta \xi_{\alpha,\tilde \beta}' \, (x,z) \big]
     \mathbbm{E}_0\big[ (x,z)' (x,z) \big]^{-1}
     \mathbbm{E}_0\big[ (x,z)' \, \Delta \xi_{\alpha,\tilde \beta}  \big]
         \leq
       \mathbbm{E}_0\left[ \Delta \xi_{\alpha,\tilde \beta}'
     \left( \mathbbm{1}_T \otimes P_{(\tilde \lambda,\lambda^0)}\right)
     \Delta \xi_{\alpha,\tilde \beta} \right] ,
     \end{align}
     which is a contradiction to Assumption \ref{ass:ID}$(v)$.

We have thus shown that
     $\gamma(\alpha; F^0_{s,X,Z})=0$ implies $\alpha=\alpha^0$, which shows that $\alpha^0$
     is uniquely identified from $F^0_{s,X,Z}$.
Using that $\alpha^0$ is identified, we can now use the first part of the proof, and
uniquely identify $\beta^0$ and $\lambda^0 f^{0\prime}$ from $F^0_{s,X,Z}$
as the unique minimizers of $Q( \alpha^0 ,\beta ,\gamma ,\lambda ,f ;F^0_{s,X,Z})$.  Note that these findings immediately preclude
observational equivalence,
  {\em viz} two sets of distinct parameters $(\alpha^0, \beta^0, \lambda^0,
  f^0)\neq (\alpha^1, \beta^1, \lambda^1, f^1)$ which are both consistent
  with the observed distribution $F^0_{s,X,Z}$.

  Assumption~\ref{ass:INV} guarantees that for given  $\alpha^0$, $\beta^0$ and $\lambda^0 f^{0\prime}$
  the map $F^0_{s,X,Z} = \Gamma(\alpha^0, \beta^0, \lambda^0 f^{0 \prime}, F^0_{e,X,Z})$
  from $F^0_{e,X,Z}$ to $F^0_{s,X,Z}$
  is invertible, i.e. we can  uniquely identify $F^0_{e,X,Z}$ from $F^0_{s,X,Z}$.
\end{proof}

\subsection{Proof of Consistency}

\begin{proof}[\bf Proof of Theorem~\ref{th:consistency}]
   \# Part 1: We show that for any consistent estimator $\widehat \alpha$
     (not necessarily the LS-MD estimator)
      we have  $\tilde \beta_{\widehat \alpha} = \beta^0 + o_p(1)$ and
      $\tilde \gamma_{\widehat \alpha} = o_p(1)$. Thus, for this part of the proof
      assume that $\widehat \alpha = \alpha^0 + o_p(1)$.
       This part of the proof is a direct extension of the consistency proof
       in Moon and Weidner \cite*{MoonWeidner2015b}.
       We denote the least square objective function by
       $Q_{JT}(\alpha,\beta,\gamma,\lambda,f) =  \frac 1 {JT} \left\| \delta(\alpha)
              - \beta \cdot X
              - \gamma \cdot Z  -  \lambda f' \right\|_F^2$.
        We first establish a lower bound on $Q_{JT}(\widehat \alpha,\beta,\gamma,\lambda,f)$.
   We have for all $\lambda$, $f$:
    \begin{align}
       Q_{JT}(\widehat \alpha,\beta,\gamma,\lambda,f)
       &=   \frac 1 {JT}  {\rm Tr} \left[
            \left( \delta(\widehat \alpha)
              - \beta \cdot X
              - \gamma \cdot Z  -  \lambda f' \right)'
            \left( \delta(\widehat \alpha)
              - \beta \cdot X
              - \gamma \cdot Z  -  \lambda f' \right) \right]
       \nonumber \\
       &\geq
        \frac 1 {JT}  {\rm Tr} \left[
            \left( \delta(\widehat \alpha)
              - \beta \cdot X
              - \gamma \cdot Z  -  \lambda f' \right)' M_{(\lambda,\lambda^0)}
            \left( \delta(\widehat \alpha)
              - \beta \cdot X
              - \gamma \cdot Z  -  \lambda f' \right) \right]
          \nonumber \\
          &=  \frac 1 {JT}  {\rm Tr} \bigg[
            \left( (\delta(\widehat \alpha) - \delta(\alpha^0)) + e - (\beta-\beta^0) \cdot X
              - \gamma \cdot Z   \right)' M_{(\lambda,\lambda^0)}
                \nonumber \\ & \qquad \qquad \qquad \qquad \qquad
            \left( (\delta(\widehat \alpha) - \delta(\alpha^0)) + e - (\beta-\beta^0) \cdot X
              - \gamma \cdot Z   \right) \bigg]
         \nonumber \\
           & \geq b \, \| \beta - \beta^0 \|^2 + b \, \| \gamma \|^2
                + o_p\left( \| \beta - \beta^0 \| +  \| \gamma - \gamma^0 \| \right)
                          + \frac 1 {JT} \, {\rm Tr} \left( e  e' \right)
                          + o_p(1).
    \end{align}
    where in the last line we used Assumption~\ref{ass:con}$(i)$, $(ii)$, $(iii)$,
    and $(iv)$. Here are some representative examples of how the bounds
    in this last step are obtained from these assumptions:
    \begin{align}
          & \frac 1 {JT}  {\rm Tr} \left[
            \left(  (\beta-\beta^0) \cdot X
              - \gamma \cdot Z   \right)' M_{(\lambda,\lambda^0)}
             \left(   (\beta-\beta^0) \cdot X - \gamma \cdot Z   \right) \right]
          \nonumber \\
            & \qquad \qquad =
          (\beta',\gamma')
           [  \ft 1 {JT} (x,z)' (\mathbbm{1}_T \otimes M_{(\lambda,\lambda^0)}) (x,z)]
           (\beta',\gamma')'
          \nonumber \\
            & \qquad  \qquad
             \geq b  (\beta',\gamma') (\beta',\gamma')'
              = b \, \| \beta - \beta^0 \|^2 + b \, \| \gamma \|^2 ,
           \nonumber \\
           & \left| \frac 1 {JT}  {\rm Tr} \left[
             (\delta(\widehat \alpha) - \delta(\alpha^0))' M_{(\lambda,\lambda^0)}
            \left( (\beta-\beta^0) \cdot X  \right) \right]  \right|
          \nonumber \\
            & \qquad  \qquad
                \leq        \frac 1 {JT}
                    \left\| \delta(\widehat \alpha) - \delta(\alpha^0) \right\|_F
                    \left\| M_{(\lambda,\lambda^0)} \left( (\beta-\beta^0) \cdot X  \right) \right\|_F
          \nonumber \\
            & \qquad  \qquad
                \leq        \frac 1 {JT}
                    \left\| \delta(\widehat \alpha) - \delta(\alpha^0) \right\|_F
                    \left\| (\beta-\beta^0) \cdot X  \right\|_F
          \nonumber \\
            & \qquad  \qquad
                  = {\cal O}_p(1)       \left\| \widehat \alpha - \alpha^0 \right\|
                                \left\| \beta - \beta^0 \right\|     = o_p(     \left\| \beta - \beta^0 \right\|  )  ,
          \nonumber \\
           & \left|  \frac 1 {JT}  {\rm Tr} \bigg[
             e' M_{(\lambda,\lambda^0)}  \left( (\beta-\beta^0) \cdot X   \right) \bigg]   \right|
          \nonumber \\
            & \qquad  \qquad
                 =  \left| \frac 1 {JT}  {\rm Tr} \bigg[
             e'  \left( (\beta-\beta^0) \cdot X   \right) \bigg] \right|
                + \left|  \frac 1 {JT}  {\rm Tr} \bigg[
             e' P_{(\lambda,\lambda^0)}  \left( (\beta-\beta^0) \cdot X   \right) \bigg] \right|
          \nonumber \\
            & \qquad  \qquad
                \leq o_p(1) \|\beta-\beta^0 \| +
                   \frac R {JT} \|e\| \left\| (\beta-\beta^0) \cdot X  \right\|
          \nonumber \\
            & \qquad  \qquad
                \leq o_p(1) \|\beta-\beta^0 \| +
                   \frac R {JT} \|e\| \left\| (\beta-\beta^0) \cdot X  \right\|_F
                 = o_p(     \left\| \beta - \beta^0 \right\|  )  .
            \label{SomeBoundsNeeded}
    \end{align}
    See the supplementary material in
    Moon and Weidner \cite*{MoonWeidner2015}
    for further details regarding the algebra here.
    Applying the same methods, we also obtain
    \begin{align}
         Q_{JT}(\widehat \alpha,\beta^0,0,\lambda^0,f^0)
          =  \frac 1 {JT} \, {\rm Tr} \left( e  e' \right) + o_p(1).
    \end{align}
    Since we could choose $\beta=\beta^0$,
    $\gamma=0$, $\lambda=\lambda^0$ and $f=f^0$
     in the first step minimization of the LS-MD estimator, the optimal LS-MD
     first stage parameters at $\widehat \alpha$ need to
    satisfy $Q_{JT}(\widehat \alpha,\tilde \beta_{\widehat \alpha},\tilde \gamma_{\widehat \alpha},\tilde \lambda_{\widehat \alpha},\tilde f_{\widehat \alpha}) \leq Q_{JT}(\widehat \alpha,\beta^0,0,\lambda^0,f^0)$.
    Using the above results thus gives
    \begin{align}
          b \, \| \tilde \beta_{\widehat \alpha} - \beta^0 \|^2 + b \, \| \tilde \gamma_{\widehat \alpha} \|^2
                + o_p\left( \| \tilde \beta_{\widehat \alpha} - \beta^0 \| +  \| \tilde \gamma_{\widehat \alpha} - \gamma^0 \| \right)
                          +  o_p(1)
                          \, \leq  \, 0 \; .
    \end{align}
    It follows that
    $\|  \tilde \beta_{\widehat \alpha} - \beta^0 \|  = o_p(1)$
    and $\tilde \gamma_{\widehat \alpha} = o_p(1)$.

   \# Part 2: Now, let $\widehat \alpha$ be the LS-MD estimator. We want to show that
      $\widehat \alpha - \alpha^0 = o_p(1)$. From part 1 of the proof we already know that
      $\tilde \gamma_{\alpha^0} = o_p(1)$.
   In the second step of the LS-MD estimator
     the optimal choice $\widehat \alpha$ minimizes
$\tilde \gamma'_{\widehat \alpha} \, \smallW_{JT} \, \tilde \gamma_{\widehat \alpha}$,
which implies that
\begin{align}
   \tilde \gamma'_{\widehat \alpha} \, \smallW_{JT} \, \tilde \gamma_{\widehat \alpha}
   \; \leq \; \tilde \gamma'_{\alpha^0} \,
   \smallW_{JT} \, \tilde \gamma_{\alpha^0} \; = \; o_p(1) \; ,
   \label{gammaBound}
\end{align}
and therefore $\tilde \gamma_{\widehat \alpha} = o_p(1)$.
Here we used that
$\smallW_{JT}$ converges to a positive definite matrix in probability.
Analogous to the identification proof we are now going to find an upper
and a lower bound for
$Q_{JT}\left( \widehat \alpha ,\tilde \beta_{\widehat \alpha} ,
\tilde \gamma_{\widehat \alpha} ,\tilde \lambda_{\widehat \alpha} ,
\tilde f_{\widehat \alpha} \right)$.
In the rest of this proof we drop the subscript
$\widehat \alpha$ on $\tilde \beta$, $\tilde \gamma$, $\tilde \lambda$
and $\tilde f$.
      Using model \eqref{model0} and our assumptions
      we obtain the following lower bound
      \begin{align}
         & Q_{JT}\left( \widehat \alpha ,\tilde \beta ,
\tilde \gamma ,\tilde \lambda ,
\tilde f \right)
          =
         \ft 1 {JT} {\rm Tr} \left[
           \left( \delta(\widehat \alpha)
              - \tilde \beta \cdot X
              - \tilde \gamma \cdot Z
              -  \tilde \lambda \tilde f' \right)'
           \left( \delta(\widehat \alpha)
              - \tilde \beta \cdot X
              - \tilde \gamma \cdot Z
              -  \tilde \lambda \tilde f' \right) \right]
        \nonumber \\
         & \quad =
          \ft 1 {JT} {\rm Tr} \Big[
           \left( \delta(\widehat \alpha) - \delta(\alpha^0)
              - (\tilde \beta - \beta^0) \cdot X
               - \tilde \gamma \cdot Z
              + \lambda^0 f^0
              -  \tilde \lambda \tilde f' + e \right)'
          \nonumber \\  & \qquad \qquad \qquad \qquad
           \left( \delta(\widehat \alpha) - \delta(\alpha^0)
              - (\tilde \beta - \beta^0) \cdot X
               - \tilde \gamma \cdot Z
              + \lambda^0 f^0
              -  \tilde \lambda \tilde f' + e \right) \Big]
       \nonumber \\
         & \quad \geq
          \ft 1 {JT} {\rm Tr} \Big[
           \left( \delta(\widehat \alpha) - \delta(\alpha^0)
              - (\tilde \beta - \beta^0) \cdot X
               - \tilde \gamma \cdot Z
              + e \right)'  M_{(\tilde \lambda,\lambda^0)}
          \nonumber \\  & \qquad \qquad \qquad \qquad \qquad
          \left( \delta(\widehat \alpha) - \delta(\alpha^0)
              - (\tilde \beta - \beta^0) \cdot X
               - \tilde \gamma \cdot Z
             + e \right) \Big]
        \nonumber \\
         & \quad =
             \ft 1 {JT}
          {\rm Tr} \Big[
           \left( \delta(\widehat \alpha) - \delta(\alpha^0)
              - (\tilde \beta - \beta^0) \cdot X \right)'
              M_{(\tilde \lambda,\lambda^0)}
           \left( \delta(\widehat \alpha) - \delta(\alpha^0)
              - (\tilde \beta - \beta^0) \cdot X
              \right) \Big]
         \nonumber \\ & \qquad \qquad
       +   \ft 2 {JT}  {\rm Tr} \Big[
           \left( \delta(\widehat \alpha) - \delta(\alpha^0) + \ft 1 2 e \right)' e \Big]
           + o_p(\| \widehat \alpha - \alpha^0\| + \|\tilde \beta - \beta^0\|  ) + o_p(1)
       \nonumber \\
         & \quad =
         \ft 1 {JT} \left[ \Delta \xi_{\widehat \alpha,\tilde \beta}'
     \left( \mathbbm{1}_T \otimes M_{(\tilde \lambda,\lambda^0)}\right)
     \Delta \xi_{\widehat \alpha,\tilde \beta} \right]
        \nonumber \\
        & \qquad \qquad
      +  \ft 2 {JT}  {\rm Tr} \Big[
           \left( \delta(\widehat \alpha) - \delta(\alpha^0) + \ft 1 2 e \right)' e \Big]
           + o_p(\| \widehat \alpha - \alpha^0\| + \|\tilde \beta - \beta^0\|  ) + o_p(1)
        \nonumber \\
         & \quad =
        \ft 1 {JT} \left[ \Delta \xi_{\widehat \alpha,\tilde \beta}'
     \Delta \xi_{\widehat \alpha,\tilde \beta} \right]
          - \ft 1 {JT}  \left[ \Delta \xi_{\widehat \alpha,\tilde \beta}'
     \left( \mathbbm{1}_T \otimes P_{(\tilde \lambda,\lambda^0)}\right)
     \Delta \xi_{\widehat \alpha,\tilde \beta} \right]
       \nonumber \\
      & \qquad \qquad
     +  \ft 2 {JT}  {\rm Tr} \Big[
           \left( \delta(\widehat \alpha) - \delta(\alpha^0) + \ft 1 2 e \right)' e \Big]
           + o_p(\| \widehat \alpha - \alpha^0\| + \|\tilde \beta - \beta^0\|  )  + o_p(1).
           \label{con_proof_lower_bnd}
      \end{align}
      The bounds used here are analogous to those in \eqref{SomeBoundsNeeded},
     and we again refer to the supplementary material in
    Moon and Weidner \cite*{MoonWeidner2015}.

      Similarly, we obtain the following upper bound
   \begin{align}
   &  Q_{JT}\left( \widehat \alpha ,\tilde \beta_{\widehat \alpha} ,
\tilde \gamma_{\widehat \alpha} ,\tilde \lambda_{\widehat \alpha} ,
\tilde f_{\widehat \alpha} \right)
= \min_{\beta,\gamma,\lambda,f}
         Q_{JT}\left( \widehat \alpha ,\beta ,\gamma ,\lambda ,f\right)
  \leq \min_{\beta,\gamma}
             Q_{JT}\left( \widehat \alpha,\beta ,\gamma, \lambda^0, f^0\right)
   \nonumber \\
   & \quad = \min_{\beta,\gamma} \ft 1 {JT} {\rm Tr} \left[
           \left( \delta(\widehat \alpha)
              - \beta \cdot X
              -  \gamma \cdot Z - \lambda^0 f^{0 \prime}\right)'
           \left( \delta(\widehat \alpha)
              - \beta \cdot X
              -  \gamma \cdot Z - \lambda^0 f^{0 \prime} \right) \right]
   \nonumber \\
   & \quad = \min_{\beta,\gamma}  \ft 1 {JT} {\rm Tr} \Big[
           \left( \delta(\widehat \alpha) - \delta(\alpha^0)
              - (\beta-\beta^0) \cdot X
              -  \gamma \cdot Z + e \right)'
         \nonumber \\ & \qquad \qquad \qquad \qquad \qquad
           \left( \delta(\widehat \alpha) - \delta(\alpha^0)
              - (\beta - \beta^0) \cdot X
              -  \gamma \cdot Z +e \right) \Big]
  \nonumber \\
  & \quad =
         \ft 2 {JT}  {\rm Tr} \Big[
           \left( \delta(\widehat \alpha) - \delta(\alpha^0) + \ft 1 2 e \right)' e \Big]
        + \min_{\beta,\gamma} \ft 1 {JT} {\rm Tr} \Big[
           \left( \delta(\widehat \alpha) - \delta(\alpha^0)
              - (\beta-\beta^0) \cdot X
              -  \gamma \cdot Z  \right)'
         \nonumber \\ & \qquad \qquad \qquad \qquad \qquad  \qquad \qquad
                           \qquad  \qquad \qquad
           \left( \delta(\widehat \alpha) - \delta(\alpha^0)
              - (\beta - \beta^0) \cdot X
              -  \gamma \cdot Z \right) \Big]
  \nonumber \\
   & \quad =
               \ft 2 {JT}  {\rm Tr} \Big[
           \left( \delta(\widehat \alpha) - \delta(\alpha^0) + \ft 1 2 e \right)' e \Big]
        \nonumber \\ & \qquad \qquad
        + \min_{\beta,\gamma}
         \ft 1 {JT} \left[
         \left(\Delta \xi_{\widehat \alpha,\tilde \beta}
         - x (\beta - \tilde \beta) - z \gamma \right)'
          \left(\Delta \xi_{\widehat \alpha,\tilde \beta}
         - x (\beta - \tilde \beta) - z \gamma \right) \right]
  \nonumber \\
   & \quad =
               \ft 2 {JT}  {\rm Tr} \Big[
           \left( \delta(\widehat \alpha) - \delta(\alpha^0) + \ft 1 2 e \right)' e \Big]
          + \ft 1 {JT} \left[ \Delta \xi_{\alpha,\tilde \beta}'
     \Delta \xi_{\widehat \alpha,\tilde \beta} \right]
        \nonumber \\ & \qquad \qquad
        -  \ft 1 {JT} \big[ \Delta \xi_{\widehat \alpha,\tilde \beta}' \, (x,z) \big]
     \big[ (x,z)' (x,z) \big]^{-1}
     \big[ (x,z)' \, \Delta \xi_{\widehat \alpha,\tilde \beta}  \big]  .
     \end{align}
     Combining this upper and lower bound we obtain
     \begin{align}
       &  \ft 1 {JT} \big[ \Delta \xi_{\widehat \alpha,\tilde \beta}' \, (x,z) \big]
     \big[ (x,z)' (x,z) \big]^{-1}
     \big[ (x,z)' \, \Delta \xi_{\widehat \alpha,\tilde \beta}  \big]
    \nonumber \\ & \qquad \qquad \qquad \qquad
         -
      \ft 1 {JT} \left[ \Delta \xi_{\widehat \alpha,\tilde \beta}'
     \left( \mathbbm{1}_T \otimes P_{(\tilde \lambda,\lambda^0)}\right)
     \Delta \xi_{\widehat \alpha,\tilde \beta} \right]
        \leq  o_p(\| \widehat \alpha - \alpha^0\| + \|\tilde \beta - \beta^0\|  ) + o_p(1).
     \end{align}
    Using Assumption~\ref{ass:con}$(v)$ we thus obtain
     \begin{align}
       &  b  \| \widehat \alpha - \alpha^0\|^2 + b \|\tilde \beta - \beta^0\| ^2
        \leq  o_p(\| \widehat \alpha - \alpha^0\| + \|\tilde \beta - \beta^0\|  ) + o_p(1),
     \end{align}
     from which we can conclude that
     $\| \widehat \alpha - \alpha^0 \|= o_p(1)$
     and $\| \tilde \beta - \beta^0 \| = o_p(1)$.

     \# Part 3: Showing consistency of $\widehat \beta$ obtained from step 3 of the LS-MD
        estimation procedure is analogous to part 1 of the proof --- one only needs
        to eliminate all $\gamma$ variables from part 1 of the proof,
        which actually simplifies the proof.
\end{proof}

\subsection{Proof of Limiting Distribution}

\begin{lemma}
   \label{lemma:sqrtJTalpha}
   Let Assumption~\ref{ass:con}
   be satisfied and in addition let
   $(JT)^{-1/2} {\rm Tr}(e X_k')={\cal O}_p(1)$,
   and $(JT)^{-1/2} {\rm Tr}(e Z_m')={\cal O}_p(1)$.
   In the limit $J,T \rightarrow \infty$ with $J/T \rightarrow \kappa^2$,
   $0 < \kappa < \infty$, we then have $\sqrt{J}(\widehat \alpha - \alpha)={\cal O}_p(1)$.
\end{lemma}

\begin{proof}[\bf Proof]
   The proof is exactly analogous to the consistency proof.
   We know from Moon and Weidner \cite*{MoonWeidner2015,MoonWeidner2015b}
   that  $\sqrt{J} \tilde \gamma_{\alpha^0} = {\cal O}_p(1)$.
   Applying the inequality \eqref{gammaBound}
   one thus finds  $\sqrt{J} \tilde \gamma_{\widehat \alpha} = {\cal O}_p(1)$.
   With the additional assumptions in the lemma one can strengthen
   the result in \eqref{con_proof_lower_bnd} as follows
      \begin{align}
          & Q_{JT}\left( \widehat \alpha ,\tilde \beta ,
\tilde \gamma ,\tilde \lambda ,
\tilde f \right)
   \nonumber \\
                &\quad \geq
        \ft 1 {JT} \left[ \Delta \xi_{\widehat \alpha,\tilde \beta}'
     \Delta \xi_{\widehat \alpha,\tilde \beta} \right]
          - \ft 1 {JT}  \left[ \Delta \xi_{\widehat \alpha,\tilde \beta}'
     \left( \mathbbm{1}_T \otimes P_{(\tilde \lambda,\lambda^0)}\right)
     \Delta \xi_{\widehat \alpha,\tilde \beta} \right]
       \nonumber \\
      & \qquad \quad
     +  \ft 2 {JT}  {\rm Tr} \Big[
           \left( \delta(\widehat \alpha) - \delta(\alpha^0) + \ft 1 2 e \right)' e \Big]
           + {\cal O}_p \left(\sqrt{J} \| \widehat \alpha - \alpha^0\| +
           \sqrt{J} \|\tilde \beta - \beta^0\|  \right)  + {\cal O}_p(1/J).
      \end{align}
      Using this stronger result
      and following the steps in the consistency proof
      then yields
       $\sqrt{J}(\widehat \alpha - \alpha)={\cal O}_p(1)$.
\end{proof}

\begin{proof}[\bf Proof of Theorem \ref{th:asympt_dist}]

Assumption \ref{ass:A4} guarantees $(JT)^{-1/2} {\rm Tr}(e X_k')={\cal O}_p(1)$,
   and $(JT)^{-1/2} {\rm Tr}(e Z_m')={\cal O}_p(1)$, so that we can apply
 Lemma \ref{lemma:sqrtJTalpha} to conclude
$\sqrt{J}(\widehat \alpha - \alpha)={\cal O}_p(1)$.

The first step in the definition of the LS-MD estimator is equivalent to the
linear regression model with interactive fixed effects, but
with an error matrix that has an additional term
$\Delta \delta(\alpha) \equiv \delta(\alpha)-\delta(\alpha^0)$, we write
$\Psi(\alpha) \equiv e+\Delta \delta(\alpha)$ for this effective error term.
Using $\widehat \alpha - \alpha^0 = o_p(1)$ and Assumption
\ref{ass:con}$(i)$
we have $\|\Psi(\widehat \alpha)\|= o_p(\sqrt{JT})$, so that the results in
Moon and Weidner \cite*{MoonWeidner2015,MoonWeidner2015b} guarantee
$\tilde \beta_{\widehat \alpha}-\beta^0 = o_p(1)$
and $\| \tilde \gamma_{\widehat \alpha} \| = o_p(1)$, which we already used
in the consistency proof. Using $\sqrt{J}(\widehat \alpha - \alpha)={\cal O}_p(1)$
and Assumption \ref{ass:A4}$(i)$ we find
$\|\Psi(\widehat \alpha)\|= {\cal O}_p(\sqrt{J})$, which allows us to truncate the
asymptotic likelihood expansion derived in Moon and Weidner \cite*{MoonWeidner2015,MoonWeidner2015b}
at an appropriate order. Namely, applying their results we have
\begin{align}
   \sqrt{JT} \left( \begin{array}{c} \tilde \beta_{\alpha} - \beta^0
                   \\[2mm] \tilde \gamma_{\alpha}  \end{array} \right)
     \, &= \, V_{JT}^{-1} \,
        \left( \begin{array}{c}
         \left[ C^{(1)}\left(X_k,\Psi(\alpha) \right)
         +C^{(2)}\left(X_k,\Psi(\alpha) \right) \right]_{k=1,\ldots,K}
                   \\[2mm]
         \left[ C^{(1)}\left(Z_m,\Psi(\alpha) \right)
         +C^{(2)}\left(Z_m,\Psi(\alpha) \right) \right]_{m=1,\ldots,M}
         \end{array} \right)
      + r^{\rm LS}(\alpha),
    \label{betagammaAsym}
\end{align}
where
\begin{align}
  V_{JT} &= \frac 1 {JT}
        \left(  \begin{array}{cc}
         \left[ {\rm Tr}(M_{f^0}  X^{\prime}_{k_1}  M_{\lambda^0}  X_{k_2})
         \right]_{k_1,k_2=1,\ldots,K}
         &
         \left[ {\rm Tr}(M_{f^0}  X^{\prime}_{k}  M_{\lambda^0}  Z_{m})
         \right]_{k=1,\ldots,K;m=1,\ldots,M}
         \\
         \left[ {\rm Tr}(M_{f^0}  Z^{\prime}_{m}  M_{\lambda^0}  X_{k})
         \right]_{m=1,\ldots,M;k=1,\ldots,K}
         &
         \left[ {\rm Tr}(M_{f^0}  Z^{\prime}_{m_1}  M_{\lambda^0}  Z_{m_2})
         \right]_{m_1,m_2=1,\ldots,M}
         \end{array} \right)
     \nonumber \\
           &= \frac 1 {JT}  \left(x^{\lambda f} , z^{\lambda f} \right)'
                            \left(x^{\lambda f} , z^{\lambda f} \right) \; ,
\end{align}
and for ${\cal X}$ either $X_k$ or $Z_m$ and $\Psi=\Psi(\alpha)$ we have
\begin{align}
  C^{(1)}\left({\cal X},\, \Psi \right) &= \frac 1 {\sqrt{JT}} \,
  {\rm Tr}\left[M_{f^0}
         \,\Psi^{\prime}\, M_{\lambda^0} \, {\cal X} \right] \; ,
   \nonumber  \\
  C^{(2)}\left({\cal X},\, \Psi \right) &= - \, \frac 1 {\sqrt{JT}} \, \bigg[
       {\rm Tr}\left(\Psi M_{f^0} \, \Psi' \, M_{\lambda^0} \, {\cal X} \,
              f^0 \, (f^{0\prime}f^0)^{-1} \, (\lambda^{0\prime}\lambda^0)^{-1} \, \lambda^{0\prime} \right)
    \nonumber \\ & \qquad \qquad \quad
       +{\rm Tr}\left(\Psi^{\prime}M_{\lambda^0} \, \Psi \, M_{f^0} \, {\cal X}^{\prime} \,
              \lambda^0 \, (\lambda^{0\prime}\lambda^0)^{-1} \, (f^{0\prime}f^0)^{-1} \, f^{0\prime} \right)
    \nonumber \\ & \qquad \qquad \quad
       +{\rm Tr}\left(\Psi^{\prime}M_{\lambda^0} \, {\cal X} \, M_{f^0} \, \Psi^{\prime}
                \, \lambda^0 \, (\lambda^{0\prime}\lambda^0)^{-1} \, (f^{0\prime}f^0)^{-1} \, f^{0\prime} \right)
                        \bigg] ,
\end{align}
and finally for the remainder we have
\begin{align}
   r^{\rm LS}(\alpha) &= {\cal O}_p \left( (JT)^{-3/2} \|\Psi(\alpha)\|^3 \|X_k\| \right)
                          +{\cal O}_p \left( (JT)^{-3/2} \|\Psi(\alpha)\|^3 \|Z_m\| \right)
           \nonumber \\ & \quad
 + {\cal O}_p \left( (JT)^{-1} \|\Psi(\alpha)\| \|X_k\|^2 \| \|\tilde \beta_\alpha-\beta^0\| \right)
 +{\cal O}_p \left( (JT)^{-1} \|\Psi(\alpha)\| \|Z_m\|^2 \| \tilde \gamma_\alpha \| \right)
 \; ,
\end{align}
which holds uniformly over $\alpha$. The first two terms in $r^{\rm LS}(\alpha)$
stem from the bound on higher order terms in the score function
($C^{(3)}$, $C^{(4)}$, etc.), where $\Psi(\alpha)$ appears three times or more
in the expansion, while the last two terms in $r^{\rm LS}(\alpha)$ reflect the bound
on higher order terms in the Hessian expansion, and beyond.
Note that Assumption \ref{ass:con}$(iv)$ already guarantees that $V_{JT} > b > 0$, wpa1.
Applying $\|X_k\|={\cal O}_p(\sqrt{JT})$, $\|Z_m\|={\cal O}_p(\sqrt{JT})$,
and
$\|\Psi(\alpha)\|={\cal O}_p(\sqrt{J})$ within $\sqrt{J} \|\alpha-\alpha^0\| < c$,
we find for all $c>0$
\begin{align}
   \sup_{\{\alpha: \, \sqrt{J} \|\alpha-\alpha^0\| < c \}}
       \frac{ \left\| r^{\rm LS}(\alpha) \right\| }
       {1 + \sqrt{JT}  \|\tilde \beta_\alpha-\beta^0\|
          + \sqrt{JT} \| \tilde \gamma_\alpha \|}  &= o_p(1) \; .
\end{align}
The inverse of the partitioned matrix $V_{JT}$ is given by
\begin{align}
   V_{JT}^{-1} &= JT \left( \begin{array}{c@{\hspace{-6mm}}c}
           \left(x^{\lambda f\prime} M_{z^{\lambda f}} x^{\lambda f} \right)^{-1}
           &
           - \left(x^{\lambda f\prime} M_{z^{\lambda f}} x^{\lambda f} \right)^{-1}
 \left(x^{\lambda f\prime} z^{\lambda f} \right) \left(z^{\lambda f\prime} z^{\lambda f} \right)^{-1}
           \\[2mm]
           - \left(z^{\lambda f\prime} M_{x^{\lambda f}} z^{\lambda f} \right)^{-1}
 \left(z^{\lambda f\prime} x^{\lambda f} \right) \left(x^{\lambda f\prime} x^{\lambda f} \right)^{-1}
           &
           \left(z^{\lambda f\prime} M_{x^{\lambda f}} z^{\lambda f} \right)^{-1}
            \end{array} \right)  .
\end{align}
Using $\sqrt{J}(\widehat \alpha - \alpha)={\cal O}_p(1)$ and
Assumption \ref{ass:A4}$(i)$ we find
\begin{align}
   \left( \begin{array}{c}
         \left[ C^{(1)}\left(X_k,\Psi(\widehat \alpha) \right)   \right]_{k=1,\ldots,K}
                   \\[2mm]
         \left[ C^{(1)}\left(Z_m,\Psi(\widehat \alpha) \right) \right]_{m=1,\ldots,M}
         \end{array} \right)
     &=  \frac 1 {\sqrt{JT}}   \left(x^{\lambda f} , z^{\lambda f} \right)'  \varepsilon
        \nonumber \\ & \qquad
        - \left[ \frac 1 {JT}   \left(x^{\lambda f} , z^{\lambda f} \right)'   g \right]
          \sqrt{JT}  ( \widehat \alpha - \alpha^0 )
         +  o_p(\sqrt{JT}  \| \widehat \alpha - \alpha^0 \|)  ,
 \nonumber \\
   \left( \begin{array}{c}
         \left[ C^{(2)}\left(X_k,\Psi(\widehat \alpha) \right)   \right]_{k=1,\ldots,K}
                   \\[2mm]
         \left[ C^{(2)}\left(Z_m,\Psi(\widehat \alpha) \right) \right]_{m=1,\ldots,M}
         \end{array} \right)
    &= \left( \begin{array}{c} c^{(2)}_x \\[2mm] c^{(2)}_z \end{array} \right)
          + {\cal O}_p\left( \sqrt{J} \| \widehat \alpha- \alpha^0 \| \right) ,
\end{align}
where
\begin{align}
    c^{(2)}_x &= \left[ C^{(2)}\left(X_k,e \right)   \right]_{k=1,\ldots,K} \; ,
    &
    c^{(2)}_z &= \left[ C^{(2)}\left(Z_m,e \right) \right]_{m=1,\ldots,M} \; .
\end{align}
From this one can conclude that
$\sqrt{JT} \| \tilde \beta_{\widehat \alpha} - \beta^0 \| = {\cal O}_p(1)
  + {\cal O}_p(\sqrt{JT}  \| \widehat \alpha - \alpha^0 \|)$
and $\sqrt{JT} \| \tilde \gamma_{\widehat \alpha} \|
   = {\cal O}_p(1) + {\cal O}_p(\sqrt{JT}  \| \widehat \alpha - \alpha^0 \|)$,
so that we find
$r^{\rm LS}(\widehat \alpha) = o_p(1) + o_p(\sqrt{JT}  \| \widehat \alpha - \alpha^0 \|)$.
Combining the above results we obtain
\begin{align}
   \sqrt{JT} \, \tilde \gamma_{\widehat \alpha}
     &=\left( \frac 1 {JT} z^{\lambda f\prime} M_{x^{\lambda f}} z^{\lambda f} \right)^{-1}
         \Bigg[  \frac 1 {\sqrt{JT}} \,z^{\lambda f \prime} \, M_{x^{\lambda f}} \varepsilon
                 + c^{(2)}_z
  - \left(z^{\lambda f\prime} x^{\lambda f} \right) \left(x^{\lambda f\prime} x^{\lambda f} \right)^{-1} c^{(2)}_x
    \nonumber \\
    & \qquad \qquad
   - \left( \frac 1 {JT} \, z^{\lambda f \prime} \, M_{x^{\lambda f}}   g \right)
         \; \sqrt{JT} \, ( \widehat \alpha - \alpha^0 ) \Bigg]
    +  o_p(1)
    +  o_p(\sqrt{JT}  \| \widehat \alpha - \alpha^0 \|) .
\end{align}
The above results holds not only for $\widehat \alpha$, but uniformly
for all $\alpha$ in any $\sqrt{J}$ shrinking neighborhood of $\alpha^0$
(we made this explicit in the bound on $r^{\rm LS}(\alpha)$
above; one could define similar remainder terms with
corresponding bounds in all intermediate steps),
i.e. we have
\begin{align}
   \sqrt{JT} \, \tilde \gamma_{\alpha}
     &=\left( \frac 1 {JT} z^{\lambda f\prime} M_{x^{\lambda f}} z^{\lambda f} \right)^{-1}
         \Bigg[  \frac 1 {\sqrt{JT}} \,z^{\lambda f \prime} \, M_{x^{\lambda f}} \varepsilon
                 + c^{(2)}_z
  - \left(z^{\lambda f\prime} x^{\lambda f} \right) \left(x^{\lambda f\prime} x^{\lambda f} \right)^{-1} c^{(2)}_x
    \nonumber \\
    & \qquad \qquad \qquad\qquad\qquad\qquad
   - \left( \frac 1 {JT} \, z^{\lambda f \prime} \, M_{x^{\lambda f}}   g \right)
         \; \sqrt{JT} \, (  \alpha - \alpha^0 ) \Bigg]
    + r^\gamma (\alpha) ,
\end{align}
where for all $c>0$
\begin{align}
    \sup_{\{\alpha: \, \sqrt{J} \|\alpha-\alpha^0\| < c \}}
       \frac{ \| r^{\gamma}(\alpha) \| }
       { 1 + \sqrt{JT} \| \alpha - \alpha^0 \| }
         &= o_p(1) \; .
\end{align}
Therefore, the objective function for $\widehat \alpha$ reads
\begin{align}
   JT \, \tilde \gamma_{\alpha}' \, \smallW_{JT} \, \tilde \gamma_{\alpha}
     &= A_0 - 2 \, A'_1 \, \left[\sqrt{JT} \, \left( \alpha - \alpha^0 \right) \right] +
     \left[\sqrt{JT} \, \left( \alpha - \alpha^0 \right) \right]' \, A_2
      \, \left[\sqrt{JT} \, \left( \alpha - \alpha^0 \right) \right]
            + r^{\rm obj}(\alpha) \; ,
\end{align}
where $A_0$ is a scalar, $A_1$ is an $L\times 1$ vector, and $A_2$ is an $L\times L$ matrix
defined by
\begin{align}
   A_0 &=  \left[  \frac 1 {\sqrt{JT}} \,z^{\lambda f \prime} \, M_{x^{\lambda f}} \varepsilon
                 + c^{(2)}_z
  - \left(z^{\lambda f\prime} x^{\lambda f} \right) \left(x^{\lambda f\prime} x^{\lambda f} \right)^{-1} c^{(2)}_x \right]'
    \left( \frac 1 {JT} z^{\lambda f\prime} M_{x^{\lambda f}} z^{\lambda f} \right)^{-1}
       \smallW_{JT}
      \nonumber \\ & \qquad \qquad
          \left( \frac 1 {JT} z^{\lambda f\prime} M_{x^{\lambda f}} z^{\lambda f} \right)^{-1}
         \left[  \frac 1 {\sqrt{JT}} \,z^{\lambda f \prime} \, M_{x^{\lambda f}} \varepsilon
                 + c^{(2)}_z
  - \left(z^{\lambda f\prime} x^{\lambda f} \right) \left(x^{\lambda f\prime} x^{\lambda f} \right)^{-1} c^{(2)}_x \right] \; ,
 \nonumber \\
    A_1 &=  \left( \frac 1 {JT} \, g' \, M_{x^{\lambda f}}  z^{\lambda f}  \right)
      \left( \frac 1 {JT} z^{\lambda f\prime} M_{x^{\lambda f}} z^{\lambda f} \right)^{-1}
     \smallW_{JT}
           \left( \frac 1 {JT} z^{\lambda f\prime} M_{x^{\lambda f}} z^{\lambda f} \right)^{-1}
  \nonumber \\ & \qquad \qquad \qquad \qquad \qquad \qquad
       \left[  \frac 1 {\sqrt{JT}} \,z^{\lambda f \prime} \, M_{x^{\lambda f}} \varepsilon
                 + c^{(2)}_z
  - \left(z^{\lambda f\prime} x^{\lambda f} \right) \left(x^{\lambda f\prime} x^{\lambda f} \right)^{-1} c^{(2)}_x \right]
      \; ,
 \nonumber \\
    A_2 &=  \left( \frac 1 {JT} \, g' \, M_{x^{\lambda f}}  z^{\lambda f}  \right)
      \left( \frac 1 {JT} z^{\lambda f\prime} M_{x^{\lambda f}} z^{\lambda f} \right)^{-1}
     \smallW_{JT}
           \left( \frac 1 {JT} z^{\lambda f\prime} M_{x^{\lambda f}} z^{\lambda f} \right)^{-1}  \left( \frac 1 {JT} \, z^{\lambda f \prime} \, M_{x^{\lambda f}}   g \right) \; ,
\end{align}
and the remainder term in the objective function satisfies
\begin{align}
    \sup_{\{\alpha: \, \sqrt{J} \|\alpha-\alpha^0\| < c \}}
       \frac{ \| r^{\rm obj}(\alpha) \| }
       { \left( 1 + \sqrt{JT} \| \alpha - \alpha^0 \| \right)^2}
         &= o_p(1) \; .
\end{align}
Under our assumptions one can show that
$\|A_1\|={\cal O}_p(1)$ and $\plim_{J,T \rightarrow \infty }A_2>0$.
Combining the expansion of the objective function with the results of $\sqrt{J}$-consistency of $\widehat \alpha$ we can thus conclude that
\begin{align}
   \sqrt{JT} \, \left( \widehat \alpha - \alpha^0 \right)
   &=  A_2^{-1} \, A_1 + o_p(1) \; .
   \label{lim_alpha}
\end{align}
Analogous to equation \eqref{betagammaAsym} for the first step, we can apply
the results in Moon and Weidner \cite*{MoonWeidner2015,MoonWeidner2015b} to the third step of
the LS-MD estimator to obtain
\begin{align}
   \sqrt{JT} (\widehat \beta-\beta^0) &= \left( \frac 1 {JT}  x^{\lambda f \prime}
                            x^{\lambda f} \right)^{-1}
        \left[  C^{(1)}\left(X_k,\Psi(\widehat \alpha) \right)
         +C^{(2)}\left(X_k,\Psi(\widehat \alpha) \right) \right]_{k=1,\ldots,K}
+ o_p(1)
   \nonumber \\
       &=  \left( \frac 1 {JT}  x^{\lambda f \prime} x^{\lambda f} \right)^{-1}
           \left[  \frac 1 {\sqrt{JT}} x^{\lambda f \prime} \, \varepsilon
                  -\left( \frac 1 {JT} x^{\lambda f \prime} \, g \right)
                     \sqrt{JT} \, (\widehat \alpha - \alpha^0)
                  +c^{(2)}_x \right] + o_p(1) \; .
   \label{lim_beta}
\end{align}
Here, the remainder term
$o_p(\sqrt{JT}  \| \widehat \alpha - \alpha^0 \|)$ is already absorbed into the $o_p(1)$
term, since \eqref{lim_alpha} already shows $\sqrt{JT}$-consistency of
$\widehat \alpha$.
Let $G_{JT}$ and  $\bigW_{JT}$
be the expressions in equation \eqref{DefGOmega} and \eqref{DefBigW}
before taking the probability limits, i.e.
$G = \plim_{J,T \rightarrow \infty} G_{JT}$
and $\bigW = \plim_{J,T \rightarrow \infty} \bigW_{JT}$.
One can show that
\begin{align}
   G_{JT} \smallW_{JT} G_{JT}'
     &=
   \frac 1 {JT} \, \left(g, x\right)' \, P_{x^{\lambda f}} \,  \left(g, x\right)
           + \left( \begin{array}{cc} A_2 & 0_{L\times K} \\ 0_{K\times L}
             & 0_{K\times K} \end{array} \right) .
\end{align}
Using this, one can rewrite equation
\eqref{lim_alpha} and \eqref{lim_beta} as follows
\begin{align}
   G_{JT} & \bigW_{JT} G_{JT}'
   \, \sqrt{JT} \left( \begin{array}{c} \widehat \alpha - \alpha^0 \\[2mm]
                                        \widehat \beta  - \beta^0 \end{array} \right)
  \nonumber \\
      &= \frac 1 {\sqrt{JT}}  \left(g, x\right)'
                 P_{x^{\lambda f}}   \varepsilon
             + \left( \begin{array}{cc} A_1
                  +   \left( \frac 1 {JT}  g'  x^{\lambda f} \right)
                    \left( \frac 1 {JT}  x^{\lambda f \prime} x^{\lambda f} \right)^{-1}
                       c^{(2)}_x
                \\[2mm] c^{(2)}_x  \end{array} \right)
         + o_p(1) \; ,
\end{align}
and therefore
\begin{align}
   \sqrt{JT} & \left( \begin{array}{c} \widehat \alpha - \alpha^0 \\[2mm]
                                        \widehat \beta  - \beta^0 \end{array} \right)
  \nonumber \\ &=   \left( G_{JT} \bigW_{JT} G_{JT}' \right)^{-1}
     G_{JT} \bigW_{JT}
           \left[ \frac 1 {\sqrt{JT}}
          \left(x^{\lambda f}, z^{\lambda f}\right)' \varepsilon \right]
      \nonumber \\
        & \quad  +  \left( G_{JT} \bigW_{JT} G_{JT}' \right)^{-1}  \left( \begin{array}{cc}
           A_3  c^{(2)}_z
       + \left[  \left( g'  x^{\lambda f} \right)
         - A_3  \left(z^{\lambda f\prime} x^{\lambda f} \right)   \right]
   \left(x^{\lambda f\prime} x^{\lambda f} \right)^{-1} c^{(2)}_x
                \\[2mm] c^{(2)}_x  \end{array} \right)
           + o_p(1)
      \nonumber \\
        &= \left( G \bigW G' \right)^{-1}
     G \bigW
           \left[ \frac 1 {\sqrt{JT}}
          \left(x^{\lambda f}, z^{\lambda f}\right)' \varepsilon
          + { c^{(2)}_x \choose c^{(2)}_z} \right]
        + o_p(1),
   \label{LimitFormula}
\end{align}
where $A_3=\left( \frac 1 {JT}  g'  M_{x^{\lambda f}}  z^{\lambda f}  \right)
      \left( \frac 1 {JT} z^{\lambda f\prime} M_{x^{\lambda f}} z^{\lambda f} \right)^{-1}
     \smallW_{JT}
           \left( \frac 1 {JT} z^{\lambda f\prime} M_{x^{\lambda f}} z^{\lambda f} \right)^{-1}$.
Having equation \eqref{LimitFormula}, all that is left to do is to derive the
asymptotic distribution of
$c^{(2)}_x$, $c^{(2)}_z$ and
$\frac 1 {\sqrt{JT}}
          \left(x^{\lambda f}, z^{\lambda f}\right)' \varepsilon$. This was
done in Moon and Weidner \cite*{MoonWeidner2015,MoonWeidner2015b} under the same assumptions that
we impose here. They show that
\begin{align}
   c^{(2)}_x &= - \kappa^{-1} \, b^{(x,1)} - \kappa \, b^{(x,2)} + o_p(1) \; , &
   c^{(2)}_z &= - \kappa^{-1} \, b^{(z,1)} - \kappa \, b^{(z,2)} + o_p(1) \; ,
\end{align}
and
\begin{align}
   \frac 1 {\sqrt{JT}}
          \left(x^{\lambda f}, z^{\lambda f}\right)' \varepsilon
      \; \; & \operatorname*{\longrightarrow}_d  \; \;
         {\cal N}\left[ - \kappa
           { b^{(x,0)} \choose b^{(z,0)}} , \, \Omega \right] \; .
\end{align}
Plugging this into \eqref{LimitFormula} gives the result on the limiting
distribution of $\widehat \alpha$ and $\widehat \beta$ which is stated in the theorem.
\end{proof}

\subsection{Consistency of Bias and Variance Estimators}

\begin{proof}[\bf Proof of Theorem \ref{th:biascorrection}]
   From Moon and Weidner \cite*{MoonWeidner2015,MoonWeidner2015b} we already
   know that under our assumptions we have $\widehat \Omega=\Omega+o_p(1)$,
   $\widehat b^{(x,i)} = b^{(x,i)} + o_p(1)$
   and $\widehat b^{(z,i)} = b^{(z,i)} + o_p(1)$, for $i=0,1,2$. They also show that
   $\| M_{\widehat \lambda} - M_{\lambda^0} \| = {\cal O}_p(J^{-1/2})$
   and  $\| M_{\widehat f} - M_{f^0} \| = {\cal O}_p(J^{-1/2})$, from which
   we can conclude that $\bigWest = \bigW+o_p(1)$.
   These results on $M_{\widehat \lambda}$ and $M_{\widehat f}$ together
   with $\sqrt{JT}$-consistency of $\widehat \alpha$
   and Assumption~\ref{ass:A5} are also sufficient to conclude
   $\widehat G=G+o_p(1)$. It follows that $\widehat B_i=B_i + o_p(1)$, for $i=0,1,2$.
\end{proof}

\section{Additional details on numerical verification of instrument
  relevance condition}

Here we present some additional details related to the numerical
verification of the instrument relevance condition, which was discussed in
Section 6.1 of the main text.

For the numerator of $\rho_F(\alpha,\beta)$ one finds
\begin{align*}
    &    \max_{\lambda \in \mathbbm{R}^{J \times R}}   \left[
          \Delta \xi_{\alpha,\beta}'
     \left( \mathbbm{1}_T \otimes P_{(\lambda,\lambda^0)}\right)
       \Delta \xi_{\alpha,\beta} \right]
    \nonumber \\
     &=  \Delta \xi_{\alpha,\beta}'
     \left( \mathbbm{1}_T \otimes P_{ \lambda^0 }\right)
       \Delta \xi_{\alpha,\beta}
       +  \max_{\lambda \in \mathbbm{R}^{J \times R}}  \left[
          \Delta \xi_{\alpha,\beta}'
     \left( \mathbbm{1}_T \otimes M_{\lambda^0} P_{\lambda} M_{\lambda^0}\right)
       \Delta \xi_{\alpha,\beta} \right]
  \nonumber \\
     &=  \Delta \xi_{\alpha,\beta}'
     \left( \mathbbm{1}_T \otimes P_{ \lambda^0 }\right)
       \Delta \xi_{\alpha,\beta}
       + \max_{\lambda \in \mathbbm{R}^{J \times R}}
       {\rm Tr}          \left[
          (\delta(\alpha)-\delta(\alpha^0)-\beta \cdot X)'
          M_{\lambda^0} P_\lambda M_{\lambda^0}
          (\delta(\alpha)-\delta(\alpha^0)-\beta \cdot X)   \right]       \nonumber \\
     &=  \Delta \xi_{\alpha,\beta}'
     \left( \mathbbm{1}_T \otimes P_{ \lambda^0 }\right)
       \Delta \xi_{\alpha,\beta}
       +  \sum_{r=1}^{R} \mu_r
          \left[
          (\delta(\alpha)-\delta(\alpha^0)-\beta \cdot X)' M_{\lambda^0}
          (\delta(\alpha)-\delta(\alpha^0)-\beta \cdot X)   \right] .
\end{align*}
In the first step we used
$P_{(\lambda,\lambda^0)} = P_{ \lambda^0 }
 +  M_{\lambda^0} P_{M_{\lambda^0} \lambda} M_{\lambda^0}$.
The optimal value of $\lambda$ in the second line
always satisfies $\lambda = M_{\lambda^0} \lambda$,
so we could write $P_{\lambda}$
instead of $P_{M_{\lambda^0} \lambda}$. In the second
step we plugged in the definition of $\Delta \xi_{\alpha,\beta}$.
In the final step we used the characterization of
the eigenvalues in terms
of a maximization problem,
and the fact that the non-zero eigenvalues of
the matrices $(\delta(\alpha)-\delta(\alpha^0)-\beta \cdot X)' M_{\lambda^0}
          (\delta(\alpha)-\delta(\alpha^0)-\beta \cdot X)$
and
$ M_{\lambda^0} (\delta(\alpha)-\delta(\alpha^0)-\beta \cdot X)
          (\delta(\alpha)-\delta(\alpha^0)-\beta \cdot X)' M_{\lambda^0}$
are identical.

Because of this, $\rho_{F}(\alpha,\beta)$ is equal to
\begin{align*}
      \rho_{\rm F}(\alpha,\beta)
       &= \frac{
           \Delta \xi_{\alpha,\beta}'
     \left( \mathbbm{1}_T \otimes P_{\lambda^0}\right)
       \Delta \xi_{\alpha,\beta}  }
     {   \Delta \xi_{\alpha,\beta}'  \Delta \xi_{\alpha,\beta}  }
   \nonumber \\
    & \qquad
    +  \frac{ \sum_{r=1}^{R} \mu_r
          \left[
          (\delta(\alpha)-\delta(\alpha^0)-\beta \cdot X)' M_{\lambda^0}
          (\delta(\alpha)-\delta(\alpha^0)-\beta \cdot X)   \right]
       }     {   \Delta \xi_{\alpha,\beta}'  \Delta \xi_{\alpha,\beta}  } .
\end{align*}
Thus, computation of $\rho_F$ only involves the numerical
calculation of the first $R$ eigenvalues $\mu_r$
of a $T \times T$ matrix, which can be done very quickly
even for relatively large values of $T$.

\end{appendix}


\begin{thebibliography}{}

\bibitem[\protect\astroncite{Ackerberg et~al.}{2007}]{ack_pakes}
Ackerberg, D., Benkard, L., Berry, S., and Pakes, A. (2007).
\newblock Econometric tools for analyzing market outcomes.
\newblock In Heckman, J. and Leamer, E., editors, {\em Handbook of
  Econometrics, Vol. 6A}. North-Holland.

\bibitem[\protect\astroncite{Ahn and Horenstein}{2013}]{AhnHorenstein2013}
Ahn, S.~C. and Horenstein, A.~R. (2013).
\newblock Eigenvalue ratio test for the number of factors.
\newblock {\em Econometrica}, 81(3):1203--1227.

\bibitem[\protect\astroncite{Ahn et~al.}{2001}]{AhnLeeSchmidt2001}
Ahn, S.~C., Lee, Y.~H., and Schmidt, P. (2001).
\newblock {GMM} estimation of linear panel data models with time-varying
  individual effects.
\newblock {\em Journal of Econometrics}, 101(2):219--255.

\bibitem[\protect\astroncite{Ahn et~al.}{2013}]{AhnLeeSchmidt2013}
Ahn, S.~C., Lee, Y.~H., and Schmidt, P. (2013).
\newblock Panel data models with multiple time-varying individual effects.
\newblock {\em Journal of Econometrics}, 174(1):1--14.

\bibitem[\protect\astroncite{Bai}{2009}]{Bai2009}
Bai, J. (2009).
\newblock {Panel data models with interactive fixed effects}.
\newblock {\em Econometrica}, 77(4):1229--1279.

\bibitem[\protect\astroncite{Bai and Ng}{2002}]{BaiNg2002}
Bai, J. and Ng, S. (2002).
\newblock Determining the number of factors in approximate factor models.
\newblock {\em Econometrica}, 70(1):191--221.

\bibitem[\protect\astroncite{Bai and Ng}{2006}]{BaiNg2006}
Bai, J. and Ng, S. (2006).
\newblock Confidence intervals for diffusion index forecasts and inference for
  factor-augmented regressions.
\newblock {\em Econometrica}, 74(4):1133--1150.

\bibitem[\protect\astroncite{Bajari et~al.}{2011}]{bajari_fox_kim_ryan}
Bajari, P., Fox, J., Kim, K., and Ryan, S. (2011).
\newblock A simple nonparametric estimator for the distribution of random
  coefficients.
\newblock {\em Quantitative Economics}, 2:381--418.

\bibitem[\protect\astroncite{Bajari et~al.}{2012}]{bajari_fox_kim_ryan2}
Bajari, P., Fox, J., Kim, K., and Ryan, S. (2012).
\newblock The random coefficients logit model is identified.
\newblock {\em Journal of Econometrics}, 166:204--212.

\bibitem[\protect\astroncite{Berry et~al.}{2013}]{BerryGandhiHaile2013}
Berry, S., Gandhi, A., and Haile, P. (2013).
\newblock {Connected Substitutes and Invertibility of Demand}.
\newblock {\em Econometrica}, 81:2087--2111.

\bibitem[\protect\astroncite{Berry et~al.}{1995}]{BerryLevinsohnPakes1995}
Berry, S., Levinsohn, J., and Pakes, A. (1995).
\newblock Automobile prices in market equilibrium.
\newblock {\em Econometrica}, 63(4):841--890.

\bibitem[\protect\astroncite{Berry et~al.}{2004}]{BerryLintonPakes2004}
Berry, S., Linton, O.~B., and Pakes, A. (2004).
\newblock Limit theorems for estimating the parameters of differentiated
  product demand systems.
\newblock {\em The Review of Economic Studies}, 71(3):613--654.

\bibitem[\protect\astroncite{Berry}{1994}]{Berry1994}
Berry, S.~T. (1994).
\newblock Estimating discrete-choice models of product differentiation.
\newblock {\em The RAND Journal of Economics}, 25(2):242--262.

\bibitem[\protect\astroncite{Berry and Haile}{2014}]{berry_haile_market}
Berry, S.~T. and Haile, P.~A. (2014).
\newblock Identification in differentiated products markets using market level
  data.
\newblock {\em Econometrica}, 82(5):1749--1797.

\bibitem[\protect\astroncite{Besanko and
  Doraszelski}{2004}]{BesankoDoraszelski2004}
Besanko, D. and Doraszelski, U. (2004).
\newblock Capacity dynamics and endogenous asymmetries in firm size.
\newblock {\em RAND Journal of Economics}, 35:23--49.

\bibitem[\protect\astroncite{Chernozhukov and
  Hansen}{2006}]{ChernozhukovHansen2006}
Chernozhukov, V. and Hansen, C. (2006).
\newblock {Instrumental quantile regression inference for structural and
  treatment effect models}.
\newblock {\em Journal of Econometrics}, 132(2):491--525.

\bibitem[\protect\astroncite{Chiappori and Komunjer}{2009}]{chiappori_komunjer}
Chiappori, P. and Komunjer, I. (2009).
\newblock On the nonparametric identification of multiple choice models.
\newblock Manuscript, Columbia University.

\bibitem[\protect\astroncite{Dube et~al.}{2012}]{DubeFoxSu2012}
Dube, J.~P., Fox, J., and Su, C. (2012).
\newblock {Improving the Numerical Performance of BLP Static and Dynamic
  Discrete Choice Random Coefficients Demand Estimation}.
\newblock {\em Econometrica}, 80:2231--2267.

\bibitem[\protect\astroncite{Esteban and Shum}{2007}]{esteban_shum1}
Esteban, S. and Shum, M. (2007).
\newblock {Durable Goods Oligopoly with Secondary Markets: the Case of
  Automobiles}.
\newblock {\em RAND Journal of Economics}, 38:332--354.

\bibitem[\protect\astroncite{Gandhi et~al.}{2010}]{GandhiKimPetrin2010}
Gandhi, A., Kim, K., and Petrin, A. (2010).
\newblock {Identification and Estimation in Discrete Choice Demand Models when
  Endogenous Variables Interact with the Error}.
\newblock Manuscript.

\bibitem[\protect\astroncite{Gandhi et~al.}{2013}]{GandhiLuShi2013}
Gandhi, A., Lu, Z., and Shi, X. (2013).
\newblock {Estimating Demand for Differentiated Products with Error in Market
  Shares}.
\newblock Working paper, University of Wisconsin.

\bibitem[\protect\astroncite{Hahn and Kuersteiner}{2002}]{Hahn:2002p717}
Hahn, J. and Kuersteiner, G. (2002).
\newblock Asymptotically unbiased inference for a dynamic panel model with
  fixed effects when both "n" and "{T}" are large.
\newblock {\em Econometrica}, 70(4):1639--1657.

\bibitem[\protect\astroncite{Hahn and Kuersteiner}{2004}]{Hahn:2004p878}
Hahn, J. and Kuersteiner, G. (2004).
\newblock Bias reduction for dynamic nonlinear panel models with fixed effects.
\newblock Manuscript.

\bibitem[\protect\astroncite{Hahn and Newey}{2004}]{Hahn:2004p882}
Hahn, J. and Newey, W. (2004).
\newblock Jackknife and analytical bias reduction for nonlinear panel models.
\newblock {\em Econometrica}, 72(4):1295--1319.

\bibitem[\protect\astroncite{Harding}{2007}]{Harding2007}
Harding, M. (2007).
\newblock Structural estimation of high-dimensional factor models.
\newblock Manuscript.

\bibitem[\protect\astroncite{Harding and Hausman}{2007}]{HardingHausman2007}
Harding, M. and Hausman, J. (2007).
\newblock {Using a Laplace Approximation to Estimate the Random Coefficients
  Logit Model By Nonlinear Least Squares}.
\newblock {\em International Economic Review}, 48:1311--1328.

\bibitem[\protect\astroncite{Hausman}{1997}]{hausmancereal}
Hausman, J. (1997).
\newblock Valuation of new goods under perfect and imperfect competition.
\newblock In {\em The Economics of New Goods}. University of Chicago Press.

\bibitem[\protect\astroncite{Holtz-Eakin
  et~al.}{1988}]{HoltzEakin-Newey-Rosen1988}
Holtz-Eakin, D., Newey, W., and Rosen, H.~S. (1988).
\newblock Estimating vector autoregressions with panel data.
\newblock {\em Econometrica}, 56(6):1371--95.

\bibitem[\protect\astroncite{Knittel and
  Metaxoglou}{2014}]{KnittelMetaxoglou2014}
Knittel, C. and Metaxoglou, K. (2014).
\newblock Estimation of random coefficient demand models: Two empiricists'
  perspective.
\newblock {\em Review of Economics and Statistics}, 96:34--59.

\bibitem[\protect\astroncite{Matzkin}{2013}]{Matzkin2013}
Matzkin, R. (2013).
\newblock Nonparametric identification in structural economic models.
\newblock In {\em Annual Review of Economics}, volume~5.

\bibitem[\protect\astroncite{Media Dynamics, Inc.}{1997}]{tv_dimensions}
Media Dynamics, Inc. (1997).
\newblock {\em {TV Dimensions}}.
\newblock {Media Dynamics, Inc.}
\newblock annual publication.

\bibitem[\protect\astroncite{Moon and Weidner}{2015a}]{MoonWeidner2015}
Moon, H. and Weidner, M. (2015a).
\newblock {Dynamic Linear Panel Regression Models with Interactive Fixed
  Effects}.
\newblock {\em Forthcoming in Econometric Theory}.

\bibitem[\protect\astroncite{Moon and Weidner}{2015b}]{MoonWeidner2015b}
Moon, H.~R. and Weidner, M. (2015b).
\newblock Linear regression for panel with unknown number of factors as
  interactive fixed effects.
\newblock {\em Econometrica}, 83(4):1543--1579.

\bibitem[\protect\astroncite{Nevo}{2001}]{nevo1}
Nevo, A. (2001).
\newblock Measuring market power in the ready-to-eat cereals industry.
\newblock {\em Econometrica}, 69:307--342.

\bibitem[\protect\astroncite{Neyman and Scott}{1948}]{NeymanScott1948}
Neyman, J. and Scott, E.~L. (1948).
\newblock Consistent estimates based on partially consistent observations.
\newblock {\em Econometrica}, 16(1):1--32.

\bibitem[\protect\astroncite{Onatski}{2010}]{Onatski2010}
Onatski, A. (2010).
\newblock Determining the number of factors from empirical distribution of
  eigenvalues.
\newblock {\em The Review of Economics and Statistics}, 92(4):1004--1016.

\bibitem[\protect\astroncite{Petrin}{2002}]{petrin}
Petrin, A. (2002).
\newblock Quantifying the benefits of new products: the case of the minivan.
\newblock {\em Journal of Political Economy}, 110:705--729.

\bibitem[\protect\astroncite{Villas-Boas and Winer}{1999}]{Villasboas1999}
Villas-Boas, J. and Winer, R. (1999).
\newblock Endogeneity in brand choice models.
\newblock {\em Management Science}, 45:1324--1338.

\end{thebibliography}
\end{document}